\RequirePackage{fix-cm}
\documentclass[twocolumn]{svjour3}          
\smartqed  
\usepackage{graphicx}

\usepackage{appendix}
\usepackage[dvipsnames]{xcolor}
\usepackage{hyperref}
\usepackage{bbold}
\usepackage{afterpage}
\usepackage{mathptmx}
\usepackage[ruled,linesnumbered]{algorithm2e}
\usepackage{amsmath}
\usepackage{hhline}
\usepackage{multirow}
\usepackage{natbib}
\begin{document}\sloppy

\title{Fisher Scoring for crossed factor Linear Mixed Models
}


\author{Thomas Maullin-Sapey        \and
        Thomas E. Nichols 
}


\institute{Thomas Maullin-Sapey \at
              Big Data Institute, Li Ka Shing Centre for Health Information and Discovery, Old Road Campus, Oxford, OX3 7LF, UK \\ \email{Thomas.Maullin-Sapey@bdi.ox.ac.uk}\\
              ORCID ID: 0000-0002-1890-330X
           \and
           Thomas E. Nichols \at
              Big Data Institute, Li Ka Shing Centre for Health Information and Discovery, Old Road Campus, Oxford, OX3 7LF, UK \\ \email{Thomas.Nichols@bdi.ox.ac.uk}\\
              ORCID ID: 0000-0002-4516-5103\\
}

\date{Received: date / Accepted: date}

\maketitle

\begin{abstract}

The analysis of longitudinal, heterogeneous or unbalanced clustered data is of primary importance to a wide range of applications. The Linear Mixed Model (LMM) is a popular and flexible extension of the linear model specifically designed for such purposes. Historically, a large proportion of material published on the LMM concerns the application of popular numerical optimization algorithms, such as Newton-Raphson, Fisher Scoring and Expectation Maximization to single-factor LMMs (i.e. LMMs that only contain one “factor” by which observations are grouped). However, in recent years, the focus of the LMM literature has moved towards the development of estimation and inference methods for more complex, multi-factored designs. In this paper, we present and derive new expressions for the extension of an algorithm classically used for single-factor LMM parameter estimation, Fisher Scoring, to multiple, crossed-factor designs. Through simulation and real data examples, we compare five variants of the Fisher Scoring algorithm with one another, as well as against a baseline established by the R package lmer, and find evidence of correctness and strong computational efficiency for four of the five proposed approaches. Additionally, we provide a new method for LMM Satterthwaite degrees of freedom estimation based on analytical results, which does not require iterative gradient estimation. Via simulation, we find that this approach produces estimates with both lower bias and lower variance than the existing methods.

\end{abstract}

\section{Introduction}
\label{intro}
\subsection{Background}\label{Intro1}

Since its conception in the seminal work of \citet{Laird1982}, the literature on Linear Mixed Model (LMM) estimation and inference has evolved rapidly. At present, many software packages exist which are capable of performing LMM estimation and inference for large and complex LMMs in an incredibly quick and memory-efficient manner. For some packages, this exceptional speed and efficiency arises from simplifying model assumptions, whilst for others, complex mathematical operations such as sparse matrix methodology and sweep operators are utilized to improve performance (\citet{Wolfinger1994}, \citet{Bates:2015pls}). 

However, situations exist in which the current methodology cannot be applied due to practical implementation considerations. One such example is given in the field of medical imaging, in which mass-univariate analyses, where hundreds of thousands of models must be estimated concurrently, are standard practice. To efficiently perform a mass-univariate analysis within a practical time-frame, the use of vectorized computation, which exploits the repetitive nature of simplistic operations to streamline calculation, must be employed (\citet{Smith2018}, \citet{Li2019}). Unfortunately, many existing LMM tools utilize complex operations, for which vectorized support does not currently exist. As a result, alternative methodology, using more conceptually simplistic mathematical operations for which vectorized support exists, is required.

The complex nature of LMM computation has partly arisen from the gradual expansion of the definition of ``Linear Mixed Model". Previously, the term ``Linear Mixed Model" was primarily used to refer to an extension of the linear regression model which contains random effects grouped by a single random factor. Examples of this definition can be seen in \citet{Laird1982}, in which ``Linear Mixed Model" refers only to ``single-factor" longitudinal models, and in \citet{Bates1988}, where more complex, multi-factor models are described as an \textit{``extension"} of the ``Linear Mixed Model". Consequently, throughout the late 1970s and 1980s, one of the main focuses of the LMM literature was to provide parameter estimation methods, such as Fisher Scoring, Newton-Raphson and Expectation Maximization, for the single-factor LMM (e.g. \citet{Dempster1977}, \citet{Jennrich1986} and \citet{Laird1987}). By exploiting structural features of the single-factor model, implementation of these methods required only conceptually simplistic mathematical operations. For instance, the Fisher Scoring algorithm proposed by \citet{Jennrich1986} relies only upon vector addition and matrix multiplication, inversion, and reshaping operations. 

In recent times, usage of the term ``Linear Mixed Models" has grown substantially to include models which contain random effects grouped by multiple random factors. Examples of this more general definition are found in \cite{Bates2009:book} and \cite{Tibaldi2007}. For this more general ``multi-factor" LMM definition, models can be described as exhibiting either a hierarchical factor structure (i.e. factor groupings are nested inside one another) or crossed factor structure (i.e. factor groupings are not nested inside one another). For instance, a study involving students grouped by the factors ``school" and ``class" contains hierarchical factors (as every class belongs to a specific school). In contrast, a study that involves observations grouped by "subject" and "location," where subjects change location between recorded measurements, contains crossed factors (as each subject is not generally associated with a specific location or vice versa). In either case, due to the more complex nature of the multi-factor LMM, the computational approaches used in the single-factor setting can not be directly applied to the multi-factor model. For this reason, parameter estimation of the multi-factor LMM has often been viewed as a much more ``difficult" problem than its single-factor counterpart (see, for example, the discussion in chapters $2$ and $8$ of \citet{west2014linear}).

Several authors have proposed and implemented methods for multi-factor LMM estimation. However, such methods typically require conceptually complex mathematical operations, which are not naturally amenable to vectorized computation, or restrictive simplifying model assumptions, which prevent some designs from being estimated. For example, the popular R package lmer performs LMM estimation via minimization of a penalized least squares cost function based on a variant of Henderson's mixed model equations (\citet{Henderson1959}, \citet{Bates:2015pls}). However, sparse matrix methods are required to achieve fast evaluation of the cost function, and advanced numerical approximation methods are needed for optimization (e.g. the Bound Optimization BY Quadratic Approximation, BOBYQA, algorithm, \citet{Powell2009}). The commonly used SAS and SPSS packages, PROC-MIXED and MIXED, employ a Newton-Raphson algorithm proposed by \citet*{Wolfinger1994}. In this approach, though, derivatives and Hessians must be computed using the sweep operator, W-transformation, and Cholesky decomposition updating methodology (\citet{SASMIXED}, \cite{SPSSMIXED}). The Hierarchical Linear Models (HLM) package takes an alternative option and restricts its inputs to only LMMs which contain hierarchical factors (\cite{Raudenbush2002}); although, the Hierarchical Cross-classified Models (HCM) sub-module does make allowances for specific use-case crossed LMMs. To perform parameter estimation, HLM employs a range of different methods, each tailored to a particular model design. As a result, there are many models for which HLM does not provide support.

Methods have also been proposed for evaluating the score vector (the derivative of the log-likelihood function) and the Fisher Information matrix (the expected Hessian of the negative log-likelihood function) required to perform Fisher Scoring for the multi-factor LMM. For example, alongside the Newton-Raphson approach adopted by PROC-MIXED and MIXED, \citet{Wolfinger1994} also describe a Fisher Scoring algorithm. However, evaluation of the expressions they provide requires use of the sweep operator, W-transformation, and Cholesky decomposition updating methodology; operations for which widespread vectorized support do not yet exist. More recently, expressions for score vectors and Fisher Information matrices were provided by \citet{zhu2018essential}. However, the approach \citet{zhu2018essential} adopt to derive these expressions produces an algorithm that requires independent computation for each variance parameter in the LMM. This algorithm's serialized nature results in substantial overheads in terms of computation time, thus limiting the method's utility in practical situations where time-efficiency is a crucial consideration.

To our knowledge, no approach has yet been provided for multi-factor LMM parameter estimation, which utilizes only simplistic, universally available operations which are naturally amenable to vectorized computation. In this work we revisit a Fisher Scoring approach suggested for the single-factor LMM, described in \citet{Demidenko:2013mfx}, extending it to the multi-factor setting, with the intention of revisiting the motivating example of the mass-univariate model in later work.

The novel contribution of this work is to provide new derivations and closed-form expressions for the score vector and Fisher Information matrix of the multi-factor LMM. We show how these expressions can be employed for LMM parameter estimation via the Fisher Scoring algorithm and can be further adapted for constrained optimization of the random effects covariance matrix. Additionally, we demonstrate that such closed-form expressions have use beyond the scope of the Fisher Scoring algorithm. An example of this is provided in the setting of mixed model inference, where the degrees of freedom of the approximate T-statistic are not known and must be estimated (\citet{verbeke2001:book}). We show how our methods may be combined with the Satterthwaite-based method for approximating degrees of freedom for the LMM by using an approach based on the work of \cite{Kuznetsova2017}.

In this paper, we first propose five variants of the Fisher Scoring algorithm. Following this, we provide a discussion of initial starting values for the algorithm, approaches to ensuring positive definiteness of the random effects covariance matrix, and methods for improving the algorithm's computational efficiency during implementation. Detail on constrained optimization, allowing for structural assumptions to be placed on the random effects covariance matrix, is then provided. Proceeding this, new expressions for Satterthwaite estimation of the degrees of freedom of the approximate T-statistic for the multi-factor LMM are given. Finally, we verify the correctness of the proposed algorithms and degrees of freedom estimates via simulation and real data examples, benchmarking the performance against that of the R package lmer.

\subsection{Preliminaries}\label{prelim}

In this section, we introduce preliminary statistical concepts and notation which will be employed throughout the remainder of this work. In Section \ref{modelDesc}, a formal definition of the LMM is provided and in Section \ref{NotationSec}, additional notation is detailed.

\subsubsection{The model}\label{modelDesc}
Both the single-factor and multi-factor LMM, for $n$ observations, take the following form:
\begin{equation}\label{LMMdef}
Y=X\beta + Zb + \epsilon \newline
\end{equation}
\begin{equation}\label{LMMdef2}
\epsilon \sim N(0,\sigma^2 I_n), \hspace{0.3cm} b \sim N(0, \sigma^2 D),
\end{equation}

\noindent
where the known quantities are:
\begin{itemize}
    \item[$\bullet$] $Y$: the $(n \times 1)$ vector of responses.
    \item[$\bullet$] $X$: the $(n \times p)$ fixed effects design matrix.
    \item[$\bullet$] $Z$: the $(n \times q)$ random effects design matrix.
\end{itemize}
\noindent
And the unknown parameters are:

\begin{itemize}
    \item[$\bullet$] $\beta$: the $(p \times 1)$ fixed effects parameter vector.
    \item[$\bullet$] $\sigma^2$: the scalar fixed effects variance.
    \item[$\bullet$] $D$: the $(q \times q)$ random effects covariance matrix.
\end{itemize}
\noindent
From (\ref{LMMdef}) and (\ref{LMMdef2}) the marginal distribution of the response vector, $Y$ is:
\begin{equation}\nonumber
Y \sim N(X\beta,\sigma^2 (I_n + ZDZ')).
\end{equation}
The log-likelihood for the LMM specified by (\ref{LMMdef}) and (\ref{LMMdef2}) is derived from the marginal distribution of $Y$. Dropping constant terms, this log-likelihood is given by:
\begin{equation}\label{llh}
l(\theta)=-\frac{1}{2}\bigg\{ n\log(\sigma^2)+\sigma^{-2}e'V^{-1}e+\log|V|\bigg\},
\end{equation}
where $\theta$ is shorthand for all the parameters $(\beta, \sigma^2, D)$, $V=I_n+ZDZ'$ and $e=Y-X\beta$. Throughout the main body of this work, we shall consider parameter estimation performed via Maximum Likelihood (ML) estimation of (\ref{llh}). However, we note that the approaches we describe for ML estimation can also be easily adapted to use a Restricted Maximum Likelihood (ReML) criterion. Further detail on ReML estimation is provided in Appendix \ref{remlApp}.

The distinction between the multi-factor and single-factor LMM lies in the specification of the random effects in the model. Random effects are often described in terms of factors, categorical variables that group the random effects, and levels, individual instances of such a categorical variable. At this point, we highlight that the term ``factor", in this work, refers only to categorical variables which group random effects and does not refer to groupings of fixed effects. We denote the total number of factors in the model as $r$ and denote the $k^{th}$ factor in the model as $f_k$ for $k \in \{1,...,r\}$. For a given factor $f_k$, $l_k$ will be used to denote the number of levels possessed by $f_k$, and $q_k$ the number of random effects which $f_k$ groups. The single-factor LMM corresponds to the case $r=1$, while the multi-factor setting corresponds to the case $r>1$.

An example of how this notation may be used in practice is given as follows. Suppose an LMM contains observations that are grouped according to ‘subject’ (i.e. the participant whose observation was recorded) and ‘location’ (i.e. the place the observation was recorded). Further, suppose that the LMM includes a random intercept and random slope which each model subject-specific behaviour, and a random intercept which models location-specific behaviour. Two factors are present in this design: the factor $f_1$ is ‘subject’ and the factor $f_2$ is ‘location’. Therefore, $r=2$. The number of subjects is $l_1$ and the number of locations is $l_2$. The number of covariates grouped by the first factor, $q_1$, is $2$ (i.e. the random intercept and the random slope) and the number grouped by the second factor, $q_2$, is $1$ (i.e. the random intercept). 

The values of $r$, $\{q_k\}_{k\in\{1,\hdots,r\}}$ and $\{l_k\}_{k\in\{1,\hdots,r\}}$ determine the structure of the random effects design matrix, $Z$, and random effects covariance matrix, $D$. To specify $Z$ formally is notationally cumbersome and of little relevance to the aims of this work. For this reason, the reader is referred to the work of \citet{Bates:2015pls} for further detail on the construction of $Z$. Here it will be sufficient to note that, under the assumption that it’s columns are appropriately ordered, $Z$ will be comprised of $r$ horizontally concatenated blocks. The $k^{th}$ block of $Z$, denoted $Z_{(k)}$, has dimension $(n \times l_kq_k)$ and describes the random effects which are grouped by the $k^{th}$ factor. Additionally, each block, $Z_{(k)}$, can be further partitioned column-wise into $l_k$ blocks of dimension $(n \times q_k)$. The $j^{th}$ block of $Z_{(k)}$, denoted $Z_{(k,j)}$, corresponds to the random effects which belong to the $j^{th}$ level of the $k^{th}$ factor. In summary,
\begin{equation}\nonumber
\begin{aligned}[b]
    & Z = [Z_{(1)},Z_{(2)},\hdots Z_{(r)}], \\
    & Z_{(k)} = [Z_{(k,1)},Z_{(k,2)},\hdots Z_{(k,l_k)}] \hspace{0.5cm}\text{(for }k\in \{1,\hdots, r\}\text{)}
    \end{aligned}
\end{equation}

An important property of the matrix $Z$ is that for arbitrary $k\leq r$, $Z_{(k,i)}$ and $Z_{(k,j)}$ satisfy the below condition for all $i\neq j$:
\begin{equation}\label{Zprop}
\begin{aligned}[b]
    & \text{ If row }w\text{ of }Z_{(k,i)}\text{ contains any non-zero values }\\
    & \implies \text{then row }w\text{ of }Z_{(k,j)}\text{ contains only zero values. }\\
    \end{aligned}
\end{equation}
 
A consequence of the above property is that for any arbitrary factor $f_k$, the rows of $Z_{(k)}$ can be permuted in order to obtain a block diagonal matrix. As, for the single-factor LMM, $Z\equiv Z_{(1)}$, it follows that the observations of the single-factor LMM can be arranged such that $Z$ is block diagonal. This feature of the single-factor LMM simplifies the derivation of the Fisher Information matrix and score vector required for Fisher Scoring. However, this simplification cannot be generalized to the multi-factor LMM. In general, it is not true that the rows of $Z$ can be permuted in such a way that the resultant matrix is block-diagonal. As emphasized in Section \ref{Intro1}, due to this, many of the results derived in the single-factor LMM have not been generalized to the multi-factor setting.

To describe the random effects covariance matrix, $D$, it is assumed that factors are independent from one another and that for each factor, factor $f_k$, there is a $(q_k \times q_k)$ unknown covariance matrix, $D_k$, representing the ``within-factor" covariance for the random effects grouped by $f_k$. The random effects covariance matrix, $D$, appearing in (\ref{LMMdef2}), can now be given as:
\begin{equation}\nonumber
    D= \bigoplus_{k=1}^r (I_{l_k} \otimes D_k),
\end{equation}
where $\oplus$ represents the direct sum, and $\otimes$ the Kronecker product. Note that, whilst $D$ is large in dimension (having dimension $(q \times q)$ where $q=\sum_i q_il_i$), $D$ contains $q_u =\sum_i\frac{1}{2}q_i(q_i+1)$ unique elements. Typically, it is true that $q_u<<q^2$. As a result, the task of describing the random effects covariance matrix $D$ reduces in practice to specifying only a small number of parameters.

\subsubsection{Notation}\label{NotationSec}

In this section, we introduce notation which will be used throughout the remainder of this work. The conventional hat operator notation is adopted to denote estimators resulting from likelihood maximisation procedures. For example, $\beta$ represents the true (unknown) fixed effects parameter vector for the LMM, whilst the maximum likelihood estimate of $\beta$ is denoted $\hat{\beta}$. The conventional subscript notation, $A_{[X,Y]}$, is used to denote the sub-matrix of matrix $A$, composed of all elements of $A$ with row indices $x \in X$ and column indices $y \in Y$. The replacement of $X$ or $Y$ with a colon, $:$, represents all rows or all columns of $A$, respectively. If a scalar, $x$ or $y$, replaces $X$ or $Y$, this represents the elements with row indices $x=X$ or column indices $y=Y$, respectively. The notation $(k)$ may also replace $X$ and $Y$ where $(k)$ represents the indices of the columns of $Z$ which correspond to factor $f_k$. Similarly, $X$ and $Y$ may be substituted for the ordered pair $(k,j)$ where $(k,j)$ represents the indices of the columns of $Z$ which correspond to level $j$ of factor $f_k$. We highlight again our earlier notation, $Z_{(k,j)}$, which, due to its frequent occurrence acts as a shorthand for $Z_{[:,(k,j)]}$, i.e. the columns of $Z$ corresponding to level $j$ of factor $f_k$. 

Finally, we shall also adopt the notations `vec', `vech', $N_k$, $K_{m,n}$, $\mathcal{D}_k$ and $\mathcal{L}_k$ as used in \citet{Magnus1980}, defined as follows:

\begin{itemize}
    \item[$\bullet$] `vec' represents the mathematical vectorizarion operator which transforms an arbitrary $(k \times k)$ matrix, $A$, to a $(k^2 \times 1)$ column vector, vec$(A)$, composed of the columns of $A$ stacked into a column vector. (Note: this is not to be confused with the concept of computational vectorization discussed in Section \ref{Intro1}).
    
    \item[$\bullet$] `vech' represents the half-vectorization operator which transforms an arbitrary square matrix, $A$, of dimension $(k \times k)$ to a $(k(k+1)/2 \times 1)$ column vector, vech$(A)$, composed by stacking the elements of $A$ which fall on and below the diagonal into a column vector.
    
    \item[$\bullet$] $N_k$ is defined as the unique matrix of dimension $(k^2 \times k^2)$ which implements symmetrization for any arbitrary square matrix $A$ of dimension $(k \times k)$ in vectorized form. I.e. $N_k$ satisfies the following relation:
    $$N_k\text{vec}(A)=\text{vec}(A+A')/2.$$
    
    \item[$\bullet$] $K_{m,n}$ is the unique ``Commutation" matrix of dimension $(mn \times mn)$, which permutes, for any arbitrary matrix $A$ of dimension $(m \times n)$, the vectorization of $A$ to obtain the vectorization of the transpose of $A$. I.e. $K_{m,n}$ satisfies the following relation:
    $$\text{vec}(A)=K_{m,n}\text{vec}(A').$$ 
    
    \item[$\bullet$] $\mathcal{D}_k$ is the unique ``Duplication" matrix of dimension $(k^2 \times k(k+1)/2)$, which maps the half-vectorization of any arbitrary symmetric matrix $A$ of dimension $(k \times k)$ to it's vectorization. I.e. $\mathcal{D}_k$ satisfies the following relation:
    $$\text{vec}(A)=\mathcal{D}_k\text{vech}(A).$$ 
    
    \item[$\bullet$] $\mathcal{L}_k$ is the unique $1-0$ ``Elimination" matrix of dimension $(k(k+1)/2 \times k^2)$, which maps the vectorization of any arbitrary lower triangular matrix $A$ of dimension $(k \times k)$ to it's half-vectorization. I.e. $\mathcal{L}_k$ satisfies the following relation:
    $$\text{vech}(A)=\mathcal{L}_k\text{vec}(A).$$ 
\end{itemize}
To help track the notational conventions employed in this work, an index of notation is provided in the Supplementary Material, Section $S10$.

\section{Methods}
\subsection{Fisher Scoring Algorithms}\label{AlgoSection}

In this section, we employ the Fisher Scoring algorithm for ML estimation of the parameters $(\beta, \sigma^2, D)$. The Fisher Scoring algorithm update rule takes the following form:

\begin{equation}\label{FS}
\theta_{s+1} = \theta_{s} + \alpha_s\mathcal{I}(\theta_{s})^{-1}\frac{dl(\theta_s)}{d\theta}_,
\end{equation}
\\
where $\theta_s$ is the vector of parameter estimates given at iteration $s$, $\alpha_s$ is a scalar step size, the score vector of $\theta_s$, $\frac{dl(\theta_s)}{d\theta}$, is the derivative of the log-likelihood with respect to $\theta$ evaluated at $\theta=\theta_s$, and $\mathcal{I}(\theta_{s})$ is the Fisher Information matrix of $\theta_s$;
    $$\mathcal{I}(\theta_{s})=E\bigg[\bigg(\frac{dl(\theta)}{d \theta}\bigg)\bigg(\frac{dl(\theta)}{d \theta}\bigg)'\bigg|\theta=\theta_s\bigg].$$

A more general formulation of Fisher Scoring, which allows for low-rank Fisher Information matrices, is given by \citet{rao1972:psuedo}:

\begin{equation}\label{FFS}
\theta_{s+1} = \theta_{s} + \alpha_s\mathcal{I}(\theta_{s})^{+}\frac{dl(\theta_s)}{d\theta}_,
\end{equation}
\\
where superscript plus, $^+$, is the Moore-Penrose (or ``Pseudo") inverse. For notational brevity, when discussing algorithms of the form (\ref{FS}) and (\ref{FFS}) in the following sections, the subscript $s$, representing iteration number, will be suppressed unless its inclusion is necessary for clarity.

For the LMM, several different representations of the parameters of interest, $(\beta, \sigma^2,D)$, can be used for numerical optimization and result in different Fisher Scoring iteration schemes. In this section, we consider the following three representations for $\theta$:
    \begin{equation}\nonumber
        \theta^h = \begin{bmatrix}\beta\\ \sigma^2\\ \text{vech}(D_1)\\\vdots\\ \text{vech}(D_r)\\
        \end{bmatrix},
        \theta^f = \begin{bmatrix}\beta\\ \sigma^2\\ \text{vec}(D_1)\\\vdots\\ \text{vec}(D_r)\\
        \end{bmatrix},
        \theta^c = \begin{bmatrix}\beta\\ \sigma^2\\ \text{vech}(\Lambda_1)\\\vdots\\ \text{vech}(\Lambda_r)\\
        \end{bmatrix},
    \end{equation}
where $\Lambda_k$ represents the lower triangular Cholesky factor of $D_k$, such that $D_k=\Lambda_k\Lambda_k'$. We will refer to the representations ($\theta^h$, $\theta^f$ and $\theta^c$) as the ``half", ``full" and ``Cholesky" representations of $(\beta, \sigma^2, D)$, respectively. In a slight abuse of notation, the function $l$ will be allowed to take any representation of $\theta$ as input, with the interpretation unchanged (i.e. $l(\theta^f)=l(\theta^h)=l(\theta^c)$). For example, if the full representation is being used, the log-likelihood will be denoted $l(\theta^f)$, but if the half representation is being used, the log likelihood will be denoted $l(\theta^h)$.

In the following sections, the score vectors and Fisher Information matrices required to perform five variants of Fisher Scoring for the multi-factor LMM will be stated with proofs provided in Appendices \ref{derivAppendix} and \ref{InfoMatAppendix}. For notational convenience, we denote the sub-matrix of the Fisher Information matrix of $\theta^{h}$ with row indices corresponding to parameter vector $a$ and column indices corresponding to parameter vector $b$ as $\mathcal{I}^{h}_{a,b}$. In other words, $\mathcal{I}^{h}_{a,b}$ is the sub-matrix of $\mathcal{I}(\theta^{h})$, defined by:
    
    $$\mathcal{I}^{h}_{a,b}=E\bigg[\bigg(\frac{dl(\theta^{h})}{d a}\bigg)\bigg(\frac{dl(\theta^{h})}{d b}\bigg)'\bigg|\theta=\theta_s\bigg].$$

For further simplification of notation, when $a$ and $b$ are equal, the second subscript will be dropped and the matrix $\mathcal{I}^{h}_{a,b}=\mathcal{I}^{h}_{a,a}$ will be denoted simply as $\mathcal{I}^{h}_{a}$. Analogous notation is used for the full and Choleksy representations.

\subsubsection{Fisher Scoring}\label{FSsection}

The first variant of Fisher Scoring we provide uses the ``half"-representation for $(\beta, \sigma^2, D)$, $\theta^h$, and is based on the standard form of Fisher Scoring given by (\ref{FS}). This may be considered the most natural approach for a Fisher Scoring algorithm as $\theta^h$ is an unmodified vector of the unique parameters of the LMM and (\ref{FS}) is the standard update rule. For this approach the elements of the score vector are:
\begin{equation}\label{FSderiv1}
    \frac{d l(\theta^h)}{d \beta} = \sigma^{-2}X'V^{-1}e,
\end{equation}
\begin{equation}\label{FSderiv2}
    \frac{d l(\theta^h)}{d \sigma^2} = -\frac{n}{2}\sigma^{-2}+\frac{1}{2}\sigma^{-4}e'V^{-1}e.    
\end{equation}
For $k \in \{1,\hdots, r\}$:
\begin{equation}\label{FSderiv3}
    \frac{d l(\theta^h)}{d \text{vech}(D_k)} = \frac{1}{2}\mathcal{D}_{q_k}'\text{vec}\bigg( \sum_{j=1}^{l_k}Z'_{(k,j)}V^{-1}\bigg(\frac{ee'}{\sigma^2}-V\bigg)V^{-1}Z_{(k,j)}\bigg).    
\end{equation}
and the entries of the Fisher Information matrix are given by:
\begin{equation}\label{FSFI1}
    \mathcal{I}^h_{\beta}=\sigma^{-2}X'V^{-1}X,\hspace{0.4cm}
    \mathcal{I}^h_{\beta,\sigma^2}=\mathbf{0}_{p,1},\hspace{0.4cm}\mathcal{I}^h_{\sigma^2}=\frac{n}{2}\sigma^{-4}.
\end{equation}

\noindent
For $k \in \{1,\hdots, r\}$:
\begin{equation}\label{FSFI2}
\begin{aligned}[b]
 & \mathcal{I}^h_{\beta,\text{vech}(D_k)}=\mathbf{0}_{p,q_k(q_k+1)/2},\\
    & \mathcal{I}^h_{\sigma^2,\text{vech}(D_k)}= \frac{1}{2}\sigma^{-2}\text{vec}'\bigg(\small\sum_{j=1}^{l_k} Z_{(k,j)}'V^{-1}Z_{(k,j)}\bigg)\mathcal{D}_{q_k}. \\
\end{aligned}
\end{equation}

\noindent
For $k_1,k_2\in\{1,\hdots,r\}$:
\begin{equation}\label{FSFI3}
  \begin{aligned}[b]
    &\mathcal{I}^h_{\text{vech}(D_{k_1}),\text{vech}(D_{k_2})}=\\ &\frac{1}{2}\mathcal{D}'_{q_{k_1}}\sum_{j=1}^{l_{k_2}}\sum_{i=1}^{l_{k_1}}(Z'_{(k_1,i)}V^{-1}Z_{(k_2,j)}\otimes Z'_{(k_1,i)}V^{-1}Z_{(k_2,j)})\mathcal{D}_{q_{k_2}}.
    \end{aligned}
\end{equation}
where $\mathbf{0}_{n,k}$ denotes the $(n \times k)$ dimensional matrix of zeros. Due to its to natural approach to Fisher Scoring, this algorithm is referred to as FS in the remainder of this text. Pseudocode for the FS algorithm is given in Algorithm \ref{FSalgorithm}. Discussion of the initial estimates used in the algorithm is deferred to Section \ref{initval}.
\vspace*{0.2cm}

\IncMargin{1em}

\begin{algorithm}\label{FSalgorithm}
\SetAlgoLined
\footnotesize{
\caption{Fisher Scoring (FS)}
\vspace{0.3cm}
 Assign $\theta^h$ to an initial estimate using (\ref{beta0}) and (\ref{D0})
 \BlankLine
 \While{current $l(\theta^h)$ and previous $l(\theta^h)$ differ by more than a predefined tolerance}{

\BlankLine
 Calculate $\frac{dl(\theta^h)}{d\theta^h}$ using (\ref{FSderiv1})-(\ref{FSderiv3}).
\BlankLine
 Calculate $\mathcal{I}(\theta^h)$ using (\ref{FSFI1})-(\ref{FSFI3})
\BlankLine
 Assign $\theta^h=\theta^h + \alpha\mathcal{I}(\theta^h)^{-1}\frac{dl(\theta^h)}{d\theta^h}$
 \BlankLine
 Assign $\alpha=\frac{\alpha}{2}$ if $l(\theta^h)$ has decreased in value.
}
\BlankLine
}
\end{algorithm}\DecMargin{1em}

\subsubsection{Full Fisher Scoring}\label{FFSsection}

The second variant of Fisher Scoring considered in this work uses the ``full", $\theta^f$, representation of the model parameters, and shall therefore be referred to as ``Full Fisher Scoring" (FFS). In this approach, for each factor, $f_k$, the elements of vec$(D_k)$ are to be treated as distinct from one another with numerical optimization for $D_k$ performed over the space of all $(q_k \times q_k)$ matrices. This approach differs from the FS method proposed in the previous section, in which optimization was performed on the space of only those $(q_k \times q_k)$ matrices that are symmetric. This optimization procedure is realized by treating symmetric matrix elements of $D_k$ as distinct and, for a given element, using the partial derivative with respect to the element during optimization instead of the total derivative with respect to both the element and its symmetric counterpart. This change is reflected by denoting the elements of the score vector which correspond to vec$(D_k)$ using a partial derivative operator, $\partial$, rather than the total derivative operator, $d$. The primary motivation for the inclusion of the FFS approach is that it serves as a basis for which the constrained covariance approaches of Section \ref{covstruct} can be built upon. However, it should be noted that, as it does not require the construction or use of duplication matrices, FFS also provides simplified expressions and potential improvement in terms of computation speed. As a result, FFS is of some theoretical and practical interest and is detailed in full here. 

An immediate obstacle to this approach is that the Fisher Information matrix of $\theta^f$ is rank-deficient, and, therefore, cannot be inverted. Intuitively, this is to be expected, as repeated entries in $\theta^f$ result in repeated rows in $\mathcal{I}(\theta^f)$. Mathematically, this can be seen by noting that $\mathcal{I}^f_{\text{vec}(D_{k})}$ can be expressed as a product containing the matrix $N_{q_k}$ (defined in Section \ref{NotationSec}), which is low-rank by construction (see Appendix \ref{FullAppendix}). Consequently, the Fisher Scoring update rule for $\theta^f$ must be based on the Pseudo-inverse formulation of Fisher Scoring given in (\ref{FFS}).

As the derivatives of the log-likelihood with respect to $\beta$ and $\sigma^2$ do not depend upon the parameterisation of $D$, both FFS and FS employ the same expressions for the elements of the score vector which correspond to $\beta$ and $\sigma^2$, given by (\ref{FSderiv1}) and (\ref{FSderiv2}) respectively. The score vector for $\{$vec($D_k)\}_{k\in \{1,\dots r\}}$ used by FFS is given by:
\begin{equation}\label{FFSderiv1}
    \frac{\partial l(\theta^f)}{\partial \text{vec}(D_k)} = \frac{1}{2}\text{vec}\bigg( \sum_{j=1}^{l_k}Z'_{(k,j)}V^{-1}\bigg(\frac{ee'}{\sigma^2}-V\bigg)V^{-1}Z_{(k,j)}\bigg).    
\end{equation}
The entries of the Fisher Information matrix of $\theta^f$, based on the derivatives given in (\ref{FSderiv1}), (\ref{FSderiv2}) and (\ref{FFSderiv1}), are given by:
\begin{equation}\nonumber
    \mathcal{I}^f_{\beta}=\mathcal{I}^h_{\beta}, \hspace{0.7cm}
    \mathcal{I}^f_{\beta,\sigma^2}=\mathcal{I}^h_{\beta,\sigma^2},\hspace{0.7cm}\mathcal{I}^f_{\sigma^2}=\mathcal{I}^h_{\sigma^2}.
\end{equation}
\noindent
For $k \in \{1,\hdots, r\}$:
\begin{equation}\label{FFSFI1}\begin{aligned}[b]
& \mathcal{I}^f_{\beta,\text{vec}(D_k)}=\mathbf{0}_{p,q_k^2},\\
& \mathcal{I}^f_{\sigma^2,\text{vec}(D_k)}= \frac{1}{2}\sigma^{-2}\text{vec}'\bigg(\small\sum_{j=1}^{l_k} Z_{(k,j)}'V^{-1}Z_{(k,j)}\bigg). \\
\end{aligned}
\end{equation}
\noindent
For $k_1,k_2\in\{1,\hdots,r\}$:
\begin{equation}\label{FFSFI2}
  \begin{aligned}[b]
    &\mathcal{I}^f_{\text{vec}(D_{k_1}),\text{vec}(D_{k_2})}=\\ &\frac{1}{2}\sum_{j=1}^{l_{k_2}}\sum_{i=1}^{l_{k_1}}(Z'_{(k_1,i)}V^{-1}Z_{(k_2,j)}\otimes Z'_{(k_1,i)}V^{-1}Z_{(k_2,j)})N_{q_k}.
    \end{aligned}
\end{equation}
Derivations for the above can be found in Appendices \ref{derivAppendix} and \ref{InfoMatAppendix}. In addition, in Appendix \ref{FullAppendix}, we show that the Full Fisher Scoring algorithm can also be expressed in the form:
\begin{equation}\label{FFS2}
\theta^f_{s+1} = \theta^f_{s} + \alpha_s F(\theta_{s}^f)^{-1}\frac{\partial l(\theta_s^f)}{\partial\theta}_,
\end{equation}
where, unlike in (\ref{FS}) and (\ref{FFS}), $F(\theta^f)$ is not the Fisher Information matrix. Rather, $F(\theta^f)$ is a matrix of equal dimensions to $\mathcal{I}(\theta^f)$ with all of its elements equal to those of $\mathcal{I}(\theta^f)$, apart from those which were specified by (\ref{FFSFI2}), which instead are, for $k_1,k_2\in\{1,\hdots,r\}$:
\begin{equation}\label{FFSF2}
  \begin{aligned}[b]
    &F_{\text{vec}(D_{k_1}),\text{vec}(D_{k_2})}=\\ &\frac{1}{2}\sum_{j=1}^{l_{k_2}}\sum_{i=1}^{l_{k_1}}(Z'_{(k_1,i)}V^{-1}Z_{(k_2,j)}\otimes Z'_{(k_1,i)}V^{-1}Z_{(k_2,j)})
    \end{aligned},
\end{equation}
where the same subscript notation has been adopted to index $F(\theta^f)$ as was adopted for $\mathcal{I}(\theta^f)$.

The above approach, which utilizes the matrix $F(\theta^f)$, rather than $\mathcal{I}(\theta^f)$, is adopted to match the formulation of the Fisher Scoring algorithm given in Section 2.11 of \citet{Demidenko:2013mfx} for the single-factor LMM. Pseudocode for the FFS algorithm using the representation of the update rule given by (\ref{FFS2}), is provided by Algorithm \ref{FFSalgorithm}.\newline
\vspace{0.3cm}
\IncMargin{1em}

\begin{algorithm}\label{FFSalgorithm}
\SetAlgoLined
\footnotesize{
\caption{Full Fisher Scoring (FFS)}
\vspace{0.3cm}
 Assign $\theta^f$ to an initial estimate using (\ref{beta0}) and (\ref{D0})
 \BlankLine
 \While{current $l(\theta^f)$ and previous $l(\theta^f)$ differ by more than a predefined tolerance}{

\BlankLine
 Calculate $\frac{\partial l(\theta^f)}{\partial\theta^f}$ using (\ref{FSderiv1}),(\ref{FSderiv2}) and (\ref{FFSderiv1}).
\BlankLine
 Calculate $F(\theta^f)$ using (\ref{FFSFI1}) and (\ref{FFSF2}).
\BlankLine
 Assign $\theta^f=\theta^f + \alpha F(\theta^f)^{-1}\frac{\partial l(\theta^f)}{\partial\theta^f}$
 \BlankLine
 Assign $\alpha=\frac{\alpha}{2}$ if $l(\theta^f)$ has decreased in value.
}
\BlankLine
}
\end{algorithm}\DecMargin{1em}

\subsubsection{Simplified Fisher Scoring}\label{SFSsection}

The third Fisher Scoring algorithm proposed in this work relies on the half-representation of the parameters $(\beta,\sigma^2,D)$ and takes an approach, similar to that of coordinate ascent, which is commonly adopted in the single-factor setting (c.f. \citet{Demidenko:2013mfx}). In this approach, instead of performing a single update step upon the entire vector $\theta^h$ in the form of $(\ref{FS})$, updates for $\beta, \sigma^2$ and $\{D_k\}_{k \in \{1,\hdots, r\}}$ are performed individually in turn. For $\beta$ and $\sigma^2$ each iteration uses the standard Generalized Least Squares (GLS) estimators given by:
\begin{equation}\label{betaUpdate}
    \beta_{s+1} = (X'V_s^{-1}X)^{-1}X'V_s^{-1}Y,
\end{equation}
\begin{equation}\label{sigma2Update}
    \sigma^2_{s+1} = \frac{e_{s+1}'V^{-1}_{s}e_{s+1}}{n}. 
\end{equation}

To update the random effects covariance matrix, $D_k$, for each factor, $f_k$, individual Fisher Scoring updates are applied to vech$(D_k)$. These updates are performed using the block of the Fisher Information matrix corresponding to vech$(D_k)$, given by (\ref{FSFI3}), and take the following form for $k \in \{1,\hdots,r\}$:
\begin{equation}\label{vechDUpdate}
    \text{vech}(D_{k,s+1})=\text{vech}(D_{k,s})+\alpha_s\big(\mathcal{I}^{h}_{\text{vech}(D_{k,s})}\big)^{-1}\frac{dl(\theta^h_s)}{d\text{vech}(D_{k,s})}.
\end{equation}
In line with the naming convention used in \cite{Demidenko:2013mfx}, this method shall be referred to as ``Simplified" Fisher Scoring (SFS). This is due to the relative simplicity, both in terms of notational and computational complexity, of the updates (\ref{betaUpdate})-(\ref{vechDUpdate}) used in the SFS algorithm in comparison to those used in the FS algorithm of Section \ref{FSsection}, given by (\ref{FSderiv1})-(\ref{FSFI3}). Pseudocode for the SFS algorithm is given by Algorithm \ref{SFSalgorithm}.\newline
\vspace{0.3cm}

\IncMargin{1em}

\begin{algorithm}\label{SFSalgorithm}
\SetAlgoLined
\footnotesize{
\caption{Simplified Fisher Scoring (SFS)}
\vspace{0.3cm}
 Assign $\theta^f$ to an initial estimate using (\ref{beta0}) and (\ref{D0})
 \BlankLine
 \While{current $l(\theta^h)$ and previous $l(\theta^h)$ differ by more than a predefined tolerance}{

\BlankLine
 Update $\beta$ using (\ref{betaUpdate})
\BlankLine
 Update $\sigma^2$ using (\ref{sigma2Update})
\BlankLine
\For{$k \in \{1,...r\}$}{
\BlankLine
Update vech($D_k$) using (\ref{vechDUpdate})
}
 \BlankLine
 Assign $\alpha=\frac{\alpha}{2}$ if $l(\theta^h)$ has decreased in value.
 \BlankLine
}}
\BlankLine
\end{algorithm}\DecMargin{1em}

\subsubsection{Full Simplified Fisher Scoring}\label{FSFSsection}

The Full Simplified Fisher Scoring algorithm (FSFS) combines the ``Full" and ``Simplified" approaches described in Sections \ref{FFSsection} and \ref{SFSsection}. In the FSFS algorithm individual updates are applied to $\beta$ and $\sigma^2$ using (\ref{betaUpdate}) and (\ref{sigma2Update}) respectively and to $\{\text{vec}(D_k)\}_{k\in \{1,\hdots, r\}}$ using a Fisher Scoring update step, based on the matrix $F_{\text{vec}(D_{k})}$ given by (\ref{FFSF2}). The update rule for $\{\text{vec}(D_k)\}_{k\in \{1,\hdots, r\}}$ takes the following form:
\begin{equation}\label{vecDUpdate}
    \text{vec}(D_{k,s+1})=\text{vec}(D_{k,s})+\alpha_s F^{-1}_{\text{vec}(D_{k,s})}\frac{\partial l(\theta^f_s)}{\partial\text{vec}(D_{k,s})}.
\end{equation}

In Appendix \ref{FullAppendix}, we justify the above update rule by demonstrating that it is equivalent to using an update rule of the form (\ref{FFS}), given by:
\begin{equation}\nonumber
    \text{vec}(D_{k,s+1})=\text{vec}(D_{k,s})+\alpha_s\big(\mathcal{I}^{f}_{\text{vec}(D_{k,s})}\big)^{+}\frac{\partial l(\theta^f_s)}{\partial\text{vec}(D_{k,s})}.
\end{equation}
Pseudocode for the FSFS algorithm is given by Algorithm \ref{FSFSalgorithm}.\newline
\vspace{0.1cm}
\IncMargin{1em}
\begin{algorithm}\label{FSFSalgorithm}
\SetAlgoLined
\footnotesize{
\caption{Full Simplified Fisher Scoring (FSFS)}
\vspace{0.3cm}
 Assign $\theta^f$ to an initial estimate using (\ref{beta0}) and (\ref{D0})
 \BlankLine
 \While{current $l(\theta^f)$ and previous $l(\theta^f)$ differ by more than a predefined tolerance}{

\BlankLine
 Update $\beta$ using (\ref{betaUpdate})
\BlankLine
 Update $\sigma^2$ using (\ref{sigma2Update})
\BlankLine
\For{$k \in \{1,...r\}$}{
\BlankLine
Update vec($D_k$) using (\ref{vecDUpdate})
\BlankLine
}
 Assign $\alpha=\frac{\alpha}{2}$ if $l(\theta^f)$ has decreased in value.
 \BlankLine}}
\BlankLine
\end{algorithm}
\DecMargin{1em}

\subsubsection{Cholesky Simplified Fisher Scoring}\label{CSFSsection}

The final variant of the Fisher Scoring algorithm we consider is based on the ``simplified" approach described in Section \ref{SFSsection} and uses the Cholesky parameterisation of $(\beta,\sigma^2,D)$, $\theta^c$. This approach follows directly from the below application of the chain rule of differentiation for vector-valued functions,

\begin{equation}\nonumber
    \frac{dl(\theta^c)}{d\text{vech}(\Lambda_k)}=\frac{\partial\text{vech}(D_k)}{\partial\text{vech}(\Lambda_k)}\frac{\partial l(\theta^c)}{\partial \text{vech}(D_k)}.
\end{equation}

An expression for the derivative which appears second in the above product was given by (\ref{FSderiv3}). It therefore follows that in order to evaluate the score vector of $\text{vech}(\Lambda_k)$ (i.e. the derivative of $l$ with respect to $\text{vech}(\Lambda_k)$), only an expression for the first term of the above product is required. This term can be evaluated to $\mathcal{L}_{q_k}(\Lambda_k' \otimes I_{q_k})(I_{q_k^2}+K_{q_k})\mathcal{D}_{q_k}$, proof of which is provided in Appendix \ref{CholAppendix}.

Through similar arguments to those used to prove Corollaries \ref{covdldDkdBetaCor}-\ref{covdldDk1Dk2Cor} of Appendix \ref{InfoMatAppendix}, it can be shown that the Fisher Information matrix for $\theta^c$ is given by:
\begin{equation}\nonumber
    \mathcal{I}^c_{\beta}=\mathcal{I}^h_{\beta},\hspace{0.4cm}
    \mathcal{I}^c_{\beta,\sigma^2}=\mathcal{I}^h_{\beta,\sigma^2},\hspace{0.4cm}
    \mathcal{I}^c_{\sigma^2}=\mathcal{I}^h_{\sigma^2}.
\end{equation}

\noindent
For $k \in \{1,\hdots, r\}$:
\begin{equation}\label{CFSFI1}\begin{aligned}[b]
& \mathcal{I}^c_{\beta,\text{vech}(\Lambda_k)}=\mathbf{0}_{p,q_k(q_k+1)/2}, \\
& \mathcal{I}^c_{\sigma^2,\text{vech}(\Lambda_k)}= \mathcal{I}^h_{\sigma^2,\text{vech}(D_k)}\bigg(\frac{\partial\text{vech}(D_k)}{\partial\text{vech}(\Lambda_k)}\bigg)'.\\
\end{aligned}
\end{equation}

\noindent
For $k_1,k_2\in\{1,\hdots,r\}$:
\begin{equation}\label{CFSFI2}
  \begin{aligned}[b]
    &\mathcal{I}^c_{\text{vech}(\Lambda_{k_1}),\text{vech}(\Lambda_{k_2})}=\\ &\bigg(\frac{\partial \text{vech}(D_{k_1})}{\partial \text{vech}(\Lambda_{k_1})}\bigg)\mathcal{I}^h_{\text{vech}(D_{k_1}),\text{vech}(D_{k_2})}\bigg(\frac{\partial \text{vech}(D_{k_2})}{\partial \text{vech}(\Lambda_{k_2})}\bigg)'.
    \end{aligned}
\end{equation}

From the above, it can be seen that a non-``simplified" Cholesky-based variant of the Fisher Scoring algorithm, akin to the FS and FFS algorithms described in Sections \ref{FSsection} and \ref{FFSsection}, may be constructed. This would involve constructing the Fisher Information matrix of $\theta^c$, $\mathcal{I}(\theta^c)$, using (\ref{CFSFI1}) and (\ref{CFSFI2}), and employing a Fisher Scoring procedure similar to that specified by Algorithms \ref{FSalgorithm} and \ref{FFSalgorithm}. Whilst this approach is feasible, in terms of implementation, preliminary tests have indicated that the performance of this approach, in terms of computation time, is significantly worse than the previously proposed algorithms. For this reason, we only consider the ``simplified" version of the Cholesky Fisher Scoring algorithm, analogous to the ``simplified" approaches described in Sections \ref{SFSsection} and \ref{FSFSsection}, is considered here. The Cholesky Simplified Fisher Scoring (CSFS) algorithm adopts
(\ref{betaUpdate}) and (\ref{sigma2Update}) as update rules for $\beta$ and $\sigma^2$, respectively, and the following update rule for $\{\text{vech}(\Lambda_k)\}_{k \in \{1,...,r\}}$.
\begin{equation}\label{vechLambdaUpdate}\small{
    \text{vech}(\Lambda_{k,s+1})=\text{vech}(\Lambda_{k,s})+\alpha_s\big(\mathcal{I}^{c}_{\text{vech}(\Lambda_{k,s})}\big)^{-1}\frac{d l(\theta^c_s)}{d \text{vech}(\Lambda_{k,s})}}.
\end{equation}
Pseudocode summarizing the CSFS approach is given by Algorithm \ref{CSFSalgorithm}.
\vspace{0.3cm}
\IncMargin{1em}
\begin{algorithm}\label{CSFSalgorithm}
\SetAlgoLined
\footnotesize{
\caption{Cholesky Simplified Fisher Scoring (CSFS)}
\vspace{0.3cm}
 Assign $\theta^c$ to an initial estimate using (\ref{beta0}) and (\ref{D0})
 \BlankLine
 \While{current $l(\theta^c)$ and previous $l(\theta^c)$ differ by more than a predefined tolerance}{

\BlankLine
 Update $\beta$ using (\ref{betaUpdate})
\BlankLine
 Update $\sigma^2$ using (\ref{sigma2Update})
\BlankLine
\For{$k \in \{1,...r\}$}{
\BlankLine
Update vec($\Lambda_k$) using (\ref{vechLambdaUpdate})
\BlankLine
}
 Assign $\alpha=\frac{\alpha}{2}$ if $l(\theta^c)$ has decreased in value.
 \BlankLine}}
\BlankLine
\end{algorithm}
\DecMargin{1em}

\subsection{Initial values}\label{initval}

Choosing which initial values of $\beta$, $\sigma^2$ and $D$ will be used as starting points for optimization is an important consideration for the Fisher Scoring algorithm. Denoting these initial values as $\beta_0$, $\sigma^2_0$ and $D_0$, respectively, this work follows the recommendations of \citet{Demidenko:2013mfx} and evaluates $\beta_0$ and $\sigma^2_0$ using the OLS estimates given by,
\begin{equation}\label{beta0}
    \beta_0 = (X'X)^{-1}X'Y, \hspace{0.6cm}\sigma^2_0 = \frac{e_0'e_0}{n}.
\end{equation}
where $e_0$ is defined as $e_0=Y-X\beta_0$. For the initial estimate of $\{D_k\}_{k\in\{1,\dots ,r\}}$, an approach similar to that suggested in \cite{Demidenko:2013mfx} is also adopted, which substitutes $V$ for $I_n$ in the update rule for vec($D_k$), equation (\ref{vecDUpdate}). The resulting initial estimate for $\{D_k\}_{k\in\{1,\dots ,r\}}$ is given by
\begin{equation}\label{D0}
  \begin{aligned}[b]
    & \text{vec}(D_{k,0}) =& \bigg(\sum_{j=1}^{l_k}Z_{(k,j)}'Z_{(k,j)} \otimes Z_{(k,j)}'Z_{(k,j)}\bigg)^{-1} \times &  \\              & &  \text{vec}\bigg(\sum_{j=1}^{l_k}Z_{(k,j)}'\bigg(\frac{e_0e_0'}{\sigma_0^2}-I_n\bigg)Z_{(k,j)}\bigg) &
    \end{aligned}.
\end{equation}

\subsection{Ensuring non-negative definiteness}

In order to ensure numerical optimization produces a valid covariance matrix when performing LMM parameter estimation, optimization must be constrained to ensure that for each factor, $f_k$, $D_k$ is non-negative definite. This is a well-documented obstacle for numerical approaches to LMM parameter estimation for which many solutions have been suggested. A common approach, utilized in \citet{Bates:2015pls} for example, is to reparameterize optimization in terms of the Cholesky factor of the covariance matrix, $\Lambda_k$, as oppose to the covariance matrix itself, $D_k$. This approach is adopted by the CSFS method, described in Section \ref{CSFSsection}. However, the methods described in Sections \ref{FSsection}-\ref{FSFSsection} require an alternative approach as optimization does not account for the constraint that $D_k$ must be non-negative definite. For the FS, FFS, SFS and FSFS algorithms, we use an approach described by \citet{Demidenko:2013mfx}, for the single factor LMM. Denoting the eigendecomposition of $\{D_k\}_{k \in \{1,\hdots,r\}}$ as $D_k=U_k\Sigma_{k}U_k'$, we project $D_k$ onto the space of all $(q_k \times q_k)$ non-negative definite matrices by, following each update of the covariance matrix, $D$, replacing $\{D_k\}_{k \in \{1,..r\}}$ with $\{D_{+,k}\}_{k \in \{1,..r\}}$,
\begin{equation}\nonumber
    D_{+,k}=U_k\Sigma_{+,k}U_k',    
\end{equation}
where $\Sigma_{+,k}=\max(\Sigma_k,0)$ with `max' representing the element-wise maximum. For each factor, $f_k$, this ensures $D_k$, and as a result $D$, is non-negative definite.

\subsection{Computational efficiency}\label{compeff}

This section provides discussion on the computational efficiency of evaluating the Fisher Information matrices and score vectors of Sections \ref{FSsection}-\ref{CSFSsection}. This discussion is, in large part, motivated by the mass-univariate
setting (c.f. Section \ref{Intro1}), in which not one, but rather hundreds of thousands of models must be estimated concurrently. As a result, this discussion prioritizes both time efficiency and memory consumption concerns and, further, operates under the assumption that sparse matrix methodology, such as that employed by the R package `lmer', cannot be employed. This assumption has been made primarily as a result of the current lack of support for computationally vectorized sparse matrix operations. 

A primary concern, for both memory consumption and time efficiency, stems from the fact that many of the matrices used in the evaluation of the score vectors and Fisher Information matrices possess dimensions which scale with $n$, the number of observations. For example, $V$ has dimensions $(n \times n)$ and is frequently inverted in the FSFS algorithm. In practice, it is not uncommon for studies to have $n$ ranging into the thousands. As such, inverting $V$ directly may not be computationally feasible. To address this issue, we define the ``product forms" as;
\begin{equation}\nonumber
    P=X'X,\hspace{0.14cm} Q=X'Y,\hspace{0.14cm} R=X'Z,\hspace{0.14cm} S=Y'Y,\hspace{0.14cm} T=Y'Z,\hspace{0.14cm} U=Z'Z.
\end{equation}
Working with the product forms is preferable to using the original matrices $X,Y$ and $Z$, as the dimensions of the product forms do not scale with $n$. Instead, the dimensions of the product forms depend only on the number of fixed effects, $p$, and the second dimension of the random effects design matrix, $q$. As an example, consider the expressions below, which appear frequently in (\ref{FFSderiv1}) and (\ref{FFSF2}) and have been reformulated in terms of the product forms:
\begin{equation}\nonumber 
\begin{aligned}[b]
    & Z_{(k_1,i)}'V^{-1}Z_{(k_2,j)} = (U - UD(I_q+DU)^{-1}U)_{[(k_1,i),(k_2,j)]},\\
     & Z_{(k,j)}'V^{-1}e = ((I_q - UD)(I_q+DU)^{-1}(T'-R'\beta))_{[(k,j),:]}.\\
    \end{aligned}
\end{equation}

For computational purposes, the right-hand side of the above expressions are much more convenient than the left-hand side. In order to evaluate the left-hand side in both cases, an $(n \times n)$ inversion of the matrix $V$ must be performed. In contrast, the right-hand side is expressible purely in terms of the product forms and $(\beta, \sigma^2, D)$, with the only inversion required being that of the $(q \times q)$ matrix $(I_q+DU)$. These examples can be generalized further. In fact, all of the previous expressions (\ref{FSderiv1})-(\ref{D0}) can be rewritten in terms of only the product forms and $(\beta, \sigma^2, D)$. This observation is important as it implies that an algorithm for mixed model parameter estimation may begin by taking the matrices $X, Y$ and $Z$ as inputs, but discard them entirely once the product forms have been constructed. As a result, both computation time and memory consumption no longer scale with $n$.

Another potential source of concern regarding computation speed arises from noting that the algorithms we have presented contain many summations of the following two forms:
\begin{equation}\label{form1}
\sum_{i=1}^{c_0} A_iB_i'  \hspace{0.2cm}\text{ and } \hspace{0.2cm}
\sum_{i=1}^{c_1}\sum_{i=1}^{c_2} G_{i,j} \otimes H_{i,j}.
\end{equation}
where matrices $\{A_i\}$ and $\{B_i\}$ are of dimension $(m_1 \times m_2)$, and matrices $\{G_{i,j}\}$ and $\{H_{i,j}\}$ are of dimension $(n_1 \times n_2)$. We denote the matrices formed from vertical concatenation of the $\{A_i\}$ and $\{B_i\}$ matrices as $A$ and $B$, respectively, and $G$ and $H$ the matrices formed from block-wise concatenation of the $\{G_{i,j}\}$ and $\{H_{i,j}\}$, respectively. Instances of such summations can be found, for example, in equations (\ref{FSderiv3}) and (\ref{FSFI3}). 

From a computational standpoint, summations of the forms shown in (\ref{form1}) are typically realized by ‘for’ loops. This cumulative approach to computation can cause a potential issue for LMM computation since typically the number of summands corresponds to the number of levels, $l_k$, of some factor, $f_k$. In particular applications of the LMM, such as repeated measures and longitudinal studies, some factors may possess large quantities of levels. As this means `for' loops of this kind could hinder computation and result in slow performance, we provide alternative methods for calculating summations of the forms shown in (\ref{form1}).

For the former summation shown in (\ref{form1}), we utilize the ``generalized vectorisation", or ``vec$_m$" operator, defined by \citet{Turkington2013} as the operator which performs the below mapping for a horizontally partitioned matrix $M$:
\begin{equation}\nonumber
M = \begin{bmatrix} M_{1} & M_{2} & \hdots M_{c_0} \end{bmatrix} \rightarrow \text{vec}_{m}(M) = \begin{bmatrix} M_{1} \\ M_{2} \\ \vdots \\ M_{c_0}\end{bmatrix},
\end{equation}
where the partitions $\{M_i\}$ are evenly sized and contain $m$ columns. Using the definition of the generalized vectorisation operator, the former summation in (\ref{form1}) can be reformulated as:
\begin{equation}\nonumber
    \sum_{i=1}^l A_iB_i'=\text{vec}_{m_2}(A')'\text{vec}_{m_2}(B').
\end{equation}

The right-hand side of the expression above is of practical utility as the `vec$_m$' operator can be implemented efficiently. The `vec$_m$' operation can be performed, for instance, using matrix resizing operations such as the `reshape' operators commonly available in many programming languages such as MATLAB and Python. The computational costs associated with this approach are significantly lesser than those experienced when evaluating the summation directly using `for' loops.

For the latter summation in (\ref{form1}), we first define the notation $\tilde{M}$ to denote the transformation below, for the block-wise partitioned matrix $M$:

\begin{equation}\label{tildeTransform} M = \begin{bmatrix} M_{1,1} & M_{1,2} & ... & M_{1,c_2} \\ M_{2,1} & M_{2,2} & ... & M_{2,c_2} \\ \vdots & \vdots & \ddots & \vdots \\ M_{c_1,1} & M_{c_1,2} & ... & M_{c_1,c_2} \end{bmatrix} \rightarrow \tilde{M} = \begin{bmatrix} \text{vec}(M_{1,1})' \\ \text{vec}(M_{1,2})' \\ \vdots \\ \text{vec}(M_{1,c_2})' \\ \text{vec}(M_{2,1})' \\ \vdots \\ \text{vec}(M_{c_1,c_2})'\end{bmatrix}.
\end{equation}

Using a modification of Lemma 3.1 (i) of \citet{Neudecker1983}, we obtain the below identity:
\begin{equation}\label{NeudMod}
    \text{vec}\bigg(\sum_{i,j} G_{i,j}\otimes H_{i,j}\bigg) = (I_{n_2} \otimes K_{n_1,n_2} \otimes I_{n_1})\text{vec}(\tilde{H}'\tilde{G}),
\end{equation}
where $K_{n_1,n_2}$ is the $(n_1n_2\times n_1n_2)$ Commutation matrix (c.f. Section \ref{NotationSec}). The matrices $\tilde{H}$ and $\tilde{G}$ can be obtained from $H$ and $G$ using resizing operators with little computation time. The matrix $(I_{n_2} \otimes K_{n_1,n_2} \otimes I_{n_1})$ can be calculated via simple means and is a permutation matrix depending only on the dimensions of $H$ and $G$. As a result, this matrix can be calculated once and stored as a vector. Therefore, the matrix multiplication in the above expression does not need to be performed directly, but instead can be evaluated by permuting the elements of vec$(\tilde{H}'\tilde{G})$ according to the permutation vector representing $(I_{n_2} \otimes K_{n_1,n_2} \otimes I_{n_1})$. To summarize, evaluation of expressions of the latter form shown in (\ref{form1}) can be performed by using only reshape operations, a matrix multiplication, and a permutation. This method of evaluation provides notably faster computation time than the corresponding `for' loop evaluated over all values of $i$ and $j$.  

The above method, combined with the product form approach, was used to obtain the results of Sections \ref{SimMeth}-\ref{RealDatRes}.

\subsection{Constrained covariance structure}\label{covstruct}

In many applications involving the LMM, it is often desirable to use a constrained parameterisation for $D_k$. Examples include compound symmetry, first-order auto-regression and a Toeplitz structure. A more comprehensive list of commonly employed covariance structures in LMM analyses can be found, for example, in \citet{Wolfinger1996}. In this section, we describe how the Fisher Scoring algorithms of the previous sections can be adjusted to model dependence between covariance elements.

When a constraint is placed on the covariance matrix $D_k$, it is assumed that the elements of $D_k$ can be defined as continuous, differentiable functions of some smaller parameter vector, vecu$(D_k)$. Colloquially, vecu$(D_k)$ may be thought of as the vector of ``unique" parameters required to specify the constrained parameterization of $D_k$. To perform constrained optimization, we adopt a similar approach to the previous sections and define the constrained representation of $\theta$ as $\theta^{con}=[\beta', \sigma^2, \text{vecu}(D_1)',...\text{vecu}(D_k)']'$. Denoting the Jacobian matrix $\partial$vec$(D_k)/\partial$vecu$(D_k)$ as $\mathcal{C}_k$, the score vector and Fisher Information matrix of $\theta^{con}$ can be constructed as follows; for $k \in \{1,\hdots, r\}$,
\begin{equation}\label{derivCk}
\begin{aligned}
& \frac{dl(\theta^{con})}{d\text{vecu}(D_k)}=\mathcal{C}_k\frac{\partial l(\theta^{con})}{\partial\text{vec}(D_k)}, \\
& \mathcal{I}^{con}_{\beta,\text{vecu}(\tilde{D}_k)}=\mathbf{0}_{p,\tilde{q}_k}, \hspace{0.5cm} \mathcal{I}^{con}_{\sigma^2,\text{vecu}(\tilde{D}_k)}= \mathcal{I}^f_{\sigma^2,\text{vec}(D_k)}\mathcal{C}_k', \\
\end{aligned}
\end{equation}
\noindent
and, for $k_1,k_2\in\{1,\hdots,r\}$,
\begin{equation}\nonumber
    \mathcal{I}^{con}_{\text{vecu}(\tilde{D}_{k_1}),\text{vecu}(\tilde{D}_{k_2})}= \mathcal{C}_{k_1}\mathcal{I}^f_{\text{vec}(D_{k_1}),\text{vec}(D_{k_2})}\mathcal{C}_{k_2}',
\end{equation}
where $\tilde{q}_k$ is the length of vecu$(D_k)$ and $\mathcal{I}^f$ is the ``full" Fisher Information matrix defined in Section \ref{FFSsection}. The above expressions can be derived trivially using the definition of the Fisher Information matrix and the chain rule. In the remainder of this work, the matrix $\mathcal{C}_k$ is referred to as a ``constraint matrix" due to the fact it ``imposes" constraints during optimization. A fuller discussion of constraint matrices, alongside examples, is provided in Appendix \ref{conApp}. 

We note here that this approach can be extended to situations in which all covariance parameters can be expressed in terms of one set of parameters, $\rho_D$, common to all $\{\text{vec}(D_k)\}_{k \in \{1,\hdots, r\}}$. In such situations, a constraint matrix, $\mathcal{C}$, may be defined as the Jacobian matrix $\partial[\text{vec}(D_1)',...\text{vec}(D_1)']'/\partial\rho_D$ and Fisher Information matrices and score vectors may be derived in a similar manner to the above. Models requiring this type of commonality between factors are rare in the LMM literature, since the covariance matrices $\{D_k\}_{k \in \{1,...r\}}$ are typically defined independently of one another. However, one such example is provided by the ACE model employed for twin studies, in which the elements of $v(D)$ can all be expressed in terms of three variance components; $\rho_D=[\sigma^2_a, \sigma^2_c, \sigma^2_e]'$. Further information on the ACE model is presented in Section \ref{AceExample}.

\subsection{Degrees of freedom estimation}\label{swsection}

Several methods exist for drawing inference on the fixed effects parameter vector, $\beta$. Often, research questions can be expressed as hypotheses of the below form:
\begin{equation}\nonumber
    H_{0}: L\beta = 0, \hspace{0.5cm} H_{1}: L\beta \neq 0,
\end{equation}
where $L$ is a fixed and known $(1 \times p)$-sized contrast matrix specifying a hypothesis, or prior belief, about linear relationships between the elements of $\beta$, upon which an inference is to be made. In the setting of the LMM, a commonly employed statistic for testing hypotheses of this form is the approximate T-statistic, given by:
\begin{equation}\nonumber
    T = \frac{L\hat{\beta}}{\sqrt{\hat{\sigma}^2L(X'\hat{V}^{-1}X)^{-1}L'}},
\end{equation}
where $\hat{V}=I_n+Z\hat{D}Z'$. As noted in \cite{Dempster1981}, using an estimate of $D$ in this manner results in an underestimation of the true variability in $\hat{\beta}$. For this reason, the relation $T \sim t_{v}$ is not exact and is treated only as an approximating distribution (i.e. an ``approximate T statistic'') where the degrees of freedom, $v$, must be estimated empirically. A standard method for estimating the degrees of freedom, $v$, utilizes the Welch-Satterthwaite equation, originally described in \cite{Satterthwaite1946} and \cite{Welch1947}, given by: 
\begin{equation}\label{swdf} v(\hat{\eta})=\frac{2(S^2(\hat{\eta}))^2}{\text{Var}(S^2(\hat{\eta}))},
\end{equation}
where $\hat{\eta}$ represents an estimate of the variance parameters $\eta=(\sigma^2,D_1,\hdots D_r)$ and $S^2(\hat{\eta})$ is given by:
\begin{equation}\nonumber
S^2(\hat{\eta})=\hat{\sigma}^2L(X'\hat{V}^{-1}X)^{-1}L'.
\end{equation}
Typically, as the variance estimates obtained under ML estimation are biased downwards, ReML estimation is employed to obtain the estimate of $\eta$ employed in the above expression. If ML were used to estimate $\eta$ instead, the degrees of freedom, $v(\hat{\eta})$, would be underestimated and, consequently, conservative p-values and a reduction in statistical power for resulting hypothesis tests would be observed.

To obtain an approximation for $v(\hat{\eta})$, a second order Taylor expansion is applied to the unknown variance on the denominator of (\ref{swdf}):
\begin{equation}\label{swdf2}\text{Var}(S^2(\hat{\eta})) \approx \bigg(\frac{d S^2(\hat{\eta})}{d \hat{\eta}}\bigg)'\text{Var}(\hat{\eta})\bigg(\frac{d S^2(\hat{\eta})}{d \hat{\eta}}\bigg).
\end{equation}

This approach is well-documented, and has been adopted, most notably, by the R package lmerTest, first presented in \cite{Kuznetsova2017}. To obtain the derivative of $S^2$ and asymptotic variance covariance matrix, lmerTest utilizes numerical estimation methods. As an alternative, we define the ``half" representation of $\hat{\eta}$, $\hat{\eta}^h$, by $\hat{\eta}^h=[\hat{\sigma}^2, \text{vech}(\hat{D}_1)',\hdots,\text{vech}(\hat{D}_r)']'$ and present the below exact closed-form expression for the derivative of $S^2$ in terms of $\hat{\eta}^h$;
\begin{equation}\nonumber\begin{aligned}
& \frac{d S^2(\hat{\eta}^h)}{d \hat{\sigma}^2}=L(X'\hat{V}^{-1}X)^{-1}L',\\
& \frac{d S^2(\hat{\eta}^h)}{d  \text{vech}(\hat{D}_k)} = 
\hat{\sigma}^{2}\mathcal{D}_{q_k}'\bigg(\sum_{j=1}^{l_k}\hat{B}_{(k,j)}\otimes \hat{B}_{(k,j)}\bigg),
\end{aligned}
\end{equation}
\\
\noindent
where $\hat{B}_{(k,j)}$ is given by:
\begin{equation}\nonumber
    \hat{B}_{(k,j)} = Z_{(k,j)}'\hat{V}^{-1}X(X'\hat{V}^{-1}X)^{-1}L'.
\end{equation}
We obtain an expression for var$(\hat{\eta}^h)$ by noting that the asymptotic variance of $\hat{\eta}^h$ is given by $\mathcal{I}(\hat{\eta}^h)^{-1}$ where $\mathcal{I}(\hat{\eta}^h)$ is a sub-matrix of $\mathcal{I}(\hat{\theta}^h)$, given by equations (\ref{FSFI1})-(\ref{FSFI3}).

In summary, we have provided all the closed-form expressions necessary to perform the Satterthwaite degrees of freedom estimation method for any LMM described by (\ref{LMMdef}) and (\ref{LMMdef2}). For the remainder of this work, estimation of $v$ using the above expressions is referred to as the ``direct-SW'' method. This name reflects the direct approach taken for the evaluation of the right-hand side of (\ref{swdf2}). We note that this approach may also be extended to models using the constrained covariance structures of Section \ref{covstruct} by employing the Fisher Information matrix given by equation \ref{derivCk} and transforming the derivative of $S^2(\eta)$ appropriately (details of which are provided by Theorem \ref{s2derivthm2} of Appendix \ref{swproof}). We conclude this section by noting that this method can also be used in a similar manner for estimating the degrees of freedom of an approximate F-statistic based on the multi-factor LMM, as described in Appendix \ref{Fsect}.


\subsection{Simulation methods}\label{SimMeth}

To assess the accuracy and efficiency of each of the proposed LMM parameter estimation methods described in Sections \ref{FSsection}-\ref{CSFSsection} and the direct-SW degrees of freedom estimation method described in Section \ref{swsection}, extensive simulations were conducted. The methods used to perform these simulations are described in this section, with the results of the simulations presented in Section \ref{SimRes}. All reported results were obtained using an Intel(R) Xeon(R) Gold 6126 2.60GHz processor with 16GB RAM.

\subsubsection{Parameter estimation}\label{paramEstSimsection}

The algorithms of Sections \ref{FSsection}-\ref{CSFSsection} have been implemented in the programming language Python. Results taken across three simulation settings, each with a different design structure, are provided for parameter estimation performed using each of the algorithms. The Fisher Scoring algorithms presented in this work are compared against one another, the R package lmer, and the baseline truth used to generate the simulations. All methods are contrasted in terms of output, computation time and, for the methods presented in this paper, the number of iterations until convergence. Convergence of each method was assessed by verifying whether successive log-likelihood estimates differed by more than a predefined tolerance of $10^{-6}$. $1000$ individual simulation instances were run for each simulation setting.

All model parameters were held fixed across all runs. In every simulation setting, test data was generated according to model (\ref{LMMdef}). Each of the three simulation settings imposed a different structure on the random effects design and covariance matrices, $Z$ and $D$. The first simulation setting employed a single factor design ($r = 1$) with two random effects (i.e. $q_1 = 2$) and $50$ levels (i.e. $l_1 = 50$). The second simulation setting employed two crossed factors ($r=2$) where the number of random effects and numbers of levels for each factor were given by $q_1=3$, $q_2 =2$, $l_1=100$ and $l_2=50$, respectively. The third simulation setting used three crossed factors (i.e. $r= 3$) with the number of random effects and levels for each of the factors given by $q_1=4$, $q_2 =3$, $q_3 =2$, $l_1=100$, $l_2=50$ and $l_3=10$, respectively. In all simulations, the number of observations, $n$, was held fixed at $1000$. In each simulated design, the first column of the fixed effects design matrix, $X$, and the first random effect in the random effects design matrix, $Z$, were treated as intercepts. Within each simulation instance, the remaining (non-zero) elements of the variables $X$ and $Z$, as well as those of the error vector $\epsilon$, were generated at random according to the standard univariate normal distribution. The random effects vector $b$ was simulated according to a normal distribution with covariance $D$, where $D$ was predefined, exhibited no particular constrained structure and contained a mixture of both zero and non-zero off-diagonal elements. The assignment of observations to levels for each factor was performed at random with the probability of an observation belonging to any specific level held constant and uniform across all levels. 

Fisher Scoring methods for the single-factor design have been well-studied (c.f. \citet{Demidenko:2013mfx}) and the inclusion of the first simulation setting is only for comparison purposes. The second and third simulation settings represent more complex crossed factor designs for which the proposed Fisher Scoring based parameter estimation methods did not previously exist. To assess the performance of the proposed algorithms, the Mean Absolute Error ($MAE$) and Mean Relative Difference ($MRD$) were used as performance metrics. The methods considered for parameter estimation were those described in Sections \ref{FSsection}-\ref{CSFSsection} and the baseline truth used for comparison was either the baseline truth used to generate the simulated data or the lmer computed estimates.

All methods were contrasted in terms of the $MAE$ and $MRD$ of both $\beta$ and the variance product $\sigma^2D$. The variance product $\sigma^2D$ was chosen for comparison instead of the individual components $\sigma^2$ and $D$ as the variance product $\sigma^2D$ is typically of intrinsic interest during the inference stage of conventional statistical analyses and is often employed for further computation. For each of the methods proposed in Sections \ref{FSsection}-\ref{CSFSsection}, both ML and ReML estimation variants of the methods were assessed, where the ReML variants of each method were performed according to the adjustments detailed in Appendix \ref{remlApp}. The computation time for each method was also recorded and contrasted against lmer. To ensure a fair comparison of computation time, the recorded computation times for lmer were based only on the performance of the `optimizeLmer' function, which is the module employed for parameter estimation by lmer.

The specific values of $\beta, \sigma^2,$ and $D$ employed for each simulation setting can be found in Section $S1$ of the Supplementary Material. Formal definitions of MAE and MRD measures used for comparison are given in Section $S2$ of the Supplementary Material.

\subsubsection{Degrees of freedom estimation}\label{dfmeth}

The accuracy and validity of the direct-SW degrees of freedom estimation method proposed in Section \ref{swsection} were assessed through further simulation. To achieve this, as LMM degrees of freedom are assumed to be specific to the experiment design, a single design was chosen at random from each of the three simulation settings described in Section \ref{paramEstSimsection}. With the fixed effects and random effects matrices, $X$ and $Z$, now held constant across simulations, $1000$ simulations were run for each simulation setting. The random effects vector, $b$, and random error vector, $\epsilon$, were allowed to vary across simulations according to normal distributions with the appropriate variances. In each simulation, degrees of freedom were estimated via the direct-SW method for a predefined contrast vector, corresponding to a fixed effect that was truly zero.

The direct-SW estimated degrees of freedom were compared to a baseline truth and the degrees of freedom estimates produced by the R package lmerTest. Baseline truth was established in each simulation setting using  $1,000,000$ simulations to empirically estimate $\text{Var}(S^2(\hat{\eta}))$, giving a single value for the denominator of $v(\hat{\eta})$ from Equation (\ref{swdf}). Following this, in each of the $1000$ simulation instances described above, the numerator of (\ref{swdf}) was recorded, giving $1000$ estimates of $v(\hat{\eta})$. The final estimate was obtained as an average of these $1000$ values. All lmerTest degrees of freedom estimates were obtained using the same simulated data as was used by the direct-SW method and were computed using the `contest1D' function from the lmerTest package. Whilst all baseline truth measures were computed using parameter estimates obtained using the FSFS algorithm of Section \ref{FSFSsection}, we believe this has not induced bias into the simulations as, as discussed in Section \ref{MainSimRes}, all simulated FSFS-derived parameter estimates considered in this work agreed with those provided by lmer within a tolerance level on the scale of machine error.

\subsection{Real data methods}\label{RealDatMeth}

To illustrate the usage of the methodology presented in Sections \ref{AlgoSection}-\ref{swsection} in practical situations, we provide two real data examples. These examples are described fully in this section and the results are presented in Section \ref{RealDatRes}. Again, all reported results were obtained using an Intel(R) Xeon(R) Gold 6126 2.60GHz processor with 16GB RAM. For each reported computation time, averages were taken across $50$ repeated runs of the given analysis.

\subsubsection{The SAT score example}\label{SATexample}

The first example presented here is based on data from the longitudinal evaluation of school change and performance (LESCP) dataset (\citet{Turnbull1999}). This dataset has notably been previously analyzed by  \citet{Raudenbush2008} and was chosen for inclusion in this work because it previously formed the basis for between-software comparisons of LMM software packages by \citet{west2014linear}. The LESCP study was conducted in $67$ American schools in which SAT (Student Aptitude Test) math scores were recorded for randomly selected samples of students. As in \citet{west2014linear}, one of the $67$ schools from this dataset was chosen as the focus for this analysis. For each individual SAT score, unique identifiers (ID's) for the student who took the test and for the teacher who prepared the student for the test were recorded. As many students were taught by multiple teachers and all teachers taught multiple students, the grouping of SAT scores by student ID and the grouping of SAT scores by teacher ID constitute two crossed factors. In total, $n=234$ SAT scores were considered for analysis. The SAT scores were taken from $122$ students taught by a total of $12$ teachers, with each student sitting between $1$ and $3$ tests.

The research question in this example concerns how well a student's grade (i.e. year of schooling) predicted their mathematics SAT score (i.e. did students improve in mathematics over the course of their education?). For this question, between-student variance and between-teacher variance must be taken into consideration as different students possess different aptitudes for mathematics exams and different teachers possess different aptitudes for teaching mathematics. In practice, this is achieved by employing an LMM which includes random intercept terms for both the grouping of SAT scores according to student ID and the grouping of SAT scores according to teacher ID. For the $k^{th}$ mathematics SAT test taken by the $i^{th}$ student, under the supervision of the $j^{th}$ teacher, such a model could be stated as follows:
\begin{equation}\nonumber
    \text{MATH}_{i,j,k} = \beta_0 + \beta_1 \times \text{YEAR}_{i,j,k} + s_i + t_j + \epsilon_{i,j,k},
\end{equation}
where  MATH$_{i,j,k}$ is the SAT score achieved and YEAR$_{i,j,k}$ is the grade of the student at the time the test was taken. In the above model, $\beta_0$ and $\beta_1$ are unknown parameters and $s_i$, $t_j$ and $\epsilon_{i,j,k}$ are independent mean-zero random variables which differ only in terms of their covariance. $s_i$ is the random intercept which models between-student variance, $t_j$ is the random intercept which models between-teacher variance, and $\epsilon_{i,j,k}$ is the random error term. The random variables $s_i$, $t_j$ and $\epsilon_{i,j,k}$ are assumed to be mutually independent and follow the below distributions:
\begin{equation}\nonumber
\begin{aligned}[b]
    & s_i \sim N(0, \sigma^2_s), \hspace{1cm} t_j \sim N(0, \sigma^2_t), \\
    & \epsilon_{i,j,k} \sim N(0,\sigma^2),
\end{aligned}
\end{equation}
where the parameters $\sigma^2_s$, $\sigma^2_t$ and $\sigma^2$ are the unknown student, teacher and residual variance parameters, respectively.

The random effects in the SAT score model can be described using the notation presented in  previous sections as follows; $r=2$ (i.e. observations are grouped by two factors, student ID and teacher ID), $q_1=q_2=1$ (i.e. one random effect is included for each factor, the random intercepts $s_i$ and $t_j$, respectively), and $l_1=122,l_2=12$ (i.e. there are $122$ students and $12$ teachers). When the model is expressed in the form described by equations (\ref{LMMdef}) and (\ref{LMMdef2}), the random effects design matrix $Z$ is a $0-1$ matrix. In this setting, the positioning of the non-zero elements in $Z$ indicates the student and teacher associated with each test score. The random effects covariance matrices for the two factors, student ID and teacher ID, are given by $D_0=[\sigma^2_s]$ and $D_1=[\sigma^2_t]$, respectively.

For the SAT score model, the estimated parameters obtained using each of the methods detailed in Sections \ref{FSsection}-\ref{CSFSsection} are reported. For comparison, the parameter estimates obtained by the R package lmer are also given. As the estimates for the fixed effects parameters, $\beta_0$ and $\beta_1$, are of primary interest, ML was employed to obtain all reported parameter estimates. For this example, methods are contrasted in terms of output, computation time and the number of iterations performed.

\subsubsection{The twin study example}\label{AceExample}

The second example presented in this work aims to demonstrate the flexibility of the constrained covariance approaches described in Section \ref{covstruct}. This example is based on the ACE model; an LMM commonly employed to analyze the results of twin studies by accounting for the complex covariance structure exhibited between related individuals. The ACE model achieves this by separating between-subject response variation into three categories: variance due to additive genetic effects ($\sigma^2_a$), variance due to common environmental factors ($\sigma^2_c$), and residual error ($\sigma^2_e$). 
For this example, we utilize data from the Wu-Minn Human Connectome Project (HCP) (\citet{Essen2013}). The HCP dataset contains brain imaging data collected from $1,200$ healthy young adults, aged between $22$ and $35$, including data from $300$ twin pairs and their siblings. We do not make use the imaging data in the HCP dataset but, instead, focus on the baseline variables for cognition and alertness.

The primary research question considered in this example focuses on how well a subject's quality of sleep predicts their English reading ability. The variables of primary interest used to address this question are subject scores in the Pittsburgh Sleep Quality Index (PSQI) questionnaire and an English language recognition test (ENG).  Other variables included the subjects' age in years (AGE) and sex (SEX), as well as an age-sex interaction effect. A secondary research question considered asks ``How much of the between-subject variance observed in English reading ability test scores can be explained by additive genetic and common environmental factors?''. To address this question, the covariance parameters $\sigma^2_a$, $\sigma^2_c$ and $\sigma^2_e$ must be estimated.

To model the covariance components of the ACE model, a simplifying assumption that all family units share common environmental factors was made. Following the work of \cite{Winkler2015}, family units were first categorized by their internal structure into what shall be referred to as ``family structure types'' (i.e. unique combinations of full-siblings, half-siblings and identical twin pairs which form a family unit present in the HCP dataset). In the HCP dataset, $19$ such family structure types were identified. In the following text, each family structure type shall be treated as a factor in the model. For the $i^{th}$ observation to belong to the $j^{th}$ level of the $k^{th}$ factor in the model may be interpreted as the $i^{th}$ subject belonging to the $j^{th}$ family exhibiting family structure of type $k$. The model employed for this example is given by:

\begin{equation}\nonumber
\begin{aligned}
    \text{ENG}_{k,j,i} = & \beta_0 + \beta_1 \times \text{AGE}_{i} + \beta_2 \times \text{SEX}_{i} + \hdots & \\ 
    &  \beta_3 \times \text{AGE}_{i} \times \text{SEX}_{i} + \beta_4 \times \text{PSQI}_{i} + ... \\
    & \gamma_{k,j,i} + \epsilon_{k,j,i}, & 
\end{aligned}
\end{equation}
where both $\gamma_{k,j,i}$ and $\epsilon_{k,j,i}$ are mean-zero random variables. The random error, $\epsilon_{k,j,i}$, is defined by;
\begin{equation}\nonumber
\begin{aligned}[b]
    & \epsilon_{k,j,i} \sim N(0,\sigma^2_e), \\
\end{aligned}
\end{equation}
and the random term $\gamma_{k,j,i}$ models the within-``family unit" covariance. Further detail on the specification of $\gamma_{k,j,i}$ can be found in Appendix (\ref{GammaApp}).

In the notation of the previous sections, the number of factors in the ACE model, $r$, is equal to the number of family structure types present in the model (i.e. $19$). For each family structure type present in the model, family structure type $k$, $l_k$ is the number of families who exhibit such structure and $q_k$ is the number of subjects present in any given family unit with such structure. As there is a unique random effect (i.e. a unique random variable, $\gamma_{k,j,i}$) associated with each individual subject, none of which are scaled by any coefficients, the random effects design matrix, $Z$, is the $(n \times n)$ identity matrix. To describe $\{D_k\}_{k \in \{1,...,r\}}$ requires the known matricies $\mathbf{K}^a_k$ and $\mathbf{K}^c_k$, which specify the kinship (expected genetic material) and environmental effects shared between individuals, respectively (See Appendix \ref{KinApp} for more details). Given $\mathbf{K}^a_k$ and $\mathbf{K}^c_k$, the covariance components $\sigma^2$ and $\{D_k\}_{k \in \{1,...,r\}}$ are given as $\sigma^2=\sigma^2_e$ and $D_k =  \sigma^{-2}_e(\sigma^2_a\mathbf{K}^a_k + \sigma^2_c\mathbf{K}^c_k)$ respectively.

As the covariance components $\sigma_a$ and $\sigma_c$ are of practical interest in this example, optimization is performed according to the ReML criterion using the adjustments described in Appendix \ref{remlApp} and covariance structure is constrained via the methods outlined in Section \ref{covstruct}. Further detail on the constrained approach for the ACE model can be found in Appendix \ref{AceConApp}. Discussion of computational efficiency for the ACE model is also provided in Appendix \ref{ACEcompApp}.

Section \ref{AceRes} reports the maximized restricted log-likelihood values and parameter estimates obtained using the Fisher Scoring method. Also given are approximate T-statistics for each fixed effects parameter, alongside corresponding degrees of freedom estimates and p-values obtained via the methods outlined in Section \ref{swsection}. To verify correctness, the restricted log-likelihood of the ACE model was also maximized numerically using the implementation of Powell’s bi-directional search based optimization method (\cite{Powell1964}) provided by the SciPy Python package. The maximized restricted log-likelihood values and parameter estimates produced were then contrasted against those provided by the Fisher Scoring method. The OLS (Ordinary Least Squares) estimates, which would be obtained had the additive genetic and common environmental variance components not been accounted for in the LMM analysis, are also provided for comparison. 

\section{Results}
\subsection{Simulation results}\label{SimRes}

\subsubsection{Parameter estimation results}\label{MainSimRes}

The results of the parameter estimation simulations of \ref{paramEstSimsection} were identical. In all three settings, all parameter estimates and maximised likelihood criteria produced by the FS, FFS, SFS and FSFS methods were identical to those produced by lmer up to a machine error level of tolerance. In terms of the maximisation of likelihood criteria, there was no evidence to suggest that lmer outperformed any of the FS, FFS, SFS or FSFS methods, or vice versa. The CSFS method provided worse performance, however, with maximised likelihood criteria which were lower than those reported by lmer in approximately $2.5\%$ of the simulations run for simulation setting $3$. The reported maximised likelihood criteria for these simulations were all indicative of convergence failure. For each parameter estimation method considered during simulation, the observed performance was unaffected by the choice of likelihood estimation criteria (i.e. ML or ReML) employed.

The $MRD$ and $MAE$ values provided evidence for strong agreement between the parameter estimates produced by lmer and the FS, FFS, SFS and FSFS methods. The largest $MRD$ values observed for these methods, taken relative to lmer, across all simulations and likelihood criteria, were $1.03 \times 10^{-3}$ and $2.12 \times 10^{-3}$ for $\beta$ and $\sigma^2D$, respectively. The largest observed $MAE$ values for these methods, taken relative to lmer, across all simulations, were given by$1.02 \times 10^{-5}$ and $4.30 \times 10^{-4}$ for $\beta$ and $\sigma^2D$, respectively. These results provide clear evidence that the discrepancies between the parameter estimates produced by lmer and those produced by the FS, FFS, SFS and FSFS methods are extremely small in magnitude. Due to the extremely small magnitudes of these differences, $MAE$ and $MRD$ values are not reported in further detail here. For more detail, see Supplementary Material Sections $S3$-$S6$.

To compare the proposed methods in terms of computational performance, Tables \ref{perfTable} and \ref{perfTable2} present the average time and number of iterations performed for each of the $5$ methods, using the ML and ReML likelihood criteria, respectively. Corresponding computation times are also provided for lmer. For both the ML and ReML likelihood criteria, all Fisher Scoring methods demonstrated considerable efficiency in terms of computation speed. While the improvement in speed is minor for simulation setting $1$, for multi-factor simulation settings $2$ and $3$, the performance gains can be seen to be considerable. Overall, the only method that consistently demonstrated notably worse performance than the others was the CSFS algorithm. In addition to requiring a longer computation time, the CSFS algorithm employed many more iterations. 

Within a setting where $q$ (i.e. the number number of columns in the random effects design matrix) is small, the results of these simulations demonstrate strong computational efficiency for the FS, FFS, SFS and FSFS methods. However, we emphasise that no claim is made to suggest that the observed results will generalise to larger values of $q$ or all possible designs. For large sparse LMMs that include a large number of random effects, without further adaptation of the methods presented in this work to employ sparse matrix methodology, it is expected that lmer will provide superior performance. We again emphasise here that our purpose in undertaking this work is not to compete with the existing LMM parameter estimation tools. Instead, the focus of this work is to develop methodology which provides performance comparable to the existing tools but can be used in situations in which the existing tools may not be applicable due to practical implementation considerations.
\afterpage{
\begin{table*}[hbt!]\footnotesize
\centering
\begin{tabular*}{\textwidth}{c @{\extracolsep{\fill}}cccccc@{}}
 \hline
 \textbf{Method} & \textbf{FS} & \textbf{FFS} & \textbf{SFS} & \textbf{FSFS} & \textbf{CSFS} & \textbf{lmer} \\ 
 \hline
    & & & & & & \\
   \textit{Simulation $1$} & & & & & & \\
    & & & & & & \\
   $t$ (Time/s) & $0.041$ & $0.037$ & $0.026$ & $0.055$ & $0.057$ & $0.059$ \\
    & $(0.012)$ & $(0.011)$ & $(0.007)$ & $(0.058)$ & $(0.017)$ & $(0.021)$ \\  
   $n_{it}$ (No. of iterations) & $5.243$ & $5.243$ & $6.048$ & $6.048$ & $10.461$ & - \\
    & $(0.613)$ & $(0.613)$ & $(0.256)$ & $(0.256)$ & $(1.851)$ & - \\ 
    & & & & & & \\
 \hline
    & & & & & & \\
   \textit{Simulation $2$} & & & & & &  \\
    & & & & & & \\
   $t$ (Time/s) & $0.129$ & $0.112$ & $0.105$ & $0.125$ & $0.220$ & $0.648$ \\
    & $(0.047)$ & $(0.042)$ & $(0.036)$ & $(0.061)$ & $(0.070)$ & $(0.050)$ \\  
   $n_{it}$ (No. of iterations) & $6.694$ &  $6.694$ & $8.323$ & $8.323$ & $14.281$ & - \\
    & $(0.964)$ & $(0.964)$ & $(0.534)$ & $(0.534)$ & $(1.381)$ & - \\ 
    & & & & & & \\
 \hline
    & & & & & & \\
   \textit{Simulation $3$} & & & & & & \\
   & & & & & & \\
   $t$ (Time/s) & $1.081$ & $0.889$ & $1.872$ & $1.941$ & $3.970$ & $5.979$ \\
    & $(0.434)$ & $(0.429)$ & $(0.915)$ & $(0.869)$ & $(1.094)$ & $(0.299)$ \\  
   $n_{it}$ (No. of iterations) & $8.133$ & $8.133$ & $16.054$ & $16.054$ & $33.437$ & - \\
    & $(0.743)$ & $(0.743)$ & $(0.835)$ & $(0.835)$ & $(24.265)$ & - \\ 
   & & & & & & \\
 \hline
\end{tabular*}
  \caption{The average time in seconds and the number of iterations reported for maximum likelihood estimation performed using the FS, FFS, SFS, FSFS and CSFS methods. For each simulation setting, results displayed are taken from averaging across $1000$ individual simulations. Also given are the average times taken by the R package lmer to perform maximum likelihood parameter estimation on the same simulated data. Standard deviations are given in brackets below each entry in the table.}
  \label{perfTable}
\end{table*}

\begin{table*}[hbt!]\footnotesize
\centering
\begin{tabular*}{\textwidth}{c @{\extracolsep{\fill}}cccccc@{}}
 \hline
 \textbf{Method} & \textbf{FS} & \textbf{FFS} & \textbf{SFS} & \textbf{FSFS} & \textbf{CSFS} & \textbf{lmer} \\ 
 \hline
    & & & & & & \\
   \textit{Simulation $1$} & & & & & & \\
    & & & & & & \\
   $t$ (Time/s) & $0.057$ & $0.052$ & $0.043$ & $0.064$ & $0.088$ & $0.071$ \\
    &  $(0.014)$ & $(0.011)$ & $(0.008)$ & $(0.042)$ & $(0.021)$ & $(0.021)$ \\  
   $n_{it}$ (No. of iterations) & $5.261$ & $5.261$ & $6.053$ & $6.053$ & $10.404$ & - \\
    & $(0.591)$ & $(0.591)$ & $(0.257)$ & $(0.257)$ & $(1.829)$ & - \\ 
    & & & & & & \\
 \hline
    & & & & & & \\
   \textit{Simulation $2$} &  &  &  & & &  \\
    & & & & & & \\
   $t$ (Time/s) & $0.175$ & $0.154$ & $0.157$ & $0.169$ & $0.306$ & $0.810$ \\
    &  $(0.043)$ & $(0.039)$ & $(0.035)$ & $(0.041)$ & $(0.067)$ & $(0.065)$ \\  
   $n_{it}$ (No. of iterations) & $6.758$ & $6.758$ & $8.413$ & $8.413$ &   $14.283$ & - \\
    & $(0.984)$ & $(0.984)$ & $(0.561)$ & $(0.561)$ & $(1.404)$ & - \\ 
    & & & & & & \\
 \hline
    & & & & & & \\
   \textit{Simulation $3$} & & & & & & \\
   & & & & & & \\
   $t$ (Time/s) & $1.615$ & $1.378$ & $2.968$ & $3.002$ & $6.604$ & $7.457$ \\
    & $(0.404)$ & $(0.344)$ & $(0.752)$ & $(0.814)$ & $(5.455)$ & $(0.442)$ \\  
   $n_{it}$ (No. of iterations) & $8.084$ & $8.084$ & $16.018$ & $16.018$ & $34.611$ & - \\
    & $(0.730)$ & $(0.730)$ & $(0.798)$ & $(0.798)$ & $(25.526)$ & - \\ 
   & & & & & & \\
 \hline
\end{tabular*}
  \caption{The average time in seconds and the number of iterations reported for restricted maximum likelihood estimation performed using the FS, FFS, SFS, FSFS and CSFS methods. For each simulation setting, results displayed are taken from averaging across $1000$ individual simulations. Also given are the average times taken by the R package lmer to perform restricted maximum likelihood parameter estimation on the same simulated data. Standard deviations are given in brackets below each entry in the table.}
  \label{perfTable2}
\end{table*}
\clearpage
}

Two potential causes for the poor performance of the CSFS method have been identified below. The first cause is presented in \citet{Demidenko:2013mfx}, in which it is argued that the Cholesky parameterised algorithm will provide slower convergence than that of the unconstrained alternative methods due to structural differences between the likelihood surfaces over which parameter estimation is performed. This reasoning offers a likely explanation for the higher number of iterations and computation time observed for the CSFS method across all simulations. However, this does not explain the small number of simulations in simulation setting $3$ in which evidence of convergence failure was observed.

The second possible explanation for the poor performance of the CSFS algorithm presents a potential reason for the observed convergence failure. \citet{Pinheiro1996} note that the Cholesky parameterisation of an arbitrary matrix is not unique and the number of possible choices for the Cholesky factorisation of a square positive definite matrix of dimension $(q \times q)$ increases with the dimension $q$. This means that, for each covariance matrix $D_k$, there are multiple Cholesky factors, $\Lambda_k$, which satisfy $\Lambda_k\Lambda_k'=D_k$. The larger $D_k$ is in dimension, the greater the number of $\Lambda_k$ that correspond to $D_k$ there are. \citet{Pinheiro1996} argue that when optimal solutions are numerous and close together in the parameter space, numerical problems can arise during optimisation. We note that when compared with the other simulation settings considered here, the design for simulation setting $3$ contains the largest number of factors and random effects. Consequently, the covariance matrices, $D_k$, are large for this simulation setting and numerous Cholesky factors which correspond to the same optimal solution for the covariance matrix, $D$, exist. For this reason, simulation setting $3$ is the most susceptible to numerical problems of the kind described by \citet{Pinheiro1996}. We suggest that this is a likely accounting for the $2.5\%$ of simulations in which convergence failure was observed. In summary, the simulation results offer strong evidence for the correctness and efficiency of the Fisher Scoring methods proposed, with the exception of the CSFS method, which experiences slower performance and convergence failure in rare cases.

\subsubsection{Degrees of freedom estimation results}

\begin{table*}[hbt!]
\centering
\begin{tabular*}{\textwidth}{c @{\extracolsep{\fill}}ccc@{\hskip 0.5in}}
 \hline
 \textbf{Method} & \textbf{Truth}  & \textbf{Direct-SW} & \textbf{lmerTest} \\ 
 \hline
    & & & \\
   \textit{Simulation $1$} & & & \\
    & & & \\
    Mean & $910.93$ & $910.68$ & $906.62$ \\
    Standard Deviation & $0.0$ & $2.36$ & $2.39$ \\
    Mean Squared Error & $0.0$ & $5.64$ & $24.23$ \\
    & & & \\
 \hline
    & & & \\
   \textit{Simulation $2$} &  &  & \\
    & & & \\
    Mean & $844.61$ & $842.17$ & $837.79$ \\
    Standard Deviation & $0.0$ & $7.16$ & $7.26$ \\
    Mean Squared Error & $0.0$ & $57.20$ & $99.21$ \\
    & & & \\
 \hline
    & & & \\
   \textit{Simulation $3$} & & & \\
    & & & \\
    Mean & $707.01$ & $700.96$ & $695.08$ \\
    Standard Deviation & $0.0$ & $17.19$ & $17.67$ \\
    Mean Squared Error & $0.0$ & $332.01$ & $454.53$ \\
    & & & \\
 \hline
\end{tabular*}
  \caption{The mean, standard deviation and mean squared error for $1,000$ degrees of freedom estimates in each simulation setting. Results are displayed for both the lmerTest and Direct-SW methods, alongside ``true" mean values, which were established using the moment-matching based approach outlined in Section \ref{dfmeth} and computed using $1,000,000$ simulation instances.}\label{dfTab}
\end{table*}

Across all degrees of freedom simulations, results indicated that the degrees of freedom estimates produced by the direct-SW method possessed both lower bias and lower variance than those produced by lmerTest. The degrees of freedom estimates, for each of the three simulation settings, are summarized in Table \ref{dfTab}. 

It can be seen from Table \ref{dfTab} that throughout all simulation settings both direct-SW and lmerTest appear to underestimate the true value of the degrees of freedom. However, the bias observed for the lmerTest estimates is notably more severe that of the direct-SW method, suggesting that the estimates produced by direct-SW have a higher accuracy than those produced by lmerTest. The observed difference in the standard deviation of the degrees of freedom estimates between the lmerTest and direct-SW methods is less pronounced. However, in all simulation settings, lower standard deviations are reported for direct-SW, suggesting that the estimates produced by direct-SW have a higher precision than those produced by lmerTest. 

For both lmerTest and direct-SW, the observed bias and variance increase with simulation complexity. The simulations provided here indicate that the severity of disagreement between direct-SW and lmerTest increases as the complexity of the random effects design increases. This observation matches the expectation that the accuracy of the numerical gradient estimation employed by lmerTest will worsen as the complexity of the parameter space increases. Reported mean squared errors are also provided, indicating further that the direct-SW method outperformed lmerTest in terms of accuracy and precision in all simulation settings.

\subsection{Real data results}\label{RealDatRes}

\subsubsection{The SAT score results}\label{SATres}
For the SAT score example described in Section \ref{SATexample}, the log-likelihood, fixed effect parameters and variance components estimates produced by the Fisher Scoring algorithms were identical to those produced by lmer. Further, the reported computation times for this example suggested little difference between all methods considered in terms of computational efficiency. Of the Fisher Scoring methods considered, the FS and FFS methods took the fewest iterations to converge while the SFS and FSFS methods took the most iterations to converge. The results presented here exhibit strong agreement with those reported in \citet{west2014linear}, in which the same model was used as the basis for a between-software comparison of LMM software packages. For completeness, the full table of results can be found in Supplementary Material Section $S9$.

\subsubsection{The twin study results}\label{AceRes}

The results for the twin study example described in Section \ref{AceExample} are presented in Table \ref{AceTab}. It can be seen from Table \ref{AceTab} that the Powell optimizer and Fisher Scoring method attained extremely similar optimized likelihood values, with Fisher Scoring converging notably faster. This result offers further evidence for the correctness of the parameter estimates produced by the Fisher Scoring method as it is unlikely that both Powell estimation and Fisher Scoring would converge to the same solution if the solution were suboptimal. The parameter estimates produced by Fisher Scoring and Powell optimization can be seen to be smaller in magnitude than those estimated by OLS, highlighting how the inclusion of additional variance terms in the model can have a meaningful impact on the conclusion of the analysis.

Also provided in Table \ref{AceTab}, are approximate t-tests based on the Fisher Scoring and Powell optimizer parameter estimates, with corresponding standard t-tests given based on the OLS estimates. At the $5\%$ significance level, the approximate t-tests conclude that the fixed effects parameters corresponding to the `Intercept', `Age and Sex Interaction' and the `PSQI Score' are non-zero in value. While it may be expected that age should affect reading ability, no significant effect was observed for the `Age' covariate. This lack of observed effect may be explained by the narrow age range of subjects present in the HCP dataset, with subjects ranging from 22 to 35 years in age, and by the fact that individual observations were recorded in units of years. In general, the OLS-based t-tests produced similar conclusions to those produced by the FS-based approximate t-tests. A notable exception, however, is given by the t-tests for the fixed effect associated with the `Sex' covariate. While the standard OLS t-test reported at the $5\%$ significance level that the `Sex' fixed effect was non-zero in value, the FS-based approximate t-test concluded that there was not evidence to support this claim. This result further highlights the importance of modeling all relevant variance terms.

\section{Discussion}

In this work, we have presented derivations for and demonstrated potential applications of, score vector and Fisher Information matrix expressions for the LMMs containing multiple random factors. While many of the examples presented in this paper were benchmarked against existing software, it is not the authors' intention to suggest that the proposed methods are superior to existing software packages. Instead, this work aims to complement existing LMM parameter estimation research. This aim is realized through careful exposition of the score vectors and Fisher Information matrices and detailed description of methodology and algorithms.

Modern approaches to LMM parameter estimation typically depend upon conceptually complex mathematical operations which require support from a range of software packages and infrastructure. This work has been, in large part, motivated by current deficiencies in vectorized support for such operations. Vectorized computation is a crucial requirement for medical imaging applications, in which hundreds of thousands of mixed models must be estimated concurrently. It is with this application in mind that this work has been undertaken. We intend to pursue the application itself in future work.
\afterpage{
\begin{table*}[hbt!]
\centering
\begin{tabular*}{\textwidth}{cc @{\extracolsep{\fill}}cccc@{}}
 \hline
 \multicolumn{2}{c}{\textbf{Estimation Method}} & \textbf{OLS} & \textbf{Powell} & \textbf{FS} \\ 
 \hline
    & & & &\\
   \multicolumn{2}{c}{\textit{Performance}} & & & \\
    \multicolumn{2}{c}{$l$ (Log-likelihood)} & $-2133.49$ & $-2005.83$ & $-2005.81$ \\
    \multicolumn{2}{c}{$t$ (Time in seconds)} & $<.001$ & $93.32$ & $2.31$ \\
    & & &\\
 \hline
    & & & \\
   \multicolumn{2}{c}{\textit{Fixed effects parameters (Standard Errors)}} & & & \\
    \multicolumn{2}{c}{$\beta_0$ (Intercept)} & $121.46$ $(3.61)$ & $118.43$ $(3.40)$ & $118.34$ $(3.42)$ \\
    \multicolumn{2}{c}{$\beta_1$ (Age)} & $-0.11$ $(0.12)$ & $-0.04$ $(0.11)$ & $-0.04$ $(0.11)$\\
    \multicolumn{2}{c}{$\beta_2$ (Sex)} & $-10.74$ $(5.07)$ & $-7.22$ $(4.64)$ & $-7.61$ $(4.66)$ \\ 
    \multicolumn{2}{c}{$\beta_3$ (Age:sex)} & $0.45$ $(0.18)$ & $0.33$ $(0.16)$ & $0.34$ $(0.16)$ \\
    \multicolumn{2}{c}{$\beta_4$ (PSQI score)} & $-0.49$ $(0.12)$ & $-0.31$ $(0.10)$ & $-0.30$ $(0.10)$\\
    & & &\\
 \hline
    & & &\\
   \multicolumn{2}{c}{\textit{Covariance parameters}} & & & \\
    \multicolumn{2}{c}{$\sigma^2_a$ (Additive genetic)} & $0.00$ & $48.20$ & $50.67$ \\
    \multicolumn{2}{c}{$\sigma^2_c$ (Common environment)} & $0.00$ & $27.53$ & $26.89$ \\
    \multicolumn{2}{c}{$\sigma^2_e$ (Residual error)} & $110.04$ & $33.57$ & $32.71$ \\
    & & & \\
 \hline
    & & &\\
   \multicolumn{2}{c}{\textit{Tests for Fixed Effects}} & & & \\
    & & &\\
   & \textit{T-statistic} & $33.67$ & $34.83$ & $34.65$ \\
     Intercept & \textit{degrees of freedom} & $1105.00$ & $920.59$ & $921.57$ \\
    & \textit{p-value} & $<.001^*$ & $<.001^*$ & $<.001^*$ \\
    & & &\\
    & \textit{T-statistic} & $-0.94$ & $-0.350$ & $-0.320$ \\
    Age & \textit{degrees of freedom} & $1105.00$ & $907.44$ & $908.42$ \\
    & \textit{p-value} & $.350$ & $.726$ & $.749$\\
    & & &\\
    & \textit{T-statistic} &  $-2.12$ & $-1.56$ & $-1.63$ \\
    Sex & \textit{degrees of freedom} & $1105.00$ & $833.01$ & $832.96$ \\
    & \textit{p-value} & $.034^*$ & $.120$ & $.103$ \\
    & & &\\
    & \textit{t-statistic} & $2.58$ & $2.03$ &  $2.10$ \\
    Age:sex & \textit{degrees of freedom} & $1105.00$ & $830.53$ & $830.61$ \\
    & \textit{p-value} & $.010^*$ & $.043^*$ & $.036^*$\\
    & & &\\
    & \textit{T-statistic} &  $-4.25$ & $-3.19$ & $-3.17$ \\
    PSQI score & \textit{degrees of freedom} & $1105.00$ & $933.67$ & $928.87$ \\
    & \textit{p-value} & $<.001^*$ & $.001^*$ & $.001^*$\\
    & & &\\
    \hline
\end{tabular*}
  \caption{Performance metrics, parameter estimates and approximate T-test results for the twin study example. Standard errors for the fixed effects parameter estimates are given in brackets alongside the corresponding estimates. For each model parameter, a T-statistic, direct-SW degrees of freedom estimate, and p-value are provided, corresponding to the approximate t-test for a non-zero effect. P-values that are significant at the $5\%$ level are indicated using a $^*$ symbol.}\label{AceTab}
\end{table*}
\clearpage
}

Although the methods presented in this work are well suited for the desired application, we stress that there are many situations where current software packages will likely provide superior performance. One such situation can be found by observing that the Fisher Scoring method requires the storage and inversion of a matrix of dimensions $(q \times q)$. This is problematic, both in terms of memory and computation accuracy, for designs which involve very large numbers of random effects, grouping factors or factor levels. While such applications have not been considered extensively in this work, we note that many of the expressions provided in this document may benefit from combination with sparse matrix methodology to overcome this issue. We suggest that the potential for improved computation time via the combination of the Fisher Scoring approaches described in this paper with sparse matrix methodology may also be the subject of future research.

\section{Declarations}

\subsection{Funding}

This work was supported by the Li Ka Shing Centre for Health Information and Discovery and NIH grant [R01EB026859] (TMS, TN) and the Wellcome Trust award [100309/Z/12/Z] (TN).

\subsection{Conflicts of interest/Competing interests}

Not applicable.

\subsection{Availability of data and material}

Data were provided in part by the Human Connectome Project, WU-Minn Consortium (Principal Investigators: David Van Essen and Kamil Ugurbil; 1U54MH091657) funded by the 16 NIH Institutes and Centers that support the NIH Blueprint for Neuroscience Research; and by the McDonnell Center for Systems Neuroscience at Washington University. 

The longitudinal evaluation of school change and performance (LESCP) dataset employed in Sections \ref{SATexample} and \ref{SATres} is publicly available and can be found, for example, as one of the example datasets included in the Heirachical Linear Models (HLM) software package (\cite{Raudenbush2002}).

\subsection{Code availability}

All code used to generate the results of the simulations and real data examples presented in Sections \ref{SimMeth}-\ref{AceRes} is publicly available and can be found in the below GitHub repository:\\
\\
\noindent
\underline{\href{https://github.com/TomMaullin/LMMPaper}{https://github.com/TomMaullin/LMMPaper}}\\
\\
Included in this repository are detailed notebooks demonstrating how the code may be executed in order to reproduce the results presented in this work.

\subsection{Authors' contributions}

Not applicable.

\section{Appendix} 
\subsection{Score vectors}\label{derivAppendix}

In this appendix, we provide full derivations for the derivatives  (\ref{FSderiv3}) and (\ref{FFSderiv1}). For derivatives (\ref{FSderiv1}) and (\ref{FSderiv2}), we note that the derivations for the multi-factor LMM are identical to those given for the single-factor setting in \citet{Demidenko:2013mfx}. As a result, we refer the interested reader to this source for proofs. To obtain the derivatives (\ref{FSderiv3}) and (\ref{FFSderiv1}), two lemmas, lemma \ref{lemma1} and lemma \ref{lemma2}, are required. Before stating the two lemmas, we briefly make explicit the definition of derivative employed throughout this paper.

Throughout this work, for arbitrary matrices $A$ and $B$ of dimensions $(p \times q)$ and $(m \times n)$ respectively, the following definition of matrix (partial) derivative has been implicitly employed:

\begin{equation}\nonumber
    \frac{\partial \text{vec}(A)}{\partial \text{vec}(B)} = \begin{bmatrix}\frac{\partial A_{[1,1]}}{\partial B_{[1,1]}} & \hdots & \frac{\partial A_{[p,1]}}{\partial B_{[1,1]}} & \hdots & \frac{\partial A_{[p,q]}}{\partial B_{[1,1]}} \\
    \vdots & \ddots & \vdots & \ddots & \vdots \\
    \frac{\partial A_{[1,1]}}{\partial B_{[m,1]}} & \hdots & \frac{\partial A_{[p,1]}}{\partial B_{[m,1]}} & \hdots & \frac{\partial A_{[p,q]}}{\partial B_{[m,1]}} \\
    \vdots & \ddots & \vdots & \ddots & \vdots \\
    \frac{\partial A_{[1,1]}}{\partial B_{[m,n]}} & \hdots & \frac{\partial A_{[p,1]}}{\partial B_{[m,n]}} & \hdots & \frac{\partial A_{[p,q]}}{\partial B_{[m,n]}} \\
    \end{bmatrix}, 
\end{equation}
with the equivalent definition holding for the total derivative. In addition, for a scalar value, $a$, and matrix, $B$, defined as above, we employ the below matrix (partial) derivative definition:

\begin{equation}\nonumber
    \frac{\partial{a}}{\partial B} = \begin{bmatrix}
     \frac{\partial a}{\partial B_{[1,1]}} & \hdots & \frac{\partial a}{\partial B_{[1,n]}} \\
     \vdots & \ddots & \vdots \\
     \frac{\partial a}{\partial B_{[m,1]}} & \hdots & \frac{\partial a}{\partial B_{[m,n]}} \\
    \end{bmatrix},
\end{equation}
with, again, the equivalent definition holding for the total derivative. The above notation can be useful since by definition, the two derivatives defined above satisfy the below relation:
\begin{equation}\nonumber
    \text{vec}\bigg(\frac{\partial{a}}{\partial B}\bigg)= \frac{\partial a}{\partial \text{vec}(B)}.
\end{equation}

The above matrix derivative definitions are commonly employed in the mathematical and statistical literature, but are not universal. For example, sources such as \cite{Neudecker1983} and \citet{Magnus1999} define $\partial \text{vec}(A)/\partial \text{vec}(B)$ as the transpose of the definition supplied above. This discrepancy between definitions is noteworthy as some of the referenced results used in the following proofs are the transpose of those found in the original source. For a full account and in-depth discussion of the different definitions of matrix derivatives used in the literature, we recommend the work of \citet{Turkington2013}.

\begin{lemma}\label{lemma1}
Let $g$ be a column vector, $A$ be a square matrix and $\{B_s\}$ be a set of arbitrary matrices of equal size. Let $K$ be an unstructured matrix which none of $g$, $A$ or any of the $\{B_s\}$ depend on. Further, assume $A,\{B_s\}, g$ and $K$ can be multiplied as required. The below matrix derivative expression now holds;
\begin{equation}\label{lem1result}
\begin{aligned}[b]
& \frac{\partial}{\partial K} \bigg[ g'(A+\sum_t B_tKB_t')^{-1}g \bigg]=\\
& -\sum_s B_s'\big(A'+\sum_t B_tKB_t'\big)^{-1}gg'\big(A'+\sum_t B_tKB_t'\big)^{-1}B_s.
\end{aligned}
\end{equation}
\end{lemma}
\begin{proof} To begin, denote $M=A+\sum_sB_sKB_s'$. For clarity and convenience, in this proof, we shall briefly depart from the usual notation of $H_{[i,j]}$ for the $(i,j)^{th}$ element of an arbitrary matrix $H$ and instead employ Ricci calculus notation, in which the $(i,j)^{th}$ element of $H$ is given by $H^i_j$ and all summations are made implicit. The left hand side of (\ref{lem1result}) is now given, using Ricci calculus notation, as;
\begin{equation}\label{lem11}
    \frac{\partial}{\partial K^m_n}\bigg[(g')_i(M^{-1})^i_jg^j\bigg] = (g')_ig^j\frac{\partial}{\partial K^m_n}\big[M^{-1}\big]^i_j.
\end{equation}
\noindent
We employ a result from \citet{giles2008:collect} which states, in Ricci calculus notation;
\begin{equation}\nonumber\frac{\partial [F^{-1}]^i_j}{\partial x}=-(F^{-1})^i_q\frac{\partial F^q_p}{\partial x}(F^{-1})^p_j.
\end{equation}
Applying this result to (\ref{lem11}), it can be seen that the right hand side of (\ref{lem11}) is equal to: 
\begin{equation}\nonumber
  -(g')_ig^j(M^{-1})^i_q\frac{\partial M^q_p}{\partial K^m_n}(M^{-1})^p_j.
\end{equation}
Evaluating the derivative in the above expression and rearranging the terms gives the below:
\begin{equation}\nonumber
\begin{aligned}
     & -(B_s')^n_p(M^{-1})^p_jg^j(g')_i(M^{-1})^i_qB^q_{sm}\\
     & =-[B_s'M^{-1}gg'M^{-1}B_s]^n_m.\\
 \end{aligned}
\end{equation}

\noindent
By applying a transposition to the matrices in the above expression, in order to interchange the indices $m$ and $n$, the result follows.\qed\end{proof}

\begin{lemma}\label{lemma2} Let $A, \{B_s\}$ and $K$ be defined as in Lemma \ref{lemma1}. Then the following is true: 

\begin{equation}\label{lemma2result}
\frac{\partial}{\partial K} \log|A+\sum_t B_tKB_t'| = \sum_s B_s'(A+\sum_t B_tK'B_t')^{-1}B_s.
\end{equation}

\end{lemma}

\begin{proof} As in the proof of Lemma \ref{lemma1}, the notation of $M=A+\sum_t B_tK'B_t'$ will be adopted and the proof will be given in Ricci calculus notation. Expanding the derivative with respect to $K^m_n$ using the partial derivatives of the elements of $M$ gives the following:
\begin{equation}\label{lem21}
    \frac{\partial}{\partial K^m_n} \log|M| = \frac{\partial\log|M|}{\partial M^i_j}\frac{\partial M^i_j}{\partial K^m_n}.    
\end{equation}
\noindent
We now make use of the following well-known result, which can be found in \citet{giles2008:collect}.

\begin{equation}\label{lem22}
    \frac{\partial \log|X|}{\partial X^{i}_j}=(X'^{-1})^i_j.
\end{equation}
\noindent
Inserting (\ref{lem22}) into (\ref{lem21}) and noting that $A^i_j$ does not depend on $K^m_n$, yields:
\begin{equation}\nonumber\label{lem23}
    \frac{\partial}{\partial K^m_n} \log|M| = (M'^{-1})^i_j\frac{\partial (B_sK(B_s)')^i_j}{\partial K^m_n}.    
\end{equation}
\noindent
Evaluating the derivative on the right hand side of the above now returns:
\begin{equation}\nonumber\label{lem24}
    \frac{\partial}{\partial K^m_n} \log|M| = (M'^{-1})^i_j(B_s)_m^i(B_s')^{n}_j = (B'_s)_i^m(M'^{-1})^i_j (B_s)^j_{n} .
\end{equation}
Combining the summation terms in the above expression yields the following.
\begin{equation}\nonumber\label{lem25}
    \frac{\partial}{\partial K^m_n} \log|M| = (B'_s M'^{-1} B_s)^m_n .
\end{equation}
Finally, converting the above into matrix notation now gives (\ref{lemma2result}) as desired. \qed
\end{proof}
We are now in a position to derive the partial derivative matrix $\frac{\partial l}{\partial D_k}$, the matrix derivative with respect to $D_k$ which does not account for the equal elements of $D_k$ induced by symmetry.

\begin{theorem}\label{DmatTheorem} The partial derivative matrix $\frac{\partial l}{\partial D_k}$ is given by the following.

\begin{equation}\label{dldDk}
 \frac{\partial l(\theta)}{\partial D_k} = \frac{1}{2}\sum_{j=1}^{l_k}Z_{(k,j)}'V^{-1}\bigg(\frac{ee'}{\sigma^2}-V\bigg)V^{-1}Z_{(k,j)}.
\end{equation}

\end{theorem}
\begin{proof}
To derive (\ref{dldDk}), we use the expression for the log-likelihood of the LMM, given by (\ref{llh}). As the first term inside the brackets of (\ref{llh}) does not depend on $D_k$, we need only consider the second and third term for differentiation. We now note that, by (\ref{Zprop}) and the block diagonal structure of $D$, it can be seen that:

\begin{equation}\label{vdef}
    V= I+ZDZ'=I+\sum_{k=1}^{r}\sum_{j=1}^{l_r}Z_{(k,j)}D_kZ_{(k,j)}'.
\end{equation}
By substituting (\ref{vdef}) into the second term of (\ref{llh}), it can be seen that:
\begin{equation}\nonumber
\begin{aligned}
& \sigma^{-2}e'V^{-1}e =\sigma^{-2}e'\bigg(I+\sum_{k=1}^{r}\sum_{j=1}^{l_r}Z_{(k,j)}D_kZ_{(k,j)}'\bigg)^{-1}e.
\end{aligned}
\end{equation}
By applying the result of Lemma \ref{lemma1}, it can now be seen that the partial derivative matrix of the above, taken with respect to $D_k$, can be evaluated to:
\begin{equation}\nonumber\frac{\partial}{\partial D_k}\bigg[\sigma^{-2}e'V^{-1}e\bigg]=\sigma^{-2}\sum_{j=1}^{l_k}Z_{(k,j)}'V^{-1}ee'V^{-1}Z_{(k,j)}.
\end{equation}
Similarly, by substituting (\ref{vdef}) into the third term of (\ref{llh}), it can be seen that:
\begin{equation}\nonumber\log|V|=\log\big|I+\sum_{k=1}^{r}\sum_{j=1}^{l_r}Z_{(k,j)}D_kZ_{(k,j)}'\big|.\end{equation}
By now applying the result of Lemma \ref{lemma2}, the partial derivative matrix of the above expression, taken with respect to $D_k$, can be seen to be given by:
\begin{equation}\nonumber
    \frac{\partial}{\partial D_k}\big[\log|V|\big]=\sum_{j=1}^{l_k}Z_{(k,j)}'V^{-1}Z_{(k,j)}.
\end{equation}
Combining the previous two derivative expressions, we obtain (\ref{dldDk}) as desired.\qed
\end{proof}
Through applying the vectorisation operator to (\ref{dldDk}), the following corollary is obtained. This is the result stated by (\ref{FFSderiv1}) in Section \ref{FFSsection}.

\begin{corollary}\label{partialDerivCorollary}
The partial derivative vector of the log-likelihood with respect to vec$(D_k)$ is given as:
\begin{equation}\nonumber
 \frac{\partial l(\theta)}{\partial \text{vec}(D_k)} = \frac{1}{2}\text{vec}\bigg(\sum_{j=1}^{l_k}Z_{(k,j)}'V^{-1}\bigg(\frac{ee'}{\sigma^2}-V\bigg)V^{-1}Z_{(k,j)}\bigg).
\end{equation}
\end{corollary}

\noindent
Using Theorem 5.12 of \citet{Turkington2013}, which states that, in our notation:
\begin{equation}\nonumber
    \frac{dl(\theta)}{d\text{vec}(D_k)}=\mathcal{D}_{q_k}\mathcal{D}_{q_k}'\frac{\partial l(\theta)}{\partial \text{vec}(D_k)},
\end{equation}
\noindent
the following corollary is now obtained.

\begin{corollary}\label{vecTotalDerivCor}
The total derivative vector of the log-likelihood with respect to vec($D_k$) is given as:
\begin{equation}\nonumber
\begin{aligned}
 &\frac{d l(\theta)}{d \text{vec}(D_k)} = \\
 &\frac{1}{2}\mathcal{D}_{q_k}\mathcal{D}'_{q_k}\text{vec}\bigg(\sum_{j=1}^{l_k}Z_{(k,j)}'V^{-1}\bigg(\frac{ee'}{\sigma^2}-V\bigg)V^{-1}Z_{(k,j)}\bigg).
\end{aligned}
\end{equation}
\end{corollary}
\noindent
Finally, by noting that the vectorisation and half-vectorisation operators satisfy the following relationship:
\begin{equation}\nonumber
    \text{vec}(D_k)=\mathcal{D}^+_{q_k}\text{vech}(D_k),
\end{equation}
the following corollary is obtained. This is the result stated by (\ref{FSderiv3}) in Section \ref{FSsection}.

\begin{corollary}\label{vechTotalDerivCor}The total derivative vector of the log-likelihood with respect to vech$(D_k)$ is given as:
\begin{equation}\nonumber
 \frac{d l(\theta)}{d \text{vech}(D_k)} = \frac{1}{2}\mathcal{D}'_{q_k}\text{vec}\bigg(\sum_{j=1}^{l_k}Z_{(k,j)}'V^{-1}\bigg(\frac{ee'}{\sigma^2}-V\bigg)V^{-1}Z_{(k,j)}\bigg).
\end{equation}
\end{corollary}
This result diverges slightly from the proof given for the single-factor LMM in \citet{Demidenko:2013mfx}, which uses the inverse duplication matrix $\mathcal{D}^+_{q_k}$ in the place of the transposed duplication matrix,  $\mathcal{D}'_{q_k}$, in the above. The authors believe that the derivative given in \citet{Demidenko:2013mfx} is, in fact, vech$\big(\frac{\partial l(\theta)}{\partial\text{vec}(D_k)}\big)$, and not $\frac{d l(\theta)}{d\text{vech}(D_k)}$ as stated. It should be noted, however, that this is a minor amendment and has little impact on the performance of the algorithm in practice.

\subsection{Fisher Information matrix}\label{InfoMatAppendix}

In this appendix, we derive the Fisher Information matrix for the full representation of the parameter vector, $\theta^f$. The derivation of the elements $\mathcal{I}^f_\beta$, $\mathcal{I}^f_{\beta,\sigma^2}$ and $\mathcal{I}^f_{\sigma^2}$ for the multi-factor LMM is identical to that used for the single factor LMM given in \citet{Demidenko:2013mfx} and, therefore, will not be repeated here. Provided here are only derivations of components of the Fisher Information matrix which relate to $D_k$, for some factor $f_k$. To derive these results, we will follow a similar argument to that used by \citet{Demidenko:2013mfx} for the single-factor setting and will begin by noting that the Fisher Information matrix elements, $\{\mathcal{I}_{a,b}^f\}_{a,b \in \{\beta, \sigma^2, \text{vec}(D_1),\hdots,\text{vec}(D_r)\}}$, can be expressed by the following formula:

\begin{equation}\nonumber
    \mathcal{I}_{a,b}^f = \text{cov}\bigg(\frac{\partial l(\theta^f)}{\partial a},\frac{\partial l(\theta^f)}{\partial b}\bigg).
\end{equation}
Theorems \ref{covdldDkdBeta} and \ref{covdldDkdSigma2} provide derivation of equation (\ref{FFSFI1}) and Theorem \ref{covdldDk1Dk2} provides derivation of equation (\ref{FFSFI2}). Following this, Corollaries \ref{covdldDkdBetaCor}-\ref{covdldDk1Dk2Cor} detail the derivation of equations (\ref{FSFI2}) and (\ref{FSFI3}).

\begin{theorem}\label{covdldDkdBeta} For any arbitrary integer $k$ between $1$ and $r$, the covariance of the partial derivatives of $l(\theta^f)$ with respect to $\beta$ and vec$(D_k)$ is given by:

\begin{equation}\nonumber
    \mathcal{I}^f_{\beta,\text{vec}(D_k)}=\text{cov}\bigg(\frac{\partial l(\theta^f)}{\partial\beta},\frac{\partial l(\theta^f)}{\partial \text{vec}(D_k)}\bigg)=\mathbf{0}_{p,q_k^2}.
\end{equation}

\end{theorem}
\begin{proof} First, let $u$ and $T_{(k,j)}$ denote the following quantities:
\begin{equation}\label{utdef}
    u = \sigma^{-1}V^{-\frac{1}{2}}e, \hspace{0.5cm} T_{(k,j)}=Z_{(k,j)}'V^{-\frac{1}{2}}.
\end{equation}
As $e\sim N(0, \sigma^2 V)$, it follows trivially that $u\sim N(0,I_n)$. Let $c$ be an arbitrary column vector of length $q_k$. Noting that the derivative of the log-likelihood function has mean zero, the below can be seen to be true:
\begin{equation}\nonumber\text{cov}\bigg(\frac{\partial l(\theta | y)}{\partial \beta},c'\frac{\partial l(\theta | y)}{\partial D_k}c \bigg) = \mathbb{E}\bigg[ \frac{\partial l(\theta | y)}{\partial \beta}c'\frac{\partial l(\theta | y)}{\partial D_k}c\bigg]. 
\end{equation}
By rewriting in terms of $u$ and $T_{(k,j)}$, and noting that $\mathbb{E}[u]=0$, the right hand side of the above can be seen to simplify to:
\begin{equation}\nonumber \mathbb{E}\bigg[\sigma^{-1}XV^{-\frac{1}{2}}u c'\bigg(\frac{1}{2}\sum_{j=1}^{l_k}(T_{(k,j)}u)(T_{(k,j)}u)'\bigg)c\bigg]
\end{equation}
\begin{equation}\nonumber
 =\frac{1}{2}\sigma^{-1}XV^{-\frac{1}{2}}\sum_{j=1}^{l_k}\mathbb{E}\bigg[u c'T_{(k,j)}uu'\bigg]T_{(k,j)}'c=\mathbf{0}_{p,q_k^2}.
\end{equation}
That the above is equal to a matrix of zeros, follows directly from noting that the third moment of the Normal distribution is $0$. The result of the theorem now follows.\qed
\end{proof}

\begin{theorem}\label{covdldDkdSigma2} For any arbitrary integer $k$ between $1$ and $r$, the covariance of the partial derivatives of $l(\theta^f)$ with respect to $\sigma^2$ and vec$(D_k)$ is given by:
\begin{equation}\nonumber
\begin{aligned}[b]
    & \mathcal{I}^f_{\sigma^2, \text{vec}(D_k)}=\text{cov}\bigg(\frac{\partial l(\theta^f)}{\partial\sigma^2},\frac{\partial l(\theta^f)}{\partial \text{vec}(D_k)}\bigg)= \\
    & \frac{1}{2\sigma^2}\text{vec}'\bigg(\sum_{j=1}^{l_k}Z'_{(k,j)}V^{-1}Z_{(k,j)}\bigg).
\end{aligned}
\end{equation}

\end{theorem}
\begin{proof} To begin, $u$ and $T_{(k,j)}$ are defined as in (\ref{utdef}). Rewriting the covariance in Theorem \ref{covdldDkdSigma2} in terms of (\ref{utdef}) and removing constant terms gives:
\begin{equation}\nonumber
\frac{1}{4\sigma^2}\text{cov}\bigg(u'u,\text{vec}\bigg(\sum_{j=1}^{l_k}(T_{(k,j)}u)(T_{(k,j)}u)'\bigg)\bigg).
\end{equation}
Noting that the Kronecker product satisfies the property vec$(aa')=a\otimes a$, for all column vectors $a$, we obtain that the above is equal to:
\begin{equation}\nonumber\frac{1}{4\sigma^2}\text{cov}\bigg(u'u,\sum_{j=1}^{l_k}\bigg[(T_{(k,j)}u)\otimes (T_{(k,j)}u)\bigg]\bigg).
\end{equation}
We now note that the Kronecker product satisfies vec$(ABC)=(C'\otimes A)$vec$(B)$ for arbitrary matrices $A$, $B$ and $C$ of appropriate dimensions. Utilizing this and applying the mixed product property of the Kronecker product to the above expression yeilds:
\begin{equation}\nonumber
\frac{1}{4\sigma^2}\text{cov}\bigg(\text{vec}'(I_n)(u \otimes u),\sum_{j=1}^{l_k}\bigg[(T_{(k,j)}\otimes T_{(k,j)})(u\otimes u)\bigg]\bigg).
\end{equation}
By moving constant values outside of the covariance function using standard results the above now becomes:
\begin{equation}\nonumber\frac{1}{4\sigma^2}\text{vec}'(I_n)\text{cov}(u \otimes u)\sum_{j=1}^{l_k}\bigg[(T_{(k,j)}\otimes T_{(k,j)})\bigg]^{'}.
\end{equation}
Noting that $u \sim N(0,I_n)$ we now employ a result from \citet{Magnus1986:dup} which states that $\text{cov}(u \otimes u)=2N_n$. Substituting this result into the previous expression gives the below:
\begin{equation}\nonumber\frac{1}{2\sigma^2}\text{vec}'(I_n)N_n\sum_{j=1}^{l_k}\bigg[(T_{(k,j)}\otimes T_{(k,j)})\bigg]'.
\end{equation}
We now note a further result given in \citet{Magnus1986:dup}; the matrix $N_n$ satisfies $(A \otimes A)N_k = N_n(A \otimes A)$ for all $A$, $n$ and $k$  such that the resulting matrix multiplications are well defined. Applying this result to the above expression, the below can be obtained:
\begin{equation}\nonumber\frac{1}{2\sigma^2}\text{vec}'(I_n)\sum_{j=1}^{l_k}\bigg[(T_{(k,j)}\otimes T_{(k,j)})\bigg]'N_{q_k}.
\end{equation}
Finally, again using the relationship vec$(ABC)=(C'\otimes A)$vec$(B)$ which is satisfied by the Kronecker product for arbitrary matrices $A$, $B$ and $C$ of appropriate dimensions, the above can be seen to reduce to:
\begin{equation}\nonumber\frac{1}{2\sigma^2}\text{vec}'\bigg(\sum_{j=1}^{l_k}T_{(k,j)}T_{(k,j)}'\bigg)N_{q_k}.
\end{equation}
Using the definition of $T_{(k,j)}$, given in (\ref{utdef}), and noting that the matrix $N_{q_k}$ satisfies the relationship $\text{vec}'(A)N_{q_k}=\text{vec}'(A)$ for all appropriately sized symmetric matrices $A$, the result now follows. \qed

\end{proof}
\begin{theorem}\label{covdldDk1Dk2} For any arbitrary integers $k_1$ and $k_2$ between $1$ and $r$, the covariance of the partial derivatives of $l(\theta^f)$ with respect to vec$(D_{k_1})$ and vec$(D_{k_2})$ is given by:

\begin{equation}\nonumber
  \begin{aligned}[b]
    &\mathcal{I}^f_{\text{vec}(D_{k_1}),\text{vec}(D_{k_2})}=\text{cov}\bigg(\frac{\partial l(\theta^f)}{\partial \text{vec}(D_{k_1})},\frac{\partial l(\theta^f)}{\partial \text{vec}(D_{k_2})}\bigg)=\\ &\frac{1}{2}N_{q_{k_1}}\sum_{j=1}^{l_{k_2}}\sum_{i=1}^{l_{k_1}}(Z'_{(k_1,i)}V^{-1}Z_{(k_2,j)}\otimes Z'_{(k_1,i)}V^{-1}Z_{(k_2,j)}).
    \end{aligned}
\end{equation}

\end{theorem}  
\begin{proof}
By Corollary \ref{partialDerivCorollary}, and using the notation introduced in (\ref{utdef}), it can be seen that the partial derivative vector with respect to vec$(D_k)$ is given by:
\begin{equation}\nonumber\frac{\partial l(\theta | y)}{\partial \text{vec}(D_k)}=\text{vec}\bigg(\frac{1}{2}\sum_{j=1}^{l_k}(T_{(k,j)}u)(T_{(k,j)}u)'-\frac{1}{2}\sum_{j=1}^{l_k}T_{(k,j)}T_{(k,j)}'\bigg).\end{equation}
By substituting the above into the covariance expression appearing in Theorem 4 and removing constant terms, the below is obtained:
\begin{equation}\nonumber
\begin{aligned}[b]
&\frac{1}{4}\text{cov}\bigg(\sum_{j=1}^{l_{k_1}}\text{vec}\bigg[(T_{(k_1,j)}u)(T_{(k_1,j)}u)'\bigg],\\
& \hspace{1.8cm}\sum_{j=1}^{l_{k_2}}\text{vec}\bigg[(T_{(k_2,j)}u)(T_{(k_2,j)}u)'\bigg]\bigg).
\end{aligned}
\end{equation}
Noting that for any column vector $a$, the Kronecker product satisfies the property vec$(aa')=a\otimes a$, the above can be seen to be equal to:
\begin{equation}\nonumber
\begin{aligned}[b]
& \frac{1}{4}\text{cov}\bigg(\sum_{j=1}^{l_{k_1}}\bigg[(T_{(k_1,j)}u)\otimes (T_{(k_1,j)}u)\bigg],\\
& \hspace{1.8cm}\sum_{j=1}^{l_{k_2}}\bigg[(T_{(k_2,j)}u)\otimes (T_{(k_2,j)}u)\bigg]\bigg).
\end{aligned}\end{equation}
Using the mixed product property of the Kronecker product and moving constant matrices outside of the covariance function, using standard results, gives:
\begin{equation}\nonumber\frac{1}{4}\sum_{j=1}^{l_{k_1}}\bigg[(T_{(k_1,j)}\otimes T_{(k_1,j)})\bigg]\text{cov}(u\otimes u)\sum_{j=1}^{l_{k_2}}\bigg[(T_{(k_2,j)}\otimes T_{(k_2,j)})\bigg]'.
\end{equation}
Noting that $u\sim N(0,I_n)$ and using a result from \cite{Magnus1986:dup} which states that cov$(u\otimes u)=2N_n$, yields:
\begin{equation}\nonumber\frac{1}{4}\sum_{j=1}^{l_{k_1}}\bigg[(T_{(k_1,j)}\otimes T_{(k_1,j)})\bigg]2N_n\sum_{j=1}^{l_{k_2}}\bigg[(T_{(k_2,j)}'\otimes T_{(k_2,j)}')\bigg].
\end{equation}
We now utilize the fact that, as noted by \cite{Magnus1986:dup}, the matrix $N_n$ satisfies $(A \otimes A)N_k = N_n(A \otimes A)$ for all $A$, $n$ and $k$  such that the resulting matrix multiplications are well defined. Using this and the mixed product property of the Kronecker product the above can now be seen to be equal to:
\begin{equation}\nonumber\frac{1}{2}N_{q_{k_1}}\sum_{j=1}^{l_{k_2}}\sum_{i=1}^{l_{k_1}}\bigg[(T_{(k_1,i)}T_{(k_2,j)}')\otimes (T_{(k_1,i)}T_{(k_2,j)}')\bigg].
\end{equation}
From the definition of $T_{(k,j)}$, it can be seen that the above is equal to the result of Theorem 4. \qed
\end{proof}

\noindent
We now turn attention to the derivation of (\ref{FSFI2}) and (\ref{FSFI3}). As in the proofs of Corollaries \ref{vecTotalDerivCor} and \ref{vechTotalDerivCor} of Appendix \ref{derivAppendix}, we begin by noting that:

\begin{equation}\nonumber
    \begin{aligned}[b]
    \frac{dl(\theta)}{d\text{vech}(D_k)}& = \mathcal{D}^+_{q_k}\frac{dl(\theta)}{d\text{vec}(D_k)}\\
    & =\mathcal{D}^+_{q_k}\mathcal{D}_{q_k}\mathcal{D}_{q_k}'\frac{\partial l(\theta)}{\partial \text{vec}(D_k)}\\
    & =\mathcal{D}_{q_k}'\frac{\partial l(\theta)}{\partial \text{vec}(D_k)},
    \end{aligned}
\end{equation}
where the first equality follows from the definition of the duplication matrix and the second equality follows from Theorem 5.12 of \citet{Turkington2013}. Applying the above identity to Theorems \ref{covdldDkdBeta}, \ref{covdldDkdSigma2} and \ref{covdldDk1Dk2} and moving the matrix $\mathcal{D}_{q_k}'$ outside the covariance function in each, leads to the following three corollaries which, when taken in combination, provide equations (\ref{FSFI2}) and (\ref{FSFI3}).

\begin{corollary}\label{covdldDkdBetaCor} For any arbitrary integer $k$ between $1$ and $r$, the covariance of the total derivatives of $l(\theta^h)$ with respect to $\beta$ and vech$(D_k)$ is given by:

\begin{equation}\nonumber
    \mathcal{I}^h_{\beta,\text{vech}(D_k)}=\text{cov}\bigg(\frac{d l(\theta^h)}{d\beta},\frac{d l(\theta^h)}{d \text{vech}(D_k)}\bigg)=\mathbf{0}_{p,q_k(q_k+1)/2}.
\end{equation}

\end{corollary}

\begin{corollary}\label{covdldDkdSigma2Cor} For any arbitrary integer $k$ between $1$ and $r$, the covariance of the total derivatives of $l(\theta^h)$ with respect to $\sigma^2$ and vech$(D_k)$ is given by:

\begin{equation}\nonumber
\begin{aligned}[b]
    & \mathcal{I}^h_{\sigma^2, \text{vech}(D_k)}=\text{cov}\bigg(\frac{d l(\theta^h)}{d\sigma^2},\frac{d l(\theta^h)}{d \text{vech}(D_k)}\bigg)= \\
    & \frac{1}{2\sigma^2}\text{vec}'\bigg(\sum_{j=1}^{l_k}Z'_{(k,j)}V^{-1}Z_{(k,j)}\bigg)\mathcal{D}_{q_k}.
\end{aligned}
\end{equation}
\end{corollary}
\begin{corollary}\label{covdldDk1Dk2Cor} For any arbitrary integers $k_1$ and $k_2$ between $1$ and $r$, the covariance of the total derivatives of $l(\theta^h)$ with respect to vech$(D_{k_1})$ and vech$(D_{k_2})$ is given by:

\begin{equation}\nonumber
  \begin{aligned}[b]
    &\mathcal{I}^h_{\text{vech}(D_{k_1}),\text{vech}(D_{k_2})}=\text{cov}\bigg(\frac{d l(\theta^h)}{d \text{vech}(D_{k_1})},\frac{d l(\theta^h)}{d \text{vech}(D_{k_2})}\bigg)=\\ &\frac{1}{2}\mathcal{D}'_{q_{k_1}}\sum_{j=1}^{l_{k_2}}\sum_{i=1}^{l_{k_1}}(Z'_{(k_1,i)}V^{-1}Z_{(k_2,j)}\otimes Z'_{(k_1,i)}V^{-1}Z_{(k_2,j)})\mathcal{D}_{q_{k_2}}.
    \end{aligned}
\end{equation}
\end{corollary}
We note that the results of Corollaries \ref{covdldDkdBetaCor}, \ref{covdldDkdSigma2Cor} and \ref{covdldDk1Dk2Cor} do not contain the matrix $N_{q_k}$, which appears in the corresponding theorems (Theorems \ref{covdldDkdBeta}, \ref{covdldDkdSigma2} and \ref{covdldDk1Dk2}). This is due to another result of \citet{Magnus1986:dup}, which states that $\mathcal{D}'_{k}N_k=\mathcal{D}'_{k}$ for any integer $k$. This concludes the derivations of Fisher Information matrix expressions given in Sections \ref{FSsection}-\ref{FSFSsection}.

\subsection{Full Fisher Scoring}\label{FullAppendix}

In this appendix, we prove that equations (\ref{FFS2}) and (\ref{vecDUpdate}) are valid Fisher Scoring rules of the form given by equation (\ref{FFS}). To achieve this, we first require the following lemma.

\begin{lemma}\label{FFSlem}
Let $n, p$ and $k$ be arbitrary positive integers with $p<n$. Let  $\tilde{N}$ be of dimension $(n \times n)$ with rank$(\tilde{N})<n$, $\tilde{K}$ be of dimension $(n \times p)$ with rank$(\tilde{K})=p$, $C$ be of dimension $(n \times n)$ with rank$(\tilde{N})=n$ and $\delta $ be of size $(n \times k)$. If the following properties are satisfied:
\begin{enumerate}
    \item $\tilde{K}\tilde{K}^+=\tilde{N}$, where $\tilde{K}^+=(\tilde{K}'\tilde{K})^{-1}\tilde{K}'$ is the generalized inverse of $\tilde{K}$.
    \item $\tilde{N}$ is symmetric and idempotent.
    \item $\tilde{N}C=C\tilde{N}=\tilde{N}C\tilde{N}$.
    \item $\tilde{N}\delta = \delta.$
\end{enumerate}
\noindent
then it follows that:
\begin{equation}\nonumber(\tilde{N}C)^+\delta = C^{-1}\delta.
\end{equation}
\end{lemma}

\begin{proof}
In proving the above lemma, we make use of the following result given by \citet{barnett1990matrices};
\begin{itemize}
    \item[$\bullet$] If $A$ is a matrix of dimension $(m \times l)$ with rank$(A)=l<m$, and $B$ is a matrix of dimension $(l \times m)$ with rank$(B)=l$, then;
    \begin{equation}\label{BarnettResult}(AB)^+=B'(BB')^{-1}(A'A)^{-1}A'.
    \end{equation}
\end{itemize}
By property 1, we have that:
\begin{equation}\nonumber(\tilde{N}C)^+=(\tilde{K}\tilde{K}^+C)^+.
\end{equation}
By applying (\ref{BarnettResult}) to the above, with $A=\tilde{K}$ and $B=\tilde{K}^+C$, and simplifying the result, the following is obtained:
\begin{equation}\label{BarnettResultApplied}(\tilde{K}\tilde{K}^+C)^+=C'\tilde{K}^{+'}(\tilde{K}^+CC'\tilde{K}^{+'})^{-1}\tilde{K}^+.
\end{equation}
We now consider the inversion in the center of the above expression, i.e. the inversion of $(\tilde{K}^+CC'\tilde{K}^{+'})$. This inversion is given by;
\begin{equation}\nonumber
(\tilde{K}^+CC'\tilde{K}^{+'})^{-1}=\tilde{K}'C^{-1'}C^{-1}\tilde{K}.
\end{equation}
We will demonstrate this by showing that the product of the matrix inside the inversion on the left hand side and the matrix on the right hand side is the identity matrix:
\begin{equation}\nonumber(\tilde{K}^+CC'\tilde{K}^{+'})(\tilde{K}'C^{-1'}C^{-1}\tilde{K}) = \tilde{K}^+CC'\tilde{K}^{+'}\tilde{K}'C^{-1'}C^{-1}\tilde{K}.
\end{equation}
Using properties $1$, $2$ and $3$ we see that the above is equal to:
\begin{equation}\nonumber \tilde{K}^+\tilde{N}CC'C^{-1'}C^{-1}\tilde{K}.
\end{equation}
Simplifying and applying property $1$ now yields that the above is equal the below:
\begin{equation}\nonumber \tilde{K}^+\tilde{N}\tilde{K} = \tilde{K}^+\tilde{K}= I,
\end{equation}
where $I$ is the identity matrix as required. Returning to (\ref{BarnettResultApplied}) and applying properties 1 and 2, we now have that:
\begin{equation}\nonumber(\tilde{N}C)^+ = C'\tilde{N}C^{-1'}C^{-1}\tilde{N}.
\end{equation}
By applying properties 2 and 3, the above can now be reduced to:
\begin{equation}\label{NCNCN}
(\tilde{N}C)^+=\tilde{N}C^{-1}\tilde{N}.
\end{equation}
We now note that, by property $3$, $\tilde{N}C=C\tilde{N}$. By multiplying the left and right side of this identity by $C^{-1}$ it can be seen that:
\begin{equation}\nonumber C^{-1}\tilde{N}=\tilde{N}C^{-1}.
\end{equation}
Multiplying the right hand side of the above by $\tilde{N}$ and applying property 2 we obtain:
\begin{equation}\nonumber C^{-1}\tilde{N}=\tilde{N}C^{-1}\tilde{N}.
\end{equation}
Substituting the above into (\ref{NCNCN}) gives:
\begin{equation}\nonumber (\tilde{N}C)^+ = C^{-1}\tilde{N}.
\end{equation}
The result of Lemma \ref{FFSlem} now follows by multiplying the right-hand side by $\delta$ and applying property 4. \qed

\end{proof}

\begin{theorem}\label{FFSthm1}
The below two update rules are equivalent, where $F(\theta^f_s)$ is the matrix defined in Section \ref{FFSsection} and $\mathcal{I}(\theta^f_s)$ is the Fisher Information matrix of $\theta^f_s$.
\begin{equation}\nonumber
\theta^f_{s+1} = \theta^f_{s} + \alpha_s F(\theta_{s}^f)^{-1}\frac{\partial l(\theta_s^f)}{\partial\theta},
    \end{equation}
\begin{equation}\nonumber
\theta^f_{s+1} = \theta^f_{s} + \alpha_s \mathcal{I}(\theta_{s}^f)^+\frac{\partial l(\theta_s^f)}{\partial\theta}.
    \end{equation}
\end{theorem}
\begin{proof}
To prove Theorem \ref{FFSthm1}, it suffices to prove the below equality:
\begin{equation}\nonumber
    \mathcal{I}(\theta_{s}^f)^+\frac{\partial l(\theta_s^f)}{\partial\theta}=F(\theta_{s}^f)^{-1}\frac{\partial l(\theta_s^f)}{\partial\theta}.
\end{equation}
To prove the above, we employ Lemma \ref{FFSlem}. To achieve this, we proceed by defining:

\begin{itemize}
    \item[$\bullet$] $\tilde{N}$ to be the matrix:
    $$\begin{bmatrix}
     I_{p+1} & 0 & 0 & \hdots & 0 \\
     0 & N_{q_1} & 0 & \hdots & 0 \\
     0 & 0 & N_{q_2} & \hdots & 0 \\
     \vdots & \vdots & \vdots & \ddots & \vdots \\ 
     0 & 0 & 0 & \hdots & N_{q_r} \\
    \end{bmatrix},$$
    \item[$\bullet$] $\tilde{K}$ to be the matrix:
    $$\begin{bmatrix}
     I_{p+1} & 0 & 0 & \hdots & 0 \\
     0 & K_{q_1,q_1} & 0 & \hdots & 0 \\
     0 & 0 & K_{q_2,q_2} & \hdots & 0 \\
     \vdots & \vdots & \vdots & \ddots & \vdots \\ 
     0 & 0 & 0 & \hdots & K_{q_r,q_r} \\
    \end{bmatrix},$$
    \item[$\bullet$] $\delta$ to be the column vector:
    $$\frac{\partial l(\theta^f)}{\partial \theta^f},$$
    \item[$\bullet$] $C$ to be the matrix $F(\theta^f)$,
\end{itemize}
and claim that $F(\theta^f)\tilde{N}=\mathcal{I}(\theta^f)$. To prove this equality, we first note the following two properties of the matrix $N_k$, given by \citet{Magnus1986:dup}. 
\begin{enumerate}
    \item $N_{k_1}(A \otimes A)=N_{k_1}(A \otimes A)N_{k_2}=(A \otimes A)N_{k_2}$ for any arbitrary matrix $A$ and scalars $k_1$ and $k_2$, such that the matrix multiplications are well defined.
    \item $N_{k}\text{vec}(A)=\text{vec}(A)$ for any arbitrary symmetric matrix $A$ of appropriate dimension.
\end{enumerate}
The first of the above properties, in combination with equations (\ref{FFSFI2}) and (\ref{FFSF2}), can be seen to imply that, for any $k_1$ and $k_2$ between $1$ and $r$:
\begin{equation}\nonumber\mathcal{I}^f_{\text{vec}(D_{k_1}),\text{vec}(D_{k_2})}=F_{\text{vec}(D_{k_1}),\text{vec}(D_{k_2})}N_{q_{k_2}}.
\end{equation}
In a similar manner, the second of the above properties, in combination with equation (\ref{FFSFI1}), can be seen to imply that, for any $k$ between $1$ and $r$: \begin{equation}\nonumber\mathcal{I}^f_{\sigma^2,\text{vec}(D_k)}=F_{\sigma^2,\text{vec}(D_k)}N_{q_k}.
\end{equation}
Expanding the multiplication $F(\theta_s^f)\tilde{N}$ block-wise demonstrates that the two equations above are sufficient to prove the equality $F(\theta_s^f)\tilde{N}=\mathcal{I}(\theta^f_s)$. The matrices $\tilde{N}, \tilde{K}$ and $C$ can be seen to satisfy properties 1-3 of Lemma \ref{FFSlem} by noting standard results concerning the matrices $\mathcal{D}_k$, $K_{m,n}$ and $N_n$, provided in \citet{Magnus1986:dup}. Property 4 of Lemma \ref{FFSlem} can be seen to hold for the matrix $\tilde{N}$ and column vector $\delta$ by noting the symmetry of $\frac{\partial l(\theta^f_s)}{\partial D_k}$ (see Theorem \ref{DmatTheorem}). By combining the previous arguments and applying Lemma \ref{FFSlem}, the below is obtained:
\begin{equation}\nonumber\mathcal{I}(\theta^f)^+\delta=(F(\theta^f)\tilde{N})^+\delta=F(\theta^f)^{-1}\delta.
\end{equation}
The result of Theorem \ref{FFSthm1} now follows.\qed

\end{proof}

\begin{theorem}\label{FFSthm2}
For any fixed integer $k$ between $1$ and $r$, the below two update rules are equivalent, where $F_{\text{vec}(D_{k,s})}$ is as defined in Section \ref{FFSsection} and $\mathcal{I}^f_{\text{vec}(D_{k,s})}$ is the Fisher Information matrix of $\text{vec}(D_k)$.
\begin{equation}\nonumber
    \text{vec}(D_{k,s+1})=\text{vec}(D_{k,s})+\alpha_s F^{-1}_{\text{vec}(D_{k,s})}\frac{\partial l(\theta^f_s)}{\partial\text{vec}(D_{k,s})},
\end{equation}
\begin{equation}\nonumber
    \text{vec}(D_{k,s+1})=\text{vec}(D_{k,s})+\alpha_s\big(\mathcal{I}^{f}_{\text{vec}(D_{k,s})}\big)^{+}\frac{\partial l(\theta^f_s)}{\partial\text{vec}(D_{k,s})}.
\end{equation}
\end{theorem}
\begin{proof}
The proof of the above proceeds by defining $\tilde{N}=N_{q_k}$, $\tilde{K}=K_{q_k,q_k}$, $C=F_{\text{vec}(D_{k})}$ and $\delta$ to be the partial derivative appearing in the statement of the theorem. Following this, the proof follows the same structure as that of Theorem \ref{FFSthm1}. For brevity, the argument will not be repeated in this section. \qed
\end{proof}

\subsection{Restricted Maximum Likelihood Estimation}\label{remlApp}

In this appendix, we describe how the methods from Section \ref{AlgoSection} may be adapted to use an alternative likelihood criteria: the criteria employed by Restricted Maximum Likelihood (ReML) estimation. A well-documented issue for ML estimation is that the variance estimates produced using ML are biased. First described by \citet{Patterson1971}, and proposed for the LMM in \cite{Laird1982}, ReML offers an alternative to ML to address this issue by maximizing the log-likelihood function of the residual vector, $e$, instead of the response vector, $Y$. Neglecting constant terms, the Restricted Maximum log-likelihood function, $l_R$, is given by:
\begin{equation}\label{llhrestricted}
    l_R(\theta^h)=l(\theta^h)-\frac{1}{2}\bigg( -p\log(\sigma^2)+\log|X'V^{-1}X| \bigg),
\end{equation}
where $l(\theta^h)$ is given in (\ref{llh}). To derive the ReML-based FS algorithm, akin to that described in Section \ref{FSsection}, the following adjustments to the score vectors given by (\ref{FSderiv2}) and (\ref{FSderiv3}) must be used.
\begin{equation}\nonumber
\frac{dl_R(\theta^h)}{d\sigma^2}=\frac{dl(\theta^h)}{d\sigma^2}+\frac{1}{2}p\sigma^{-2},
\end{equation}
\begin{equation}\nonumber
  \begin{aligned}[b]
& \frac{dl_R(\theta^h)}{d\text{vech}(D_k)}=\frac{dl(\theta^h)}{d\text{vech}(D_k)}+\\
& \frac{1}{2}\mathcal{D}'_{q_k}\text{vec}\bigg(\sum_{j=1}^{l_k}Z_{(k,j)}'V^{-1}X(X'V^{-1}X)^{-1}X'V^{-1}Z_{(k,j)}\bigg).
\end{aligned}
\end{equation}
The latter of the above results can be derived through trivial adjustment of the proofs of Theorem \ref{DmatTheorem} and Corollaries \ref{partialDerivCorollary}-\ref{vechTotalDerivCor} given in Appendix \ref{derivAppendix}. As the derivative of $l_R(\theta^h)$ with respect to $\beta$ is identical to that of $l(\theta^h)$, given in Section \ref{FSsection}, we do not list it again above.

For a derivation of the ReML score vectors of $\beta$ and $\sigma^2$, we direct the reader to the works of \citet{Demidenko:2013mfx} where proofs may be found for the single-factor LMM. As these proofs do not depend on the number of factors in the model, they can be easily seen to also apply in the multi-factor LMM setting without further adjustment.

Due to the asymptotic equivalence of the ML and ReML estimates, the parameter estimates produced by ML and ReML have the same asymptotic covariance matrix. Consequently, the Fisher Information matrix for the parameter estimates produced by ReML, which can be shown to be equal to the inverse of the asymptotic covariance matrix, is identical to that of the parameter estimates produced by ML. As a result, the ReML-based FS algorithm utilizes the Fisher Information matrix specified by (\ref{FSFI1})-(\ref{FSFI3}) and requires no further derivation.

To summarize, the ReML-based FS algorithm for the multi-factor LMM is almost identical to that given in Algorithm \ref{FSalgorithm}. The only adaptations required occur on line 3 of Algorithm \ref{FSalgorithm} where the score vectors must be substituted for their ReML counterparts, provided above. Analogous adjustments can be made for the algorithms presented in Sections \ref{FFSsection}-\ref{CSFSsection}.

\subsection{Constraint Matrices}\label{conApp}
In this appendix, we provide further detail on the approach to constrained optimization which was outlined in Section \ref{covstruct}. We begin by providing examples of how the notation described in Section \ref{covstruct} may be used in practice. For a given $k$, we introduce the notation $\{d_i\}$ and $\{\tilde{d}_i\}$ to represent the set of unique parameters required to specify the constrained and unconstrained representations of $D_k$, respectively. For example, if $D_k$ is $(3 \times 3)$ in dimension and we wish to induce Toeplitz structure onto $D_k$, then $D_k$ can be considered to take the following ``constrained" and ``unconstrained" representations:
\begin{equation}\nonumber
    D_k = \begin{bmatrix}
    d_1 & d_4 & d_7 \\
    d_2 & d_5 & d_8 \\
    d_3 & d_6 & d_9 \\
    \end{bmatrix} = \begin{bmatrix}
    \tilde{d}_1 & \tilde{d}_2 & \tilde{d}_3 \\
    \tilde{d}_2 & \tilde{d}_1 & \tilde{d}_2 \\
    \tilde{d}_3 & \tilde{d}_2 & \tilde{d}_1 \\
    \end{bmatrix}.
\end{equation}

The notation vecu$(A)$ is used to denote the column vector of unique variables in the matrix $A$, given in order of first occurrence as listed in column major order. In the above example, this implies vecu$(D_k)=[\tilde{d}_1,\tilde{d}_2, \tilde{d}_3]'$ whilst vec$(D_k)=[d_1,\hdots, d_9]'$. Using this notation, the constraint matrix $\mathcal{C}_k$ may be defined element-wise. To see this, note that the chain rule for the total derivative is given by:
\begin{equation}\nonumber
    \frac{dl(\theta^{con})}{d\tilde{d}_j}=\sum_i\frac{\partial d_i}{\partial\tilde{d}_j}\frac{\partial l(\theta^{con})}{\partial d_i}.
\end{equation}
Comparing this with the derivative expression provided in equation (\ref{derivCk}), it can be seen that the elements of $\mathcal{C}_k$ are given by:
\begin{equation}\nonumber
    \mathcal{C}_{k[j,i]}=\frac{\partial d_i}{\partial\tilde{d}_j}.
\end{equation}

In general, derivation of the constraint matrix for the $k^{th}$ factor is relatively straightforward for most practical applications. Table \ref{constTab} supports this statement by providing examples of the constraint matrix for the $k^{th}$ factor operating under the assumption that $D_k$ exhibits several commonly employed covariance structures. In addition, we note that two instances in which a constraint matrix approach was employed implicitly can be found in the main body of this work. In Section \ref{FSsection}, we took $\mathcal{C}_k=\mathcal{D}_{q_k}$ to enforce symmetry in $D_k$ and, in Section \ref{CSFSsection}, we took $\mathcal{C}_k=(d\text{vech}(D_k)/d\text{vech}(\Lambda_k))\mathcal{D}_{q_k}$ in order to derive the Fisher Scoring update in terms of the Cholesky factor, $\Lambda_k$.

\begin{table}[h]
  \centering
  \begin{tabular}{@{}|c|c|@{}}
 \hline
 \textbf{Covariance Structure} & \textbf{Constraint matrix, $\mathcal{C}_k$} \\ 
 \hline
    & \\
 Diagonal &   \\ 
 $\begin{bmatrix}
    \tilde{d}_1 & 0 & 0 \\
    0 & \tilde{d}_1 & 0 \\
    0 & 0 & \tilde{d}_1 \\
    \end{bmatrix}$ & $\begin{bmatrix} 
    1 & 0 & 0 & 0 & 1 & 0 & 0 & 0 & 1 \\ \end{bmatrix}$ \\
    & \\
 \hline
    & \\
 Variance components &   \\ 
 $\begin{bmatrix}
    \tilde{d}_1 & 0 & 0 \\
    0 & \tilde{d}_2 & 0 \\
    0 & 0 & \tilde{d}_3 \\
    \end{bmatrix}$ & $\begin{bmatrix} 
    1 & 0 & 0 & 0 & 0 & 0 & 0 & 0 & 0 \\
    0 & 0 & 0 & 0 & 1 & 0 & 0 & 0 & 0 \\
    0 & 0 & 0 & 0 & 0 & 0 & 0 & 0 & 1 \\ \end{bmatrix}$ \\
    & \\
 \hline
    & \\
 Toeplitz &   \\ 
 $\begin{bmatrix}
    \tilde{d}_1 & \tilde{d}_2 & \tilde{d}_3 \\
    \tilde{d}_2 & \tilde{d}_1 & \tilde{d}_2 \\
    \tilde{d}_3 & \tilde{d}_2 & \tilde{d}_1 \\
    \end{bmatrix}$ & $\begin{bmatrix} 
    1 & 0 & 0 & 0 & 1 & 0 & 0 & 0 & 1 \\
    0 & 1 & 0 & 1 & 0 & 1 & 0 & 1 & 0\\
    0 & 0 & 1 & 0 & 0 & 0 & 1 & 0 & 0 \end{bmatrix}$ \\
    & \\
    \hline
    & \\
 Compound Symmetry  &   \\ 
 $\begin{bmatrix}
    1 & \tilde{d}_1 & \tilde{d}_1 \\
    \tilde{d}_1 & 1 & \tilde{d}_1 \\
    \tilde{d}_1 & \tilde{d}_1 & 1 \\
    \end{bmatrix}$ & $\begin{bmatrix} 
    0 & 1 & 1 & 1 & 0 & 1 & 1 & 1 & 0 \\ \end{bmatrix}$ \\
    & \\
    \hline
    & \\
 AR(1)  &   \\ 
 $\begin{bmatrix}
    1 & \tilde{d}_1 & \tilde{d}_1^2 \\
    \tilde{d}_1 & 1 & \tilde{d}_1 \\
    \tilde{d}_1^2 & \tilde{d}_1 & 1 \\
    \end{bmatrix}$ & \small{$\begin{bmatrix} 
    0 & 1 & 2\tilde{d}_1 & 1 & 0 & 1 & 2\tilde{d}_1 & 1 & 0 \\ \end{bmatrix}$} \\
    & \\
    \hline
\end{tabular}
  \caption{The constraint matrix for the $k^{th}$ factor, $\mathcal{C}_k$, in the setting in which the $k^{th}$ factor groups $3$ random effects, under various covariance structures. }\label{constTab}
\end{table}

It is also worth noting that this approach to constrained optimization allows for different covariance structures to be enforced on different factors in the LMM. For instance, a researcher may enforce the first factor in an LMM to have an AR(1) structure while allowing the second factor to exhibit a completely different structure such as compound symmetry. As noted in Section \ref{covstruct}, this approach can also be extended further to allow all elements of the vectorized covariance matrices, $\text{vec}(D_1),\hdots \text{vec}(D_r)$, to be expressed as functions of one set of parameters, $\rho_D$, common to all $\{\text{vec}(D_k)\}_{k \in \{1,\hdots, r\}}$. To further detail this, we define $v(D)$ and extend the definition $\theta^{con}$:
\begin{equation}\nonumber v(D)=\begin{bmatrix}\text{vec}(D_1)\\ \vdots\\ \text{vec}(D_r)\end{bmatrix}, \hspace{1cm}
\theta^{con}=\begin{bmatrix}\beta\\ \sigma^2\\ \rho_D \end{bmatrix},
\end{equation}

Through the same process used to obtain equation the equations presented in Section \ref{covstruct}, the below can be obtained;

\begin{equation}\label{FullDequation}\begin{aligned}
    & \frac{dl(\theta^{con})}{d \rho_D}=\mathcal{C}\frac{\partial l(\theta^{con})}{\partial v(D)}, \\
    & \mathcal{I}^{con}_{\rho_D}= \mathcal{C}\mathcal{I}^f_{\text{v}(D)}\mathcal{C}',
    \end{aligned}
\end{equation}
where $\mathcal{C}$ is the constraint matrix defined to be the Jacobian matrix of partial derivatives of $v(D)$ taken with respect to $\rho_D$. The score vector and Fisher Information matrix associated to $v(D)$ are readily seen to be given by block-wise combinations of the matrix expressions (\ref{FFSderiv1}) and (\ref{FFSFI2}).

\subsection{Cholesky Fisher Scoring}\label{CholAppendix}

In this appendix, we prove the identity stated by the below theorem, Theorem \ref{cholTheorem}, which was utilized in Section \ref{CSFSsection}.
\begin{theorem}\label{cholTheorem} Let $D_k$ be a square symmetric positive-definite matrix with Cholesky decomposition given by $D_k=\Lambda_k\Lambda_k'$ where $\Lambda_k$ is the lower triangular Cholesky factor. The below expression gives the derivative of vech$(D_k)$ with respect to vech$(\Lambda_k)$.
\begin{equation}\nonumber
    \frac{\partial\text{vech}(D_k)}{\partial\text{vech}(\Lambda_k)}=\mathcal{L}_{q_k}(\Lambda'_k \otimes I_{q_k})(I_{q_k^2}+K_{q_k})\mathcal{D}_{q_k}.
\end{equation}

\end{theorem}
\begin{proof}
By the chain rule for vector-valued functions, as stated by \citet{Turkington2013}, the derivative in Theorem \ref{cholTheorem} can be expanded in the following manner:
\begin{equation}\nonumber
    \frac{\partial\text{vech}(D_k)}{\partial \text{vech}(\Lambda_k)}=\frac{\partial \text{vec}(\Lambda_k)}{\partial \text{vech}(\Lambda_k)}\frac{\partial \text{vec}(D_k)}{\partial \text{vec}(\Lambda_k)}\frac{\partial \text{vech}(D_k)}{\partial \text{vec}(D_k)}.
\end{equation}
We now consider each of the above derivatives in turn. The first derivative in the product is given by Theorem 5.9 of \citet{Turkington2013}, which states:
\begin{equation}\nonumber
    \frac{\partial\text{vec}(\Lambda_k)}{\partial\text{vech}(\Lambda_k)}=\mathcal{L}_{q_k}.
\end{equation}
The second derivative in the product is given by a result of \citet{Magnus1999}, which states that:
\begin{equation}\nonumber
\frac{\partial\text{vec}(\Lambda_k\Lambda_k')}{\partial\text{vec}(\Lambda_k)}=(\Lambda'_k \otimes I_{q_k})(I_{q_k^2}+K_{q_k}).
\end{equation}
For the final derivative in the product, Theorem 5.10 of \citet{Turkington2013} gives:
\begin{equation}\nonumber
    \frac{\partial\text{vech}(D_k)}{\partial\text{vec}(D_k)}=\mathcal{D}_{q_k}.
\end{equation}
Combining the above derivatives yields the desired result. \qed

\end{proof}

\subsection{Satterthwaite degrees of freedom estimation}
In this appendix, we provide an in-depth accounting of Satterthwaite estimation for degrees of freedom in the multi-factor LMM setting. In Section \ref{Tsect}, we detail how the expression described by (\ref{swdf}) is derived for estimating the degrees of freedom of the approximate T-statistic. Following this, in Section \ref{Fsect}, we describe how this approach is extended to degrees of freedom estimation for the approximate F-statistic. Finally, in Section \ref{swproof}, we prove the derivative result stated in Section \ref{swsection} and extend it to the setting of the constrained covariance matrices discussed in Section \ref{covstruct}. The first two appendices follow the proofs outlined in the work of \citet{Kuznetsova2017}. The reader is referred to this work for a more in-depth discussion of Satterthwaite degrees of freedom estimation for the LMM.

\subsubsection{Satterthwaite estimation for the approximate T-statistic}\label{Tsect}

To derive the estimator given by (\ref{swdf}), we first begin by noting that the approximate T-statistic used for LMM null hypothesis testing is given by:
\begin{equation}\nonumber
    T=\frac{L\hat{\beta}}{\sqrt{L\widehat{\text{Var}}(\hat{\beta})L'}},
\end{equation}
where $\hat{\beta}$ is the estimator of $\beta$ obtained from numerical optimization and $\widehat{\text{Var}}(\hat{\beta})$ is an estimate of the variance of the $\beta$ estimator, given by:

\begin{equation}\nonumber
    \widehat{\text{Var}}(\hat{\beta})=\hat{\sigma}^{2}(X'\hat{V}^{-1}X)^{-1}.
\end{equation}
We highlight that the above expression differs from the T-statistic typically employed in the setting of linear regression hypothesis testing. Under the standard assumptions for linear regression, it is possible to derive an exact expression for the distribution of the OLS variance estimate of $\hat{\beta}$ given by $\hat{\sigma}^{2}(X'X)^{-1}$. However, the above estimator, $\widehat{\text{Var}}(\hat{\beta})$, differs from the OLS estimate as it includes an estimate of $V$ which has an unknown distribution. As a direct result, no equivalent theory is available to obtain the distribution of $\widehat{\text{Var}}(\hat{\beta})$ in the setting of the LMM. A consequence of this is that, whilst $T$ follows a student's $t_{n-p}$ distribution in the setting of linear regression, the distribution of $T$ in the LMM setting is unknown. The aim of this section is to approximate the degrees of freedom of $T$ assuming that it follows a students $t$-distribution.

To meet this aim, we now denote Var$(\hat{\beta})$ as the unknown true variance of the $\hat{\beta}$ estimator. By multiplying both the numerator and denominator of the T-statistic by $(L\text{Var}(\hat{\beta})L)^{1/2}$, the below expression for $T$ is obtained:

\begin{equation}\nonumber
T=\frac{Z}{\sqrt{\frac{L\widehat{\text{Var}}(\hat{\beta})L'}{L\text{Var}(\hat{\beta})L'}}},
\end{equation}
where $Z$ is a random variable with standard normal distribution, given by $Z=L\hat{\beta}/\sqrt{L\text{Var}(\hat{\beta})L'}$. We now consider the definition of the students $t$-distribution, which states that a random variable is $t$-distributed with $v$ degrees of freedom if and only if it takes the following form:

\begin{equation}\nonumber
T = \frac{Z}{\sqrt{\frac{X}{v}}},
\end{equation}
\noindent
where $Z\sim N(0,1)$ and $X\sim\chi^2_{v}$. When the degrees of freedom $v$ are unknown, as both $X$ and $T$ are distributed with $v$ degrees of freedom, it suffices to estimate the degrees of freedom of the chi-square variable $X$ instead of the $t$-distributed variable $T$. By equating the approximate T-statistic with the defining form of a $t$-distributed random variable, under the assumption that the T-statistic is truly $t$-distributed, the below expression for $X$ can be obtained:
\begin{equation}\nonumber
    X=\frac{vL\widehat{\text{Var}}(\hat{\beta})L'}{L\text{Var}(\hat{\beta})L'}\sim \chi^2_v.
\end{equation}
We now define the notation $S^2(\hat{\eta})$ to represent the estimated variance of $L\hat{\beta}$ as a function of the estimated variance parameters given by $\hat{\eta}=(\hat{\sigma}^2,\hat{D}_1,\hdots,\hat{D}_r)$. Similarly, we define $\Sigma^2$ to be the true unknown variance of $L\hat{\beta}$. Noting that $\Sigma^2=L\text{Var}(\hat{\beta})L'$ and by rearranging the above, the below can be obtained:
\begin{equation}\nonumber
S^2(\hat{\eta}) = L\widehat{\text{Var}}(\hat{\beta})L' \sim \frac{\Sigma^2}{v}\chi^2_{v}.
\end{equation}
By considering the variance of the above expression and noting that, as $X\sim \chi^2_v$, Var$(X)=2v$, the below expression is obtained.
\begin{equation}\nonumber
    \text{Var}(S^2(\hat{\eta}))=\frac{(\Sigma^2)^2}{v^2}\text{Var}(X)=\frac{2(\Sigma^2)^2}{v}.
\end{equation}
To obtain an estimator of $v$ in terms of $S^2(\hat{\eta})$, the above is rearranged and the approximation $\Sigma^2\approx S^2(\hat{\eta})$ is used. This yields the below approximation.
\begin{equation}\nonumber
    \hat{v}(\hat{\eta}) = \frac{2(S^2(\hat{\eta}))^2}{\text{Var}(S^2(\hat{\eta}))}.
\end{equation}
Under the assumption that $T\sim t_v$, $\hat{v}(\hat{\eta})$ is an estimator for the degrees of freedom of $X$ and, therefore, for the degrees of freedom of $T$. This concludes the derivation of (\ref{swdf}).

\subsubsection{Satterthwaite estimation for the approximate F-statistic}\label{Fsect}

In this appendix, we detail how the Satterthwaite method described in Appendix \ref{Tsect} may be extended in order to estimate the degrees of freedom of the approximate F-statistic. The approximate F-statistic for the LMM takes the below form:

\begin{equation}\nonumber
F=\frac{\hat{\beta}'L'[L\widehat{\text{Var}}(\hat{\beta})L']^{-1}L\hat{\beta}}{r},
\end{equation}
where, throughout this section only, the notation $F$ will represent the approximate F-statistic and $r$ will be used to denote $r=\text{rank}(L)$. The aim of this section is to approximate $F$ with an $\mathit{F}_{r,v}$ distribution where $v$, the denominator degrees of freedom, is estimated using a Satterthwaite method of approximation. To simplify the notation we will first consider $Q$, which is the numerator of $F$;
\begin{equation}\nonumber
Q=rF=\hat{\beta}'L'[L\text{Var}(\hat{\beta})L']^{-1}L\hat{\beta}.
\end{equation}
We begin by considering the eigendecomposition of $L\text{Var}(\hat{\beta})L'$, which we will denote:

\begin{equation}\nonumber
L\text{Var}(\hat{\beta})L'=U\Lambda U',
\end{equation}
where $U$ is an orthonormal matrix of eigenvectors and $\Lambda$ is a diagonal matrix of eigenvalues. Using standard properties of the eigendecomposition the below is obtained:

\begin{equation}\nonumber
\begin{aligned}[b]
    Q & = (U'L\hat{\beta})'\Lambda^{-1}(U'L\hat{\beta})\\
    &= \sum_{i=1}^r\frac{(U'L\hat{\beta})_i^2}{\lambda_i},
    \end{aligned}
\end{equation}
where $(U'L\hat{\beta})_i$ is the $i^{th}$ element of the vector $U'L\hat{\beta}$ and $\lambda_i$ is the $i^{th}$ diagonal element (eigenvalue) of $\Lambda$, $\Lambda_{[i,i]}$. For purposes which will become clear, we define $\tilde{L}$ as $U'L$. By noting that $\lambda_i=(U'L\text{Var}(\hat{\beta})L'U)_i$, the above can be seen to be equal to:

\begin{equation}\nonumber
Q= \sum_{i=1}^r\Bigg(\frac{\tilde{L}_i\hat{\beta}}{\sqrt{\tilde{L}_i\text{Var}(\hat{\beta})\tilde{L}'_i}}\Bigg)^2,
\end{equation}
where $\tilde{L}_i$ represents the $i^{th}$ row of $\tilde{L}$. Noting now the definition of $\tilde{L}$, it can be seen that the above is a sum of independent squared T-statistics of the form discussed in Appendix \ref{Tsect}. By denoting the unknown degrees of freedom of each of the T-statistics in the above sum by $v_i$, it follows that:

\begin{equation}\nonumber
\begin{aligned}[b]
\mathbb{E}(Q) & = \sum_{i=1}^r\mathbb{E}\Bigg[\Bigg(\frac{\tilde{L}_i\hat{\beta}}{\sqrt{\tilde{L}_i\text{Var}(\hat{\beta})\tilde{L}'_i}}\Bigg)^2\Bigg] \\ & = \sum_{i=1}^r\frac{v_i}{v_i-2},
\end{aligned}
\end{equation}
where the second equality follows from the fact that a squared T-statistic with degrees of freedom $v_i$ is an $F_{1,v_i}$-statistic, which in turn has expectation $\frac{v_i}{v_i-2}$. Noting that, in practice, the true variance of $\hat{\beta}$ is unknown, the methods of \ref{Tsect} must be employed to estimate each of the $v_i$ in the above sum.

To use the above expression for $\mathbb{E}(Q)$ to estimate the denominator degrees of freedom of $F$, three cases are considered; the cases $r>2, r=2$ and $r=1$. Each of these will now be considered in turn.\newline
\newline
\noindent
\textbf{Case 1 ($\mathbf{r>2}$):} By recalling $Q=rF$, and by the standard formula for the expectation of an $F$ distribution with numerator degrees of freedom $v$ greater than $2$, the below is obtained. 
\begin{equation}\nonumber
 \mathbb{E}(Q)= \mathbb{E}(rF)= \frac{rv}{v-2}.
\end{equation}
Setting the above equal to the previous expression for $\mathbb{E}(Q)$ and rearranging, the below expression for $v$ is obtained:

\begin{equation}\nonumber v = \frac{2\sum_{i=1}^r\frac{v_i}{v_i-2}}{\bigg(\sum_{i=1}^r\frac{v_i}{v_i-2}\bigg)-r}.
\end{equation}
By using the Satterthwaite degrees of freedom approximation method described in Appendix \ref{Tsect} to estimate $\{v_i\}_{i \in \{1,\hdots,r\}}$, it follows that the above formula can be used to estimate the degrees of freedom of the approximate $F$-statistic when $r>2$.
\newline
\newline
\noindent
\textbf{Case 2 ($\mathbf{r=2}$):} When $r=2$, the expectation of $Q$ is infinite and, resultantly, the equality used in the argument of case 1 no longer holds. A common convention for estimation of $v$ in the case $r=2$ is based on the observation that for $r>2$ the expectation of $Q$ satisfies the following relationship with $v$:
\begin{equation}\nonumber
    v = \frac{2\mathbb{E}(Q)}{\mathbb{E}(Q)-r}.
\end{equation}
By taking the limit from above of both sides of this expression, and noting that $\mathbb{E}(Q)\rightarrow \infty$ as $v \downarrow 2$, the following estimate for $v$ is obtained.
\begin{equation}\nonumber
    v = \lim_{r\downarrow 2}\frac{2\mathbb{E}(Q)}{\mathbb{E}(Q)-r}=2.
\end{equation}
\newline
\noindent
\textbf{Case 3 ($\mathbf{r=1}$):} For the case $r=1$, it is well known that if $F\sim \mathit{F}_{1,v_1}$ for some integer $v_1$, then $F$ is the square of a T-statistic with the same degrees of freedom, $v_1$. Resultantly, in this case, the unknown denominator degrees of freedom of the F-statistic, $v_1$, can be estimated using the methods of Appendix \ref{Tsect}.

\subsubsection{Derivative of $S^2(\hat{\eta}^h)$ with respect to vech$(\hat{D}_k)$}\label{swproof}

In this appendix, proof of the derivative result which was given in Section \ref{swsection} is provided. Following this, an extension of this result is provided for models containing constrained covariance matrices of the form described in Section \ref{covstruct}. The result given in Section \ref{swsection} is restated by the below theorem.

\begin{theorem}\label{s2derivthm} Define $\hat{\eta}^h$ as in section \ref{swsection} and define the function $S^2(\hat{\eta}^h)$ by:
\begin{equation}\nonumber
    S^2(\hat{\eta}^h)=\sigma^2L(X'\hat{V}^{-1}X)^{-1}L'.
\end{equation}
The derivative of $S^2(\hat{\eta}^h)$ with respect to the half vectorisation of $\hat{D}_k$ is given by:
\begin{equation}\nonumber
\frac{d S^2(\hat{\eta}^h)}{d \text{vech}(\hat{D}_k)} = 
\hat{\sigma}^{2}\mathcal{D}_{q_k}'\bigg(\sum_{j=1}^{l_k}\hat{B}_{(k,j)}\otimes \hat{B}_{(k,j)}\bigg),
\end{equation}
where $B_{(k,j)}$ is given by:
\begin{equation}\nonumber
    \hat{B}_{(k,j)} = Z_{(k,j)}'\hat{V}^{-1}X(X'\hat{V}^{-1}X)^{-1}L'.
\end{equation}
\end{theorem}

\begin{proof}
By the chain rule for vector valued functions, as stated by \citet{Turkington2013}, and by noting that the function $S^2(\hat{\eta}^h)$ outputs a $(1 \times 1)$ scalar value, it can be seen that:

\begin{equation}\nonumber
\begin{aligned}[b]
     & \frac{\partial S^2(\hat{\eta}^h)}{\partial \text{vec}(\hat{D}_k)}=\hat{\sigma}^2\frac{\partial\big(L(X'\hat{V}^{-1}X)^{-1}L'\big)}{\partial\text{vec}(\hat{D}_k)} =\\
    & \hat{\sigma}^2 \frac{\partial\text{vec}(\hat{V})}{\partial\text{vec}(\hat{D}_k)}
    \frac{\partial\text{vec}(\hat{V}^{-1})}{\partial\text{vec}(\hat{V})} \frac{\partial\text{vec}(X'\hat{V}^{-1}X)}{\partial\text{vec}(\hat{V}^{-1})} 
    \frac{\partial\big(L(X'\hat{V}^{-1}X)^{-1}L'\big)}{\partial\text{vec}(X'\hat{V}^{-1}X)}.
    \end{aligned}
\end{equation}
We now consider each of the terms in this product in turn. Using the expansion given in (\ref{vdef}) the first derivative in the product is evaluated to the below.
\begin{equation}\nonumber
\frac{\partial\text{vec}(\hat{V})}{\partial\text{vec}(\hat{D}_k)}=\sum_{j=1}^{l_k}Z'_{(k,j)}\otimes Z'_{(k,j)}.
\end{equation}
For the second derivative in the product, we apply result (4.17) from \citet{Turkington2013}, which states:
\begin{equation}\nonumber
\frac{\partial\text{vec}(\hat{V}^{-1})}{\partial\text{vec}(\hat{V})}=-\hat{V}^{-1}\otimes \hat{V}^{-1}.
\end{equation}
For the third term of the product, a result stated in Chapter 5.7 of \citet{Turkington2013} gives the following:
\begin{equation}\nonumber
    \frac{\partial\text{vec}(X'\hat{V}^{-1}X)}{\partial\text{vec}(\hat{V}^{-1})}=X \otimes X.
\end{equation}
By a variant of (4.17) from \citet{Turkington2013}, given on the line following the statement of (4.17), the below is obtained:
\begin{equation}\nonumber
\begin{aligned}[b]
& \frac{\partial\big(L(X'\hat{V}^{-1}X)^{-1}L'\big)}{\partial\text{vec}(X'\hat{V}^{-1}X)}=\\
& -\text{vec}(L(X'\hat{V}^{-1}X)^{-1}(L(X'\hat{V}^{-1}X)^{-1})')\\
& =-((X'\hat{V}^{-1}X)^{-1}L')\otimes ((X'\hat{V}^{-1}X)^{-1}L').
\end{aligned}
\end{equation}
Following this, application of the mixed product property of the Kronecker product and multiplication by the transposed duplication matrix yields the result of Theorem \ref{s2derivthm}. \qed

\end{proof}

\begin{theorem}\label{s2derivthm2} Define $\rho_{\hat{D}}$ and $\mathcal{C}$ as in Section (\ref{covstruct}) and $\hat{B}_{(k,j)}$ and $\hat{\eta}^h$ as in Theorem \ref{s2derivthm}. The derivative of $S^2(\hat{\eta}^h)$ with respect to $\rho_{\hat{D}}$ is given by:
\begin{equation}\nonumber
\frac{d S^2(\hat{\eta}^h)}{d \rho_{\hat{D}}} = 
\sigma^{2}\mathcal{C}\hat{\mathcal{B}},
\end{equation}
where $\hat{\mathcal{B}}$ is defined by:
\begin{equation}\nonumber
\hat{\mathcal{B}}= \left[\bigg(\sum_{j=1}^{l_1}\hat{B}_{(1,j)}'\otimes \hat{B}_{(1,j)}'\bigg),\hdots,\bigg(\sum_{j=1}^{l_r}\hat{B}_{(r,j)}'\otimes \hat{B}_{(r,j)}'\bigg)\right]'.
\end{equation}
\end{theorem}

\begin{proof} From the proof of Theorem \ref{s2derivthm}, it can be seen that the partial derivative of $S^2(\hat{\eta}^h)$ with respect to vec$(\hat{D}_k)$ is given by:
\begin{equation}\nonumber
\frac{\partial S^2(\hat{\eta}^h)}{\partial \text{vec}(\hat{D}_k)} = 
\hat{\sigma}^{2}\sum_{j=1}^{l_k}\hat{B}_{(k,j)}\otimes \hat{B}_{(k,j)}.
\end{equation}
By defining $v(\hat{D})$ as in Section \ref{covstruct}, it can be seen from the above that the partial derivative of $S^2(\hat{\eta}^h)$ with respect to $v(\hat{D})$ is given by:
\begin{equation}\nonumber
\frac{\partial S^2(\hat{\eta}^h)}{\partial v(\hat{D})} = 
\hat{\sigma}^{2}\hat{\mathcal{B}}.
\end{equation}
By the chain rule for vector valued functions, as stated by \citet{Turkington2013}, the below can now be obtained;
\begin{equation}\nonumber
\frac{d S^2(\hat{\eta}^h)}{d \rho_{\hat{D}}} = 
\frac{\partial v(\hat{D})}{\partial \rho_{\hat{D}}}\frac{\partial S^2(\hat{\eta}^h)}{\partial v(\hat{D})}=\mathcal{C}\frac{\partial S^2(\hat{\eta}^h)}{\partial v(\hat{D})},
\end{equation}
where the second equality follows from the definition of $\mathcal{C}$. Substituting the partial derivative of $S^2(\hat{\eta}^h)$ with respect to $\rho_{\hat{D}}$ into the above completes the proof. \qed

\end{proof}

\subsection{The ACE model}\label{ACEapp}
In this appendix, further detail concerning the ACE model of Section \ref{AceExample} is provided. First, Appendix \ref{GammaApp} details how the random effects vector, $b$, and covariance matrix, $D$, may be constructed for the ACE model. Following this, Appendix \ref{KinApp} details how the matrices which model the additive genetic and common environmental variance components of the ACE model are defined. Next, Appendix \ref{AceConApp} provides an overview of the constrained optimization procedure adopted for the ACE model, alongside an expression for the constraint matrix employed for optimization. Finally, Appendix \ref{ACEcompApp} provides detail on how the parameter estimation method that was employed to obtain the results of Section \ref{AceRes} was efficiently implemented.

\subsubsection{Specification of random effects}\label{GammaApp}

In this Appendix, we provide detail on how the random effects vector, $b$, and covariance matrix, $D$, are defined for the ACE model. To do so, we first describe the covariance of the random terms $\gamma_{k,j,i}$, which appear in twin study model from Section \ref{AceExample}. Following this, we use the $\gamma_{k,j,i}$ terms to construct the random effects vector, $b$. To simplify notation, we make the assumption that for each family structure type, subjects within family units which exhibit such a structure are ordered consistently. For example, if the family structure describes ``families containing one twin-pair and one half-sibling'', it may be assumed that the members of every family who exhibit such a structure are given in the same order: (twin, twin, half-sibling). 

As noted in Section \ref{AceExample}, $\gamma_{k,j,i}$ models the within-``family unit" covariance. As such, $\text{cov}(\gamma_{k,j_1,i},\gamma_{k,j_2,i})=0$, for any two distinct family units, $j_1\neq j_2$. Within any individual family unit (e.g. family unit $j$ of structure type $k$), the random effects $\{\gamma_{k,j,i}\}_{i \in \{1,\hdots, q_k\}}$ are defined to have the below covariance matrix:
\begin{equation}\nonumber
\begin{aligned}[b]
    & \text{cov}\left(\begin{bmatrix}
     \gamma_{k,j,1} \\
     \gamma_{k,j,2} \\
     \vdots \\
     \gamma_{k,j,q_k} \\
    \end{bmatrix}\right) = \sigma^2_a\mathbf{K}^a_k + \sigma^2_c\mathbf{K}^c_k,\\
\end{aligned}
\end{equation}
where $\mathbf{K}^a_k$ and $\mathbf{K}^c_k$ are the known, predefined kinship (expected additive genetic material) and common environmental effects matrices, respectively (see Appendix \ref{KinApp} for further detail).

The random effects vector $b$, may be constructed by vertical concatenation of the random $\gamma_{k,i,j}$ terms, i.e. $b = [\gamma_{1,1,1},\gamma_{1,1,2},...,\gamma_{r,l_r,q_r}]'$. To derive an expression for the covariance matrix of $b$, we note from equation (\ref{LMMdef2}) that $\sigma^2_eD$ is equal to cov$(b)$. Equating this with the previous expression, it may now be seen that $D$ is block diagonal, with its $k^{th}$ unique diagonal block given by:
\begin{equation}\nonumber
\begin{aligned}
    D_k & =  \sigma^{-2}_e(\sigma^2_a\mathbf{K}^a_k + \sigma^2_c\mathbf{K}^c_k) \\
\end{aligned}. 
\end{equation}

\subsubsection{Kinship and environmental matrices}\label{KinApp}

In twin studies, twin pairs are typically described as either Monozygotic (MZ) or Dizygotic (DZ). MZ twin pairs are identical twins, whilst DZ twins are non-identical. The additive genetic kinship matrix, $\mathbf{K}^a_k$, describes the additive genetic similarity between members of a family unit. For a family of $q_k$ members, $\mathbf{K}^a_k$, is a $(q_k \times q_k)$ dimensional constant matrix with its $(i,j)^{th}$ element given by:

\begin{equation}\nonumber\small{
    (\mathbf{K}^a_k)_{[i,j]} = \begin{cases}
    1 & \text{if subjects }i \text{ and }j \text{ are MZ twins or }i=j.\\
    \frac{1}{2} & \text{if subjects }i \text{ and }j \text{ are DZ twins or full siblings.}\\
    \frac{1}{4} & \text{if subjects }i \text{ and }j \text{ are half siblings.} \\
    0 & \text{if subjects }i \text{ and }j \text{ are unrelated.} \\
    \end{cases}}
\end{equation}

The common environmental matrix, $\mathbf{K}^c_k$, describes the similarity between members of a family unit which is attributed to individuals' being reared in a shared environment. $\mathbf{K}^c_k$ is a $(q_k \times q_k)$ dimensional constant matrix with its $(i,j)^{th}$ element given by:
\begin{equation}\nonumber\small{
    (\mathbf{K}^a_k)_{[i,j]} = \begin{cases}
    1 & \text{if subjects }i \text{ and }j \text{ were reared together.}\\
    0 & \text{if subjects }i \text{ and }j \text{ were not reared together.} \\
    \end{cases}}
\end{equation}

For more information on the construction of kinship and common environmental matrices for twin studies, the reader is referred to \citet{lawlor2009family}.

\subsubsection{Constrained optimization for the ACE model}\label{AceConApp}

In this section, a brief overview of the constrained optimization procedure which was adopted for parameter estimation of the ACE model is provided. An expression for the constraint matrix, $\mathcal{C}$ is then provided by Theorem \ref{ADEThm}, alongside derivation.

To perform parameter estimation for the ACE model, an approach based on the ReML criterion described in Appendix \ref{remlApp} and the constrained covariance structure methods outlined in Section \ref{covstruct} was adopted. The resulting optimization procedure was performed in terms of the parameter vector $\theta^{ACE}=(\beta, \tau_a, \tau_c, \sigma^2_e)$, where $\tau_a=\sigma^{-1}_e\sigma_a$ and  $\tau_c=\sigma^{-1}_e\sigma_c$. $\beta$ and $\sigma^2_e$ were updated according to the GLS update rules provided by equations (\ref{betaUpdate}) and (\ref{sigma2Update}), respectively, whilst updates for the parameter vector  $[\tau_a, \tau_c]'$ were performed via a Fisher Scoring update rule. The Fisher Scoring update rule employed was of the form (\ref{FS}) with $\theta$ substituted for $[\tau_a, \tau_c]'$. To obtain the necessary gradient vector and information matrix required to perform this update step, equation (\ref{FullDequation}) of Appendix \ref{conApp} was employed. An expression for the required constraint matrix, $\mathcal{C}$, alongside derivation, is given by Theorem \ref{ADEThm} below.

\begin{theorem}\label{ADEThm}

Let $v(D)$ denote the vector of covariance parameters $[\text{vec}(D_1)', \text{vec}(D_2)',\hdots  \text{vec}(D_r)']'$. For the ACE model, the constraint matrix (Jacobian) which maps a partial derivative vector taken with respect to $v(D)$ to a total derivative vector taken with respect to $\tau=[\tau_a,\tau_c]'$ is given by:
\begin{equation}\small{\nonumber
    \begin{aligned} \mathcal{C} =  \bigg(\mathbb{1}_{(1,r)} \otimes \begin{bmatrix}
         2\tau_a & 0 \\ 0 & 2\tau_c
        \end{bmatrix}\bigg)  \bigg(\bigoplus_{k=1}^r \begin{bmatrix}
         \text{vec}(\mathbf{K}^a_k)' \\ \text{vec}(\mathbf{K}^c_k)'
        \end{bmatrix}\bigg),
    \end{aligned}}
\end{equation}
where $\oplus$ represents the direct sum. 
\end{theorem}

\begin{proof}
To prove Theorem \ref{ADEThm}, we first define $\tilde{\tau}_{a,1},\hdots \tilde{\tau}_{a,r}$ and $\tilde{\tau}_{c,1},\hdots \tilde{\tau}_{c,r}$ as variables which satisfy the below expressions, for all $k \in \{1,\hdots,r\}$.
\begin{equation}\nonumber
    \begin{aligned}
     & \tau_a = \tilde{\tau}_{a,k}, \hspace{0.2cm} \tau_c = \tilde{\tau}_{c,k}, \\
     & \text{vec}(D_k) = \tilde{\tau}^2_{a,k}\text{vec}(\mathbf{K}^a_k) + \tilde{\tau}^2_{c,k}\text{vec}(\mathbf{K}^c_k).
    \end{aligned}    
\end{equation}
Following this, we define the vector $\tilde{\tau}$ as $\tilde{\tau}=[\tilde{\tau}_{a,1},\tilde{\tau}_{c,1}, \hdots, \tilde{\tau}_{a,r},\tilde{\tau}_{c,r}]'$. By the chain rule for vector-valued functions, as stated by \citet{Turkington2013}, and by the definition of $\mathcal{C}$, provided in Section \ref{covstruct}, it can be seen that:
\begin{equation}\small{\label{ADEchain}
   \mathcal{C}=\frac{d v(D)}{d\tau}=\frac{\partial \tilde{\tau}}{\partial\tau}\frac{\partial v(D_k)}{\partial\tilde{\tau}}}.
\end{equation}
By direct evaluation, the first partial derivative in the product on the right hand side of the above can be seen to be equal to the following.
\begin{equation} \nonumber
    \frac{\partial \tilde{\tau}}{\partial\tau} = \begin{bmatrix}
     1 & 0 & 1 & 0 & 1 & \hdots 0 \\
     0 & 1 & 0 & 1 & 0 & \hdots 1 \\
    \end{bmatrix} = \mathbb{1}_{(1,r)} \otimes I_2.
\end{equation}
To evaluate the second derivative in the product, we first consider the partial derivative of vec$(D_k)$ with respect to the parameter vector $[\tilde{\tau}_{a,k},\tilde{\tau}_{c,k}]'$ for arbitrary $k \in \{1,\hdots,r\}$. From the definitions of $\tilde{\tau}_{a,k}$ and $\tilde{\tau}_{c,k}$, it can be seen that:
\begin{equation}\nonumber
\begin{aligned}
\frac{\partial \text{vec}(D_k)}{\partial [\tilde{\tau}_{a,k}, \tilde{\tau}_{c,k}]'} & =  \frac{\partial }{\partial [\tilde{\tau}_{a,k}, \tilde{\tau}_{c,k}]'}(\tilde{\tau}^2_{a,k}\text{vec}(\mathbf{K}^a_k) + \tilde{\tau}^2_{c,k}\text{vec}(\mathbf{K}^c_k)) \\
& =\begin{bmatrix}
 \frac{\partial }{\partial \tilde{\tau}_{a,k}} (\tilde{\tau}^2_{a,k}\text{vec}(\mathbf{K}^a_k) + \tilde{\tau}^2_{c,k}\text{vec}(\mathbf{K}^c_k))' \\
 \frac{\partial }{\partial \tilde{\tau}_{c,k}} (\tilde{\tau}^2_{a,k}\text{vec}(\mathbf{K}^a_k) + \tilde{\tau}^2_{c,k}\text{vec}(\mathbf{K}^c_k))'\\
\end{bmatrix}\\
& =\begin{bmatrix}
 2\tilde{\tau}_{a,k}\text{vec}(\mathbf{K}^a_k)' \\
 2\tilde{\tau}_{c,k}\text{vec}(\mathbf{K}^c_k)'\\
\end{bmatrix}.\\
\end{aligned}
\end{equation}
By similar reasoning it can be seen, for arbitrary $k_1,k_2 \in \{1,\hdots,r\}$, such that $k_1 \neq k_2$, that the below is true:
\begin{equation}\nonumber
\frac{\partial \text{vec}(D_{k_1})}{\partial [\tilde{\tau}_{a,k_2}, \tilde{\tau}_{c,k_2}]'} =  \mathbf{0}_{(2,q^2_{k_1})},
\end{equation}
where $\mathbf{0}_{(2,q^2_{k_1})}$ is the $(2 \times q^2_{k_1})$ dimensional matrix of zero elements. By combining the above expressions and noting the definitions of $v(D)$ and $\tilde{\tau}$, it can now be seen that the derivative of $v(D)$ with respect to $\tilde{\tau}$ is given by the below.
\begin{equation}\nonumber
    \begin{aligned}
    & \frac{\partial v(D)}{\partial\tilde{\tau}}= \\ & \begin{bmatrix}
     \begin{bmatrix}
 2\tilde{\tau}_{a,1}\text{vec}(\mathbf{K}^a_1)' \\
 2\tilde{\tau}_{c,1}\text{vec}(\mathbf{K}^c_1)'\\
\end{bmatrix} & \mathbf{0}_{(2,q^2_2)} & \hdots & \mathbf{0}_{(2,q^2_r)} \\
     \mathbf{0}_{(2,q^2_1)} &\begin{bmatrix}
 2\tilde{\tau}_{a,2}\text{vec}(\mathbf{K}^a_2)' \\
 2\tilde{\tau}_{c,2}\text{vec}(\mathbf{K}^c_2)'\\
\end{bmatrix} &  \hdots & \mathbf{0}_{(2,q^2_r)} \\
 \vdots & \vdots & \ddots & \vdots \\
     \mathbf{0}_{(2,q^2_1)} & \mathbf{0}_{(2,q^2_2)} & \hdots & \begin{bmatrix}
 2\tilde{\tau}_{a,r}\text{vec}(\mathbf{K}^a_r)' \\
 2\tilde{\tau}_{c,r}\text{vec}(\mathbf{K}^c_r)'\\
\end{bmatrix}\\
    \end{bmatrix} \\
    & = \bigoplus_{k=1}^r \begin{bmatrix}
         2\tilde{\tau}_{a,k}\text{vec}(\mathbf{K}^a_k)' \\ 2\tilde{\tau}_{c,k}\text{vec}(\mathbf{K}^c_k)'
        \end{bmatrix}. \\
    \end{aligned}
\end{equation}
By substituting the above partial derivative results into (\ref{ADEchain}), the following expression can now be obtained for the constraint matrix, $\mathcal{C}$.
\begin{equation}\small{\nonumber
   \mathcal{C}=\bigg(\mathbb{1}_{(1,r)} \otimes I_2\bigg)\bigg(\bigoplus_{k=1}^r \begin{bmatrix}
         2\tilde{\tau}_{a,k}\text{vec}(\mathbf{K}^a_k)' \\ 2\tilde{\tau}_{c,k}\text{vec}(\mathbf{K}^c_k)'
        \end{bmatrix}\bigg)}.
\end{equation}
By substituting $\tau_a = \tilde{\tau}_{a,k}$ and $\tau_c = \tilde{\tau}_{c,k}$ and rearranging, the result stated in Theorem \ref{ADEThm} can now be obtained. \qed
\end{proof}

\subsubsection{Computation and the ACE model}\label{ACEcompApp}

A key aspect in which the ACE model differs from other LMM's considered in this work is that, in the ACE model, the second dimension of the random effects design matrix, $q$, is equal to the number of observations, $n$. Resultantly, the random effects design matrix, $Z$, which is given by the $(n \times n)$ identity matrix, has notably large dimensions. Due to the large size of $Z$, the methods of Section \ref{compeff}, which assume that the second dimension of $Z$ is much smaller than $n$ (i.e. $q << n$), cannot be applied to the ACE model in order to improve computation speed. In this appendix, we briefly outline how the FSFS algorithm of Section \ref{FSFSsection} may be implemented efficiently, in conjunction with the ReML adjustments described in Appendix \ref{remlApp} and the constrained optimization methods of Section \ref{covstruct} to perform parameter estimation for the ACE model.

To begin, we consider the GLS updates, given by equations (\ref{betaUpdate}) and (\ref{sigma2Update}), which are employed for updating the parameters $\beta$ and $\sigma^2$, respectively. By defining the matrix, $\bar{D}_k$ by $\bar{D}_k=I_{q_k}+D_k$, and noting that the matrix $Z$ is the $(n \times n)$ identity matrix, the below matrix equality can be seen to hold for the ACE model:
\begin{equation}\nonumber
\begin{aligned}
    & X'V^{-1}X = \sum_{k=1}^r \sum_{j=1}^{l_k}X_{(k,j)}'\bar{D}^{-1}_k X_{(k,j)} \\
\end{aligned},
\end{equation}
where $X_{(k,j)}$ represents the design matrix for the $j^{th}$ family unit which possesses family structure of type $k$. Via well-known properties of the Kronecker product, the above equality can be seen to be equivalent to the following:
\begin{equation}\nonumber
\begin{aligned}
    & \text{vec}(X'V^{-1}X) = \sum_{k=1}^r \bigg(\sum_{j=1}^{l_k}X_{(k,j)}'\otimes X_{(k,j)}'\bigg) \text{vec}(\bar{D}^{-1}_k). \\
\end{aligned}
\end{equation}

The above equality is noteworthy as the matrix expression which appears inside the large brackets on the right-hand side is fixed and does not depend on $\beta$, $\sigma^2$ or $D$. Further, the matrix expression inside the brackets is typically small in size and has dimensions $(p^2 \times q^2_k)$. As a result, this expression can be calculated and stored prior to the iterative procedure of the FSFS algorithm. This means that the above equality offers a quick and convenient method for repeated evaluation of the matrix, $X'V^{-1}X$, which is employed for many calculations throughout the FSFS algorithm. Similar logic can be applied to the evaluation of the vector, $X'V^{-1}Y$, and the scalar value, $Y'V^{-1}Y$. Using this approach, it can now be seen that the GLS estimators, (\ref{betaUpdate}) and (\ref{sigma2Update}), can be calculated using only $\{D_k\}_{k \in \{1,\hdots,r\}}$ and pre-calculated matrices which are small in dimension. Since the matrices $\{D_k\}_{k \in \{1,\hdots,r\}}$ are also typically small in dimension, this calculation is extremely quick and computationally efficient.

Next, we consider the matrix $F_{\text{vec}(D_k)}$ and the partial derivative vector of $l_R (\theta^f)$, taken with respect to $\text{vec}(D_k)$. Expressions for these quantities can be obtained using equations (\ref{FFSFI2}) and (\ref{FFSderiv1}), respectively, in combination with the adjustments described in Appendix \ref{remlApp}. Again noting that the random effects design matrix, $Z$, is equal to the identity matrix, the ReML equivalent of equation (\ref{FFSderiv1}) simplifies substantially to the below expression:
\begin{equation}\small{\nonumber
    \begin{aligned}
    & \frac{\partial l_R(\theta^f)}{\partial \text{vec}(D_k)}= \\
    & \frac{1}{2}\text{vec}\bigg(\bar{D}_k^{-1}\bigg(\frac{E_kE_k'}{\sigma^2}-l_k\bar{D}_k + \sum_{j=1}^{l_k}X_{(k,j)}(X'V^{-1}X)^{-1}X_{(k,j)}'\bigg)\bar{D}_k^{-1}\bigg),
    \end{aligned}}
\end{equation}
where $E_k=\text{vec}_{q_k}(e_{(k)}')'$, vec$_m$ is defined as the generalized vectorization operator described in Section \ref{compeff} and $e_{(k)}$ is defined as the residual vector corresponding to subjects who belong to a family unit with family structure type $k$. Using equation (\ref{FFSFI2}), the sub-matrix of $F$ corresponding to vec$(D_k)$ can also be seen to be given by the following simplified expression.
\begin{equation}\nonumber
    F_{\text{vec}(D_k)}= \frac{1}{2}l_k(\bar{D}_k^{-1} \otimes \tilde{D}_k^{-1}).
\end{equation}
The two expressions above can be evaluated computationally in an extremely efficient manner primarily for two reasons. The first of these is that the matrices involved in computation are typically small; $D_k$, $E_k$ and $X_{(k,j)}$ have dimensions $(q_k \times q_k)$, $(q_k \times l_kq_k)$ and $(q_k \times p)$, respectively. The second reason is that, excluding the inversions of $\bar{D}_k$ and $X'V^{-1}X$, the only operations required to evaluate the above expressions are computationally quick to perform. These operations are: matrix multiplication, the reshape operation, matrix addition and the Kronecker product.

Lastly, we consider the evaluation of the score vector and Fisher Information matrix associated to $\tau=[\tau_a,\tau_c]'=\sigma^{-1}_e[\sigma_a,\sigma_c]'$, which shall be denoted as $\frac{d l(\theta^{ACE})}{d \tau}$ and $\mathcal{I}^{ACE}_\tau$, respectively. We define $\mathbf{K}_k=[\text{vec}(\mathbf{K}^a_k), \text{vec}(\mathbf{K}^c_k)]'$. By employing equation (\ref{FullDequation}) of Appendix \ref{conApp}, and using the constraint matrix derived in Appendix \ref{AceConApp}, the score vector and Fisher Information matrix associated to $\tau$ can be reduced to the following:
\begin{equation}\nonumber
\begin{aligned}
    & \mathcal{I}^{ACE}_\tau = (\tau\tau') \odot\bigg(\sum_{k=1}^r \mathbf{K}_kF_{\text{vec}(D_k)}\mathbf{K}_k'\bigg),\\
    & \frac{d l_R(\theta^{ACE})}{d \tau} = \tau \odot \bigg(\sum_{k=1}^r\mathbf{K}_k\frac{\partial l_R(\theta^{f})}{\partial \text{vec}(D_k)}\bigg),
\end{aligned}
\end{equation}
where $\odot$ represents the Hadamard (entry-wise) product. The results of Section \ref{AceRes} were produced by using the GLS updates for the parameters $\beta$ and $\sigma^2$ and by employing the two expressions above to perform Fisher Scoring update steps for the parameter vector $\tau$.

It should be noted that, in the ACE model, misspecification of variance components can result in a low-rank Fisher Information matrix. This can be seen by noting the positioning of the Hadamard product in the above Fisher Information matrix expression and by considering the case in which one of the elements of $\tau$ equals $0$. If not accounted for appropriately, this issue may, in practice, result in failure of the algorithm to converge. A simple, practical solution to this issue is to execute the algorithm multiple times using different initial starting points; once with a starting point where $\tau_a$ is set to zero, once with $\tau_c$ set to zero and once with neither $\tau_a$ nor $\tau_c$ set to zero. The results of Section \ref{AceRes} were obtained using this approach.


%
%

\bibliographystyle{spbasic}      
\bibliography{references}   

\begin{thebibliography}{41}
\providecommand{\natexlab}[1]{#1}
\providecommand{\url}[1]{{#1}}
\providecommand{\urlprefix}{URL }
\expandafter\ifx\csname urlstyle\endcsname\relax
  \providecommand{\doi}[1]{DOI~\discretionary{}{}{}#1}\else
  \providecommand{\doi}{DOI~\discretionary{}{}{}\begingroup
  \urlstyle{rm}\Url}\fi
\providecommand{\eprint}[2][]{\url{#2}}

\bibitem[{Barnett(1990)}]{barnett1990matrices}
Barnett S (1990) Matrices: Methods and Applications. Oxford applied mathematics
  and computing science series, Clarendon Press

\bibitem[{Bates et~al.(2015)Bates, Mächler, Bolker, and
  Walker}]{Bates:2015pls}
Bates D, Mächler M, Bolker B, Walker S (2015) Fitting linear mixed-effects
  models using lme4. Journal of Statistical Software, Articles 67(1):1--48,
  \doi{10.18637/jss.v067.i01}

\bibitem[{Demidenko(2013)}]{Demidenko:2013mfx}
Demidenko E (2013) Mixed models. Theory and applications with R. 2nd ed.
  \doi{10.1002/9781118651537}

\bibitem[{Dempster et~al.(1977)Dempster, Laird, and Rubin}]{Dempster1977}
Dempster AP, Laird NM, Rubin DB (1977) Maximum likelihood from incomplete data
  via the em algorithm. Journal of the Royal Statistical Society Series B
  (Methodological) 39(1):1--38,
  \urlprefix\url{http://www.jstor.org/stable/2984875}

\bibitem[{Dempster et~al.(1981)Dempster, Rubin, and Tsutakawa}]{Dempster1981}
Dempster AP, Rubin DB, Tsutakawa RK (1981) Estimation in covariance components
  models. Journal of the American Statistical Association 76(374):341--353,
  \doi{10.1080/01621459.1981.10477653}

\bibitem[{Giles(2008)}]{giles2008:collect}
Giles MB (2008) Collected matrix derivative results for forward and reverse
  mode algorithmic differentiation

\bibitem[{Henderson et~al.(1959)Henderson, Kempthorne, Searle, and von
  Krosigk}]{Henderson1959}
Henderson CR, Kempthorne O, Searle SR, von Krosigk CM (1959) The estimation of
  environmental and genetic trends from records subject to culling. Biometrics
  15(2):192--218, \urlprefix\url{http://www.jstor.org/stable/2527669}

\bibitem[{Hong and Raudenbush(2008)}]{Raudenbush2008}
Hong G, Raudenbush SW (2008) Causal inference for time-varying instructional
  treatments. Journal of Educational and Behavioral Statistics 33(3):333--362,
  \doi{10.3102/1076998607307355}

\bibitem[{{IBM Corp}(2015)}]{SPSSMIXED}
{IBM Corp} (2015) IBM SPSS Advanced Statistics 23. Armonk, NY: IBM Corp.

\bibitem[{Jennrich and Schluchter(1986)}]{Jennrich1986}
Jennrich RI, Schluchter MD (1986) Unbalanced repeated-measures models with
  structured covariance matrices. Biometrics 42(4):805--820,
  \urlprefix\url{http://www.jstor.org/stable/2530695}

\bibitem[{Kuznetsova et~al.(2017)Kuznetsova, Brockhoff, and
  Christensen}]{Kuznetsova2017}
Kuznetsova A, Brockhoff P, Christensen R (2017) lmertest package: Tests in
  linear mixed effects models. Journal of Statistical Software, Articles
  82(13):1--26, \doi{10.18637/jss.v082.i13},
  \urlprefix\url{https://www.jstatsoft.org/v082/i13}

\bibitem[{Laird et~al.(1987)Laird, Lange, and Stram}]{Laird1987}
Laird N, Lange N, Stram D (1987) Maximum likelihood computations with repeated
  measures: Application of the em algorithm. Journal of the American
  Statistical Association 82(397):97--105, \doi{10.1080/01621459.1987.10478395}

\bibitem[{Laird and Ware(1982)}]{Laird1982}
Laird NM, Ware JH (1982) Random-effects models for longitudinal data.
  Biometrics 38(4):963--974,
  \urlprefix\url{http://www.jstor.org/stable/2529876}

\bibitem[{Lawlor et~al.(2009)Lawlor, Lawlor, and Mishra}]{lawlor2009family}
Lawlor D, Lawlor D, Mishra G (2009) Family Matters: Designing, Analysing and
  Understanding Family Based Studies in Life Course Epidemiology. Life Course
  Approach to Adult Health, OUP Oxford

\bibitem[{Li et~al.(2019)Li, Guo, and Li}]{Li2019}
Li X, Guo N, Li Q (2019) Functional neuroimaging in the new era of big data.
  Genomics, Proteomics \& Bioinformatics 17(4):393 -- 401,
  \doi{https://doi.org/10.1016/j.gpb.2018.11.005}, big Data in Brain Science

\bibitem[{Lindstrom and Bates(1988)}]{Bates1988}
Lindstrom MJ, Bates DM (1988) Newton-raphson and em algorithms for linear
  mixed-effects models for repeated-measures data. Journal of the American
  Statistical Association 83(404):1014--1022,
  \urlprefix\url{http://www.jstor.org/stable/2290128}

\bibitem[{Magnus and Neudecker(1980)}]{Magnus1980}
Magnus JR, Neudecker H (1980) The elimination matrix: Some lemmas and
  applications. SIAM Journal on Algebraic Discrete Methods 1(4):422--449,
  \doi{10.1137/0601049}

\bibitem[{Magnus and Neudecker(1986)}]{Magnus1986:dup}
Magnus JR, Neudecker H (1986) Symmetry, 0-1 matrices, and jacobians : a review.
  Econometric Theory p~46

\bibitem[{Magnus and Neudecker(1999)}]{Magnus1999}
Magnus JR, Neudecker H (1999) Matrix differential calculus with applications in
  statistics and econometrics, rev. ed edn. Wiley Series in Probability and
  Statistics, John Wiley

\bibitem[{Neudecker and Wansbeek(1983)}]{Neudecker1983}
Neudecker H, Wansbeek T (1983) Some results on commutation matrices, with
  statistical applications. Canadian Journal of Statistics 11(3):221--231,
  \doi{10.2307/3314625}

\bibitem[{Patterson and Thompson(1971)}]{Patterson1971}
Patterson HD, Thompson R (1971) Recovery of inter-block information when block
  sizes are unequal. Biometrika 58(3):545--554,
  \urlprefix\url{http://www.jstor.org/stable/2334389}

\bibitem[{Pinheiro and Bates(2009)}]{Bates2009:book}
Pinheiro J, Bates D (2009) Mixed-Effects Models in S and S-PLUS. Statistics and
  Computing, Springer

\bibitem[{Pinheiro and Bates(1996)}]{Pinheiro1996}
Pinheiro JC, Bates DM (1996) Unconstrained parametrizations for
  variance-covariance matrices. Statistics and Computing 6:289--296

\bibitem[{Powell(2009)}]{Powell2009}
Powell M (2009) The bobyqa algorithm for bound constrained optimization without
  derivatives. Technical Report, Department of Applied Mathematics and
  Theoretical Physics

\bibitem[{Powell(1964)}]{Powell1964}
Powell MJD (1964) {An efficient method for finding the minimum of a function of
  several variables without calculating derivatives}. The Computer Journal
  7(2):155--162, \doi{10.1093/comjnl/7.2.155},
  \urlprefix\url{https://doi.org/10.1093/comjnl/7.2.155}

\bibitem[{Rao and Mitra(1972)}]{rao1972:psuedo}
Rao C, Mitra SK (1972) Generalized Inverse of Matrices and Its Applications.
  Probability and Statistics Series, Wiley

\bibitem[{Raudenbush and Bryk(2002)}]{Raudenbush2002}
Raudenbush SW, Bryk AS (2002) Hierarchical Linear Models: Applications and Data
  Analysis Methods, 2nd edn. Advanced Quantitative Techniques in the Social
  Sciences 1, SAGE Publications

\bibitem[{{SAS Institute Inc.}(2015)}]{SASMIXED}
{SAS Institute Inc} (2015) SAS/STAT® 14.1 User’s Guide The MIXED Procedure.
  Springer Berlin Heidelberg, Cary, NC: SAS Institute Inc.

\bibitem[{Satterthwaite(1946)}]{Satterthwaite1946}
Satterthwaite FE (1946) An approximate distribution of estimates of variance
  components. Biometrics Bulletin 2(6):110--114,
  \urlprefix\url{http://www.jstor.org/stable/3002019}

\bibitem[{Smith and Nichols(2018)}]{Smith2018}
Smith SM, Nichols TE (2018) Statistical challenges in “big data” human
  neuroimaging. Neuron 97(2):263 -- 268,
  \doi{https://doi.org/10.1016/j.neuron.2017.12.018}

\bibitem[{Tibaldi et~al.(2007)Tibaldi, Verbeke, Molenberghs, Renard, Van~den
  Noortgate, and de~Boeck}]{Tibaldi2007}
Tibaldi FS, Verbeke G, Molenberghs G, Renard D, Van~den Noortgate W, de~Boeck P
  (2007) Conditional mixed models with crossed random effects. British Journal
  of Mathematical and Statistical Psychology 60(2):351--365,
  \doi{10.1348/000711006X110562}

\bibitem[{Turkington(2013)}]{Turkington2013}
Turkington DA (2013) Generalized Vectorization, Cross-Products, and Matrix
  Calculus. Cambridge University Press, \doi{10.1017/CBO9781139424400}

\bibitem[{Turnbull et~al.(1999)Turnbull, Welsh, Heid, Davis, and
  Ratnofsky}]{Turnbull1999}
Turnbull BJ, Welsh ME, Heid CA, Davis W, Ratnofsky AC (1999) The longitudinal
  evaluation of school change and performance (lescp) in title i schools.
  interim report to congress.

\bibitem[{Van~Essen et~al.(2013)Van~Essen, Smith, Barch, Behrens, Yacoub, and
  Ugurbil}]{Essen2013}
Van~Essen D, Smith S, Barch D, Behrens T, Yacoub E, Ugurbil K (2013) The
  wu-minn human connectome project: an overview. NeuroImage 80,
  \doi{10.1016/j.neuroimage.2013.05.041}

\bibitem[{Verbeke and Molenberghs(2001)}]{verbeke2001:book}
Verbeke G, Molenberghs G (2001) Linear Mixed Models for Longitudinal Data.
  Springer Series in Statistics, Springer New York

\bibitem[{Welch(1947)}]{Welch1947}
Welch BL (1947) The generalization of `student's' problem when several
  different population variances are involved. Biometrika 34(1/2):28--35,
  \urlprefix\url{http://www.jstor.org/stable/2332510}

\bibitem[{West et~al.(2014)West, Welch, and Galecki}]{west2014linear}
West B, Welch K, Galecki A (2014) Linear Mixed Models: A Practical Guide Using
  Statistical Software. CRC Press

\bibitem[{Winkler et~al.(2015)Winkler, Webster, Vidaurre, Nichols, and
  Smith}]{Winkler2015}
Winkler AM, Webster MA, Vidaurre D, Nichols TE, Smith SM (2015) Multi-level
  block permutation. NeuroImage 123:253 -- 268,
  \doi{https://doi.org/10.1016/j.neuroimage.2015.05.092}

\bibitem[{Wolfinger(1996)}]{Wolfinger1996}
Wolfinger R (1996) Heterogeneous variance: Covariance structures for repeated
  measures. Journal of Agricultural, Biological, and Environmental Statistics
  1:205, \doi{10.2307/1400366}

\bibitem[{Wolfinger et~al.(1994)Wolfinger, Tobias, and Sall}]{Wolfinger1994}
Wolfinger R, Tobias R, Sall J (1994) Computing gaussian likelihoods and their
  derivatives for general linear mixed models. Siam Journal on Scientific
  Computing 15, \doi{10.1137/0915079}

\bibitem[{Zhu and Wathen(2018)}]{zhu2018essential}
Zhu S, Wathen AJ (2018) Essential formulae for restricted maximum likelihood
  and its derivatives associated with the linear mixed models.
  \eprint{1805.05188}

\end{thebibliography}

\end{document}